\documentclass[aos]{imsart}

\RequirePackage{amsthm,amsmath,amsfonts,amssymb}
\RequirePackage[numbers,sort&compress]{natbib}
\RequirePackage[colorlinks,citecolor=blue,urlcolor=blue]{hyperref}
\RequirePackage{graphicx}
\startlocaldefs
\usepackage{booktabs}
\theoremstyle{plain}

\newtheorem{theorem}{Theorem}[section]
\newtheorem{lemma}[theorem]{Lemma}
\newtheorem{corollary}[theorem]{Corollary}
\newtheorem{proposition}{Proposition}[section]
\theoremstyle{definition}
\newtheorem{definition}{Definition}
\newtheorem{example}{Example}
\newtheorem{assumption}{Assumption}
\newtheorem{remark}{Remark}

\newcommand{\continuedexamplename}{}
\newtheorem*{continuedexampleinner}{\continuedexamplename}

\endlocaldefs

\begin{document}

\begin{frontmatter}
\title{Small area learning}
\runtitle{Small area learning}

\begin{aug}
%%%%%%%%%%%%%%%%%%%%%%%%%%%%%%%%%%%%%%%%%%%%%%%
%% Only one address is permitted per author. %%
%% Only division, organization and e-mail is %%
%% included in the address.                  %%
%% Additional information such as            %%
%% identifying the corresponding author must %%
%% be included in in the Acknowledgments     %%
%% section if necessary.                     %%
%% ORCID can be inserted by command:         %%
%% \orcid{0000-0000-0000-0000}               %%
%%%%%%%%%%%%%%%%%%%%%%%%%%%%%%%%%%%%%%%%%%%%%%%
\author[A]{\fnms{Katarzyna}~\snm{Reluga} \ead[label=e1]{katarzyna.reluga@hu-berlin.de}\orcid{0000-0001-9651-5030}}
\author[B]{\fnms{Nicola}~\snm{Salvati}\ead[label=e2]{nicola.salvati@unipi.it}\orcid{0000-0002-4160-9387}}
\author[C]{\fnms{Setareh}~\snm{Ranjbar}\ead[label=e3]{setareh.ranjbar@chuv.ch}\orcid{0000-0002-0238-2976}}
\author[D]{\fnms{Nikos}~\snm{Tzavidis}\ead[label=e4]{n.tzavidis@soton.ac.uk}\orcid{0000-0002-8413-8095}}
\author[E]{\fnms{Mark}~\snm{van der Laan}\ead[label=e5]{laan@berkeley.edu}}
%%%%%%%%%%%%%%%%%%%%%%%%%%%%%%%%%%%%%%%%%%%%%%
%% Addresses                                %%
%%%%%%%%%%%%%%%%%%%%%%%%%%%%%%%%%%%%%%%%%%%%%%
\address[A]{Chair of AI in Business and Economics, Humboldt University of Berlin, Germany\printead[presep={ ,\ }]{e1}}

\address[B]{Department of Economics and Management, 
University of Pisa, Italy\printead[presep={,\ }]{e2}}

\address[C]{Department of Psychiatry \& Lausanne University Hospital,
University of Lausanne, Switzerland, \printead[presep={,\ }]{e3}}

\address[D]{Department of Social Statistics and Demography \& Southampton Statistical Sciences Research Institute,
University of Southampton, UK\printead[presep={,\ }]{e4}}

\address[E]{School of Public Health, University of California, Berkeley, USA\printead[presep={,\ }]{e5}}
\end{aug}

\begin{abstract}
Small area estimation (SAE) traditionally borrows strength across domains through a prespecified working model. We introduce
\emph{small area learning} (SAL), a novel learner-agnostic framework that formulates SAE as a multi-target
semi-supervised learning problem and treats outcome prediction as a nuisance task. Areas may differ through unrestricted covariate shift and structured conditional shift, with information shared through common and area-specific components. We develop estimators with and without cross-fitting, accommodating low-complexity and flexible learners and unifying synthetic, model-based, and semi-supervised estimators. For area means, we establish a unified asymptotic theory for these estimator families, characterizing conditions for consistency, asymptotic normality and providing variance estimators and confidence intervals without prescribing a particular learning algorithm. We extend the theory to settings with only aggregated population covariates and characterize the resulting aggregation bias. SAL reveals clearly a bias--variance trade-off: partial residual correction in classical estimators can reduce mean squared error, whereas full correction yields estimators robust to the misspecification of working models, potentially at higher variance. Simulations yield RMSE reductions of 34\%--49\% relative to direct estimation, with the best configuration varying across settings. The numerical studies also illustrate how shrinkage can produce overly optimistic confidence intervals at the expense of strong modeling assumptions. 
%\textcolor{red}{An empirical application illustrates how the framework supports the comparison of learners and estimation strategies within a common inferential formulation.}
%\textcolor{magenta}{Simulations yield RMSE reductions of 34\%--49\% relative to direct estimation, with the best configuration varying across settings. An empirical application favours a cross-fitted estimator using a flexible learner under the chosen validation criterion. The numerical studies also illustrate how shrinkage can produce overly optimistic confidence intervals at the expense of strong modeling assumptions. }
%The abstract should summarize the contents of the paper.It should be clear, descriptive, self-explanatory and not longerthan 200 words. It should also be suitable for publication in abstracting services. Formulas should be used as sparingly as possible within the abstract. The abstract should not make reference to results, bibliography or formulas in the body of the paper---it should be self-contained. This is a sample input file.  Comparing it with the output it generates can show you how to produce a simple document of your own.
\end{abstract}



\end{frontmatter}
%%%%%%%%%%%%%%%%%%%%%%%%%%%%%%%%%%%%%%%%%%%%%%
%% Please use \tableofcontents for articles %%
%% with 50 pages and more                   %%
%%%%%%%%%%%%%%%%%%%%%%%%%%%%%%%%%%%%%%%%%%%%%%
%\tableofcontents

%SAE
%Predictive theory
%Model-based
%Machine Learning
%Why this paper?
%How do we do it?

\section{Introduction}
\subsection{Motivation}\label{sec:motivation}

Reliable statistics at disaggregated geographical and socio-demographic levels
are essential for policy design and resource allocation, yet such domains often have small sample sizes. Direct survey estimators based only on
domain-specific observations, such as those of \cite{horvitz} and
\cite{hajek1958theory} (see Section~\ref{sec:stat_setup}), may therefore have high variability. Small area estimation (SAE) addresses this problem by borrowing information across domains to improve the precision of domain-specific estimates.

The main goal of SAE is to estimate target parameters such as means, medians, totals, and poverty indicators for individual areas or subpopulations. To overcome limited within-area sample sizes, survey statisticians have traditionally modelled the relationship between the variable of interest and auxiliary covariates, drawing on prediction theory for finite populations
\citep{cassel1977foundations,valliant2000finite,bolfarine2012prediction}. The observed population values are viewed as realizations from a superpopulation, or working model, which is then used to predict nonsampled outcomes and estimate the finite-population target.

Since the true relationship between the response and the covariates may be complex, simple working models, such as the linear regression models underlying generalized regression estimators \citep{deville1992calibration}, may provide only approximations to the data-generating process. To account for heterogeneity across areas, model-based small area estimation (SAE) has relied extensively on linear mixed models, generalized linear mixed models, and their Bayesian counterparts. Foundational examples include the area-level model of \cite{fay1979estimates} and the unit-level model of \cite{battese1988error}, both of which remain central to model-based SAE.

%Since the true data-generating process may involve complex relationships between covariates and stochastic errors, it is a generally well understood that simple working models such as linear regression model used in generalized regression estimators  \citep{deville1992calibration} can provide only an approximation of a true relationship. To capture area level population-level heterogeneity, much of SAE methodology has relied on linear mixed models, generalized linear mixed models, and their Bayesian counterparts. Seminal examples are the area-level model of \cite{fay1979estimates} and the unit-level model of \cite{battese1988error}, which remain central to model-based SAE.

Under model-based approaches, prediction and inference for the target parameter, including mean squared error estimation and confidence interval construction, typically depend on the adequacy of the working model. If important features of the conditional mean or error structure are misspecified, theoretical guarantees and conclusions from numerical studies may no longer hold
\citep[see, for example, the discussion in][]{jiang2020robust}. To reduce sensitivity to model misspecification, the SAE literature has considered
relaxing model assumptions \citep{kubo2017,diallo2018small}, introducing data transformations \citep{sugasawa2017transforming,RojasEtal2020}, benchmarking
\citep{pfeffermann2013new}, and robust
\citep{chambers2006m,sinha2009robust} or semiparametric methods
\citep{opsomer2008splines}.

In many SAE applications, the working model is primarily a device for predicting the nonsampled part of the finite population rather than an object of substantive interest in itself. Our central premise is therefore that the emphasis should shift from selecting a particular parametric model to selecting an estimator that performs well for the target parameter. This perspective is familiar in semiparametric statistics and semi-supervised inference (see Section~\ref{sec:literature_review} for a literature review), where an outcome regression used to construct the target estimator is termed a nuisance function. Here, ``nuisance'' describes its supporting role; its estimation remains crucial to the accuracy and validity of inference for the target. Once the target and the assessment criterion are specified, nuisance learning can be carried out by any procedure, from linear regression to flexible machine learning, that satisfies the conditions required for valid inference. The relevant question is then not which working model is ``correct'', but which learner, possibly selected data-adaptively, yields the best performance of the resulting target estimator.

This perspective naturally opens SAE to flexible learning methods, which have demonstrated strong predictive performance across a wide range of applications. It also motivates our notion of \emph{small area learning}: rather than attaching SAE to a particular regression, mixed model or a machine learning algorithm, we formulate the estimation problem first and allow the predictive
component to be supplied by a broad class of learners. The inferential problem then becomes one of identifying conditions under which such learners can be used without compromising valid area-level inference. As we show below, these conditions can be formulated either through low-complexity restrictions or through cross-fitting.

\subsection{Literature review}\label{sec:literature_review}

We outline four strands of the literature relevant to our work. 

\emph{Machine learning for SAE}.
%\subsubsection{Machine learning for SAE}\label{sec:lit_ML_SAE}
Recent work in SAE has explored replacing or augmenting classical mixed models with flexible learning methods, including tree-based approaches \citep{bilton2017classification,trees_Peru,merf_aos,zadlo2025integration,frink2026small} and gradient boosting \citep{messer2026gradient}. A related strand uses covariate-powered
empirical Bayes methods to combine noisy direct estimates with flexible covariate-based predictions, generalizing the shrinkage logic of the Fay--Herriot model \citep{ignatiadis2019covariate}.
These contributions typically develop particular machine-learning or shrinkage procedures for small-area prediction, often together with methods for uncertainty assessment. Our perspective is reversed: rather than starting with a particular learning algorithm and adapting it to SAE, we formulate SAE itself as a multi-target semi-supervised learning problem. The prediction rule is then a nuisance component that may be chosen data-adaptively provided that the resulting learner satisfies the conditions required for inference. %see  either under low-complexity restrictions or with cross-fitting.

\emph{Machine learning in design-based inference}
%\subsubsection{Machine learning in design-based inference}\label{sec:ML_design-based} 
A substantial design-based literature extends the generalized-difference/GREG
principle to flexible learners, including neural networks
\citep{montanari2005nonparametric} and tree-based methods
\citep{toth2011building,Mehdi,McConville}. More recently,
\cite{an2026agnostic} developed learner-agnostic,
cross-fitted model-assisted methods under complex sampling designs. Their framework accommodates a broader class of sampling designs, whereas the SAL framework developed in this paper addresses a distinct inferential problem: superpopulation inference for multiple area-specific targets while explicitly accounting for heterogeneity across areas. %Their framework is substantially more general than ours with respect to the samplingdesign, whereas SAL presented in our manuscript addresses a different problem: superpopulation inference for multiple area-specific targets with explicit cross-area heterogeneity.

%\subsubsection{Semi-supervised, transfer, and multi-task learning}
%\label{sec:lit_ssl}

\emph{Semi-supervised, transfer, and multi-task learning}. Our formulation is also related to several learning frameworks. In particular, \cite{zhang_high-dimensional_2022} study semi-supervised estimation of a population mean when outcomes are observed for a small labeled sample but covariates are available more broadly, and develop cross-fitted inference under flexible working models. While this setting motivates aspects of our construction, SAL addresses inference for multiple area-specific targets, explicitly formalizes borrowing strength across areas, and allows within-area sampling fractions to vanish. More broadly, information sharing across areas connects SAE to transfer learning, covariate shift, and multi-task learning \citep{shimodaira2000covariateshift,Yang_Zhang_Dai_Pan_2020,caruana1997multitask,baxter2000model}. Our emphasis, however, is not on developing a particular learning or transfer algorithm, but on establishing inferential properties for a broad class of admissible learners.

%Our formulation also connects to several learning frameworks. In particular, \cite{zhang_high-dimensional_2022} studies semi-supervised estimation of a population mean when outcomes are observed for a small labelled sample but covariates are available much more broadly, and develops cross-fitted inference under flexible working models. This setting motivates part of our construction, but SAL instead considers multiple area-specific targets, explicitly formalizes borrowing strength across them, accommodates vanishing within-area sampling fractions, and allows a class of admissible learners rather than a single semi-supervised mean problem. More broadly, the sharing of information across related areas linksSAE to transfer learning, covariate shift, and multi-task learning\citep{shimodaira2000covariateshift,Yang_Zhang_Dai_Pan_2020,caruana1997multitask,baxter2000model}. Our emphasis, however, is not on aparticular transfer algorithm but on formulating the SAE problem and itsinferential properties while leaving the learner largely unrestricted.

\emph{Generalized-difference and debiased estimation}.
%\subsubsection{Generalized-difference and debiased estimation}\label{sec:lit_review}
For mean estimation, SAL is also related to classical model-assisted survey estimation, particularly the generalized-difference principle underlying generalized regression (GREG) estimators \citep[see, e.g.,][]{deville1992calibration}. Under the corresponding sampling design, the fully corrected SAL estimator takes the form of a prediction-based mean augmented by a correction based on sample residuals. Other members of the SAL family differ in the weight assigned to this correction. Related residual-correction and orthogonalization ideas arise in augmented inverse-probability weighting \citep{robins1994estimation} and debiased or double machine learning \citep{chernozhukov2018}. These connections concern the structure of the estimating equations, rather than the inferential targets and frameworks considered here.

%For mean estimation, SAL also connects to the classical model-assisted generalized-difference estimator \citep[GREG, see for example][]{cassel1976some}. Under the corresponding sampling design, the fully corrected SAL estimator adds a correction based on sample residuals to the prediction-based mean estimate, following the generalized-difference principle. Other members of the SAL family vary the weight assigned to this residual correction. Related residual-correction and orthogonalization ideas appear in augmented inverse-probability weighting \citep{robins1994estimation} and debiased or double machine learning \citep{chernozhukov2018}. These connections concern the form of the estimating equation; the targets and inferential frameworks considered here are different.

\subsection{Our contribution}\label{sec:contribution}
We introduce \textit{small area learning} (SAL), a novel learner-agnostic framework for small area
estimation that separates the structure used to borrow information across areas from the procedure used to estimate the outcome mechanism. The resulting formulation brings cross-area heterogeneity, pooled outcome learning, and area-specific inference into a common statistical framework. Areas may differ through unrestricted covariate shift and structured conditional shift, while borrowing strength is achieved through a shared component learned from the pooled sample together with constrained area-specific deviations. The predictive component may be estimated parametrically, semiparametrically, or by flexible learning methods and is treated as a nuisance quantity for inference.

Our contributions are as follows. First, we formulate unit-level SAE as a multi-target semi-supervised learning problem and formalize borrowing strength through shared and area-specific components. Second, we introduce standard and cross-fitted SAL families that contain, as particular cases, synthetic, model-based SAE, and fully residual-corrected estimators, thereby making their bias--variance trade-off explicit. Third, for area means, under low-complexity or cross-fitting conditions, we establish a unified asymptotic theory for the SAL families, characterizing conditions for area-specific consistency and
asymptotic normality and providing variance estimators and confidence intervals. This theory makes explicit how inferential properties depend on the learner's area-specific approximation and estimation errors, without prescribing a particular prediction algorithm.
Fourth, we extend SAL to settings in which only aggregated population covariates are available and characterize the additional bias induced by replacing unit-level population information by area-level summaries. Fifth, we show how classical SAE constructions, including nested-error and area-level shrinkage models, fit within the structured conditional-shift framework and clarify their connections with direct, generalized-difference, and shrinkage estimators.

Our simulations assess when borrowing strength translates into improved estimation. Across all simulation settings considered, the best-performing SAL configurations substantially
reduce RMSE relative to direct estimation, with the preferred learner depending on the setting and the degree of residual correction. The magnitude of the gain depends strongly on the predictive power of the covariates and on the underlying data-generating setting. The empirical application further illustrates data-driven selection among candidate learners, the practical flexibility of SAL and illustrates the role of shrinkage.

The remainder of the paper is organised as follows.
Section~\ref{sec:stat_setup} introduces the superpopulation and
finite-population targets and the sampling-design assumptions.
Section~\ref{sec:foundations} presents the structured conditional-shift
formulation and illustrates how classical and machine-learning SAE methods
can be embedded within it. Section~\ref{sec:salearners} develops the proposed
small area learners and their inferential theory. Section~\ref{sec:sal_agg} considers settings in which only survey data and aggregated auxiliary information is available. Section~\ref{sec:discussion} discusses interpretation and limitations, Section~\ref{sec:simulations} presents simulation results, and
Section~\ref{sec:conclusions} concludes. The Supplementary Material
\citep{RelugaSAL_SM} contains proofs and inference under aggregation (Section~1), further discussion (Section~2), description of the practical method of learner selection (Section~3), simulation mechanisms and design (Sections~4--5), and an empirical study of income levels across Italian provinces (Section~6).

\section{Statistical setup}
\label{sec:stat_setup}

\subsection{Superpopulation framework}
\label{sec:super_pop}

Let $D$ denote the number of areas and
$\mathcal D:=\{1,\ldots,D\}$, with $d\in\mathcal D$ indexing areas and $j$
indexing units within area $d$. We use ``area'' to denote any group in a partition of the population, such as a subpopulation, cluster, or domain. Following the superpopulation framework
commonly used in survey statistics
\citep{cassel1977foundations,valliant2000finite}, let
$(Y_{dj},\mathbf X_{dj})\in\mathcal Y\times\mathcal X$ be composed of an
outcome (target) variable $Y_{dj}$ and covariate (feature) vector
$\mathbf X_{dj}$. We decompose the covariates as
\begin{equation}
\label{eq:X_decomp}
\mathbf X_{dj}=(\mathbf Z_{dj}^{\top},G_{dj})^{\top},
\qquad
\mathbf Z_{dj}=(\mathbf W_{dj}^{\top},\mathbf U_{dj}^{\top})^{\top},
\end{equation}
where $\mathbf W_{dj}\in\mathcal W\subseteq\mathbb R^p$ contains unit-level
covariates, $\mathbf U_{dj}\in\mathcal U\subseteq\mathbb R^q$ contains
area-level covariates, often referred to as contextual covariates, and $G_{dj}\in\mathcal D$ is the area indicator, so that
$G_{dj}=d$. Contextual covariates are constant across units within an area:
$\mathbf U_{dj}=\mathbf u_d$ for every $j$, so that, conditionally on $G=d$,
the $\mathbf U$-marginal of $P_d^{\mathbf Z}$ is degenerate at
$\mathbf u_d$. We write $\mathcal Z=\mathcal W\times\mathcal U$ and
$\mathcal X=\mathcal Z\times\mathcal D$. Throughout the manuscript, uppercase
letters, i.e., $Y_{dj},\mathbf X_{dj}$, denote random quantities, whereas
their lowercase equivalents, i.e., $y_{dj},\mathbf x_{dj}$, refer to their
realizations. For every $d\in\mathcal D$, we define the area-specific law
\begin{equation}
\label{eq:area_law}
P_d
:=
\mathcal L\{(Y,\mathbf Z)\mid G=d\}
=
P_d^{Y\mid\mathbf Z}\,P_d^{\mathbf Z},
\;
P_d^{\mathbf Z}:=\mathcal L(\mathbf Z\mid G=d),
\;
P_d^{Y\mid\mathbf Z}:=\mathcal L(Y\mid\mathbf Z,G=d),
\end{equation}
with $P_d^{Y\mid\mathbf Z}\in\mathcal P_d^{Y\mid\mathbf Z}$ and
$P_d^{\mathbf Z}\in\mathcal P_d^{\mathbf Z}$ for conditional and marginal
modelling classes $\mathcal P_d^{Y\mid\mathbf Z}$ and
$\mathcal P_d^{\mathbf Z}$, the joint modelling class being
\begin{equation*}
\mathcal P_d
=
\left\{
P^{Y\mid\mathbf Z}P^{\mathbf Z}:
P^{Y\mid\mathbf Z}\in\mathcal P_d^{Y\mid\mathbf Z},\;
P^{\mathbf Z}\in\mathcal P_d^{\mathbf Z}
\right\}
\subseteq
\mathcal M_1(\mathcal Y\times\mathcal Z),
\qquad
\mathcal P=\bigcup_{d\in\mathcal D}\mathcal P_d,
\end{equation*}
where $\mathcal M_1(\mathcal Y\times\mathcal Z)$ denotes the set of
probability measures on $\mathcal Y\times\mathcal Z$. Intuitively, the factorization in \eqref{eq:area_law} separates two possible
sources of between-area heterogeneity: differences in the marginal
distribution of the covariates $P_d^{\mathbf Z}$ and differences in the
conditional distribution of the outcome given the covariates
$P_d^{Y\mid\mathbf Z}$. This distinction will be central to the structured conditional-shift formulation developed in
Section~\ref{sec:foundations}. In classical unit-level SAE models, covariates are typically treated as fixed and $P_d$ is understood as the conditional law given the realized random effects. Such models are accommodated by our framework; see Sections~\ref{sec:nerm}-\ref{sec:merf} for worked examples.

\subsection{Superpopulation target parameters }
\label{sec:superpopulation_targets}

The superpopulation target parameter associated with area $d$ is defined
through a Euclidean-valued functional
\begin{equation}
\label{eq:functional_map}
\theta_d
=
\Theta(P_d)
=
\Theta(P_d^{Y\mid\mathbf Z}\,P_d^{\mathbf Z}),
\qquad
\Theta:\mathcal P\rightarrow\mathcal T\subseteq\mathbb R.
\end{equation}
Depending on the choice of $\Theta$, the superpopulation target may be a mean,
quantile, or another finite-dimensional functional, i.e., SAL framework is not limited to one particular estimand. For the sake of clarity, in
Section~\ref{sec:salearners}  we illustrate the learning pipeline by developing estimation strategies and
establishing asymptotic theory for the mean functional
$\Theta_{\mathrm{mean}}$ in \eqref{eq:mean_functional}. To this end, let
$\mathbf z\in\mathcal Z$ and define the area regression function as
\begin{equation}
\label{eq:area_regression}
m_d(\mathbf z)
:=
E(Y_{dj}\mid\mathbf Z_{dj}=\mathbf z,G_{dj}=d)
=
\int_{\mathcal Y}y\,dP_d^{Y\mid\mathbf Z=\mathbf z}(y).
\end{equation}
Assume that $E_{P_d}|Y_{dj}|<\infty$. Then the area-level mean functional
$\theta_d=E_{P_d}(Y)$ is 
\begin{equation}
\label{eq:mean_functional}
\theta_d
=
\Theta_{\mathrm{mean}}(P_d)
=
\int_{\mathcal Y\times\mathcal Z}y\,dP_d(y,\mathbf z)
=
\int_{\mathcal Z}m_d(\mathbf z)\,dP_d^{\mathbf Z}(\mathbf z)
=
E_{P_d^{\mathbf Z}}\{m_d(\mathbf Z)\}
=:
\mu_d,
\end{equation}
where the second equality follows by the definition
of the mean functional, the third by the tower
property and $m_d(\mathbf{z})$ in
\eqref{eq:area_regression}, and the fourth by the definition
of expected value.

\subsection{Finite-population framework and target parameters}
\label{sec:finite_population}

Within the finite-population setup, area $d$ contains $N_d$ units,
indexed by $\mathcal Q_d:=\{1,\ldots,N_d\}$. Unit indices are understood to carry their area labels;
thus, unions of index sets across areas are disjoint unions. We set $\mathcal Q=\bigcup_{d\in\mathcal D}\mathcal Q_d$ and
$N=\sum_{d\in\mathcal D}N_d$.
%\begin{equation*}
%\mathcal Q=\bigcup_{d\in\mathcal D}\mathcal Q_d,
%\qquad
%N=\sum_{d\in\mathcal D}N_d.
%\end{equation*}
The law $P_d$ in \eqref{eq:area_law} describes how the values of these units
are generated; the sampling design determines which of them are observed.
Let $I_{dj}\in\{0,1\}$ denote the sample membership indicator and $p_{dj} = \Pr(I_{dj}=1\mid Y_{dj},\mathbf Z_{dj},G_{dj}=d)$
the first-order inclusion probability.

\begin{assumption}[Sampling design]
\label{assump:design}
For every $d\in\mathcal D$ and $j\in\mathcal Q_d$: (i) the design is
noninformative given the covariates,
$I_{dj}\perp\!\!\!\perp Y_{dj}\mid(\mathbf Z_{dj},G_{dj})$; (ii) the design is
self-weighting within areas, $p_{dj}=p_d$ for $P_d$-almost every
$(Y_{dj},\mathbf Z_{dj})$, with $p_d>0$; and (iii) conditionally on the
realized area sample size $n_d$, the sampled pairs
$(Y_{dj},\mathbf Z_{dj})_{j\in\mathcal Q_d^s}$ are independent and
identically distributed draws from $P_d$.
\end{assumption}

Under Assumption~\ref{assump:design}, $\mathcal L\{(Y,\mathbf Z)\mid G=d,I=1\}=P_d$: the sampled units of area $d$ are draws from the area law itself, the property exploited by the learners of Section~\ref{sec:salearners}. Parts
(i) and (ii) leave inclusion probabilities free to vary \emph{across} areas, permitting disproportionate allocation, whereas dependence of $p_{dj}$ on $\mathbf Z_{dj}$ \emph{within} an area is excluded. Self-weighting within areas is a standard simplifying condition in unit-level SAE theory and permits the residual correction to be formed as an unweighted area sample mean.
Part (iii) holds, for example, under simple random sampling within an area when the finite-population values are generated independently from $P_d$ and the selection of units is independent of those values. Other sampling designs are discussed in Section 2.3 of the Supplementary Material \citep{RelugaSAL_SM}. Writing
$\mathcal Q_d^s=\{j\in\mathcal Q_d:I_{dj}=1\}$ for the sampled units and
$\mathcal Q_d^r=\mathcal Q_d\setminus\mathcal Q_d^s$ for the nonsampled ones,
with $|\mathcal Q_d^s|=n_d$, we set $\mathcal Q^s=\bigcup_{d\in\mathcal D}\mathcal Q_d^s$,
$\mathcal Q^r=\bigcup_{d\in\mathcal D}\mathcal Q_d^r$,
$n=\sum_{d\in\mathcal D}n_d$. The finite-population mean of area $d$ is
\begin{equation}
\label{eq:target_decomposition}
\bar y_d
=
\frac{1}{N_d}
\left(
\sum_{j\in\mathcal Q_d^s}y_{dj}
+
\sum_{j\in\mathcal Q_d^r}y_{dj}
\right),
\end{equation}
where the first sum is composed of the observed (sampled) units and the second from unknown out-of-sample units. The decomposition in \eqref{eq:target_decomposition} motivates the learners of
Section~\ref{sec:salearners}. Assumption~\ref{assump:iid} provides the
superpopulation framework linking %the finite-population target 
$\bar y_d$ in \eqref{eq:target_decomposition}
to %the area mean 
$\mu_d$ in
\eqref{eq:mean_functional}.

\begin{assumption}[Superpopulation structure]
\label{assump:iid}
Let $\mathcal F_{\mathcal P}$ denote the
$\sigma$-field generated by the area-specific
components governing the area laws
$(P_d)_{d\in\mathcal D}$.
Conditionally on $\mathcal F_{\mathcal P}$, the
population pairs
$(Y_{dj},\mathbf Z_{dj})_{j\in\mathcal Q_d}$
are independent and identically distributed
according to $P_d$ within each area, and the
collections corresponding to distinct areas are
mutually independent. Equivalently, conditional
on $\mathcal F_{\mathcal P}$, the joint population
law factorises as
$\bigotimes_{d\in\mathcal D}P_d^{\otimes N_d}$.
\end{assumption}

The conditioning in Assumption~\ref{assump:iid} accommodates classical SAE models. For example, in the nested-error model of Section~\ref{sec:nerm}, Assumption~\ref{assump:iid} holds conditional on the realized area effect. Marginally, units within an area are exchangeable but dependent due to their shared random effect.
Accordingly, probabilistic statements involving the
area laws are understood conditionally on
$\mathcal F_{\mathcal P}$, whose area-specific
components are introduced formally in
Section~\ref{sec:foundations}. Let $\bar Y_d:=N_d^{-1}\sum_{j\in\mathcal Q_d}Y_{dj}$
denote the superpopulation-random counterpart of
the realized finite-population mean $\bar y_d$.
Under Assumption~\ref{assump:iid} and %the integrability condition 
$E_{P_d}|Y|<\infty$ stated above, $\bar Y_d\to\mu_d$ almost surely as $N_d\to\infty$. The two quantities serve distinct roles: $\mu_d$, a functional of the area law $P_d$, defines the population parameter underlying the learning problem and asymptotic theory, whereas $\bar y_d$ is the realized finite-population mean and the primary SAE target. Corollary~\ref{cor:finite_pop} shows that the asymptotic results for $\mu_d$ also extend to $\bar y_d$. 

\subsection{Direct estimation and the need for borrowing strength}
\label{sec:direct_estimation}

Let $p_{dj}$ denote the first-order inclusion probability of
Assumption~\ref{assump:design}, $w_{dj}=p_{dj}^{-1}$ the corresponding
sampling weight. Standard design-based estimators of $\bar y_d$ are the
H\'ajek and Horvitz--Thompson estimators
\citep{hajek1958theory,horvitz}:
\begin{equation}
\label{eq:Hajek_HT}
\widehat{\bar y}_{d}^{\,h}
=
\frac{1}{\widehat N_d^{\,h}}
\sum_{j\in\mathcal Q_d^s}w_{dj}y_{dj},
\qquad
\widehat{\bar y}_{d}^{\,ht}
=
\frac{1}{N_d}
\sum_{j\in\mathcal Q_d^s}w_{dj}y_{dj},
\qquad
\widehat N_d^{\,h}
=
\sum_{j\in\mathcal Q_d^s}w_{dj}.
\end{equation}

Under Assumption~\ref{assump:design}(ii), the weights are
constant within each area, so that
$\widehat{\bar y}_{d}^{\,h}$ reduces to the unweighted sample mean. These estimators provide natural design-based benchmarks for finite-population inference \citep{sarndal1992model}, but their precision depends on the area sample size $n_d=|\mathcal Q_d^s|$. In small area applications, these sample sizes are often too small to attain the desired precision, motivating methods that borrow strength across areas while retaining area-specific heterogeneity. Direct estimators play two roles in what follows: they provide design-based benchmarks, as in Section~\ref{sec:simulations} and, when aggregated to coarser domains, validation targets for learner selection in applied studies as in Section 6 of the Supplementary Material \citep{RelugaSAL_SM}.

\section{Foundations of small area learning}
\label{sec:foundations}

\subsection{SAE as structured conditional shift learning} This section introduces small area learning (SAL) as a general inferential
framework for estimating area-specific functionals. Recall from
Section~\ref{sec:super_pop} that the target
$\theta_d=\Theta(P_d)$ is a functional of the area law
$P_d=P_d^{Y\mid\mathbf Z}P_d^{\mathbf Z}$. The factorisation identifies exactly two channels through which areas, and hence the targets, may differ: the conditional component $P_d^{Y\mid\mathbf Z}$ and the marginal component $P_d^{\mathbf Z}$. In the transfer-learning literature, heterogeneity in the
former is termed \emph{conditional shift}
\citep{Yang_Zhang_Dai_Pan_2020}, whereas heterogeneity confined to the latter
is termed \emph{covariate shift}
\citep{shimodaira2000covariateshift}. Neither notion, taken alone, describes
small area estimation. Requiring
$P_d^{Y\mid\mathbf Z}=P^{Y\mid\mathbf Z}$ for all $d$ is too strong: it rules out any area-specific effect in the conditional law, in particular, the
random area effects that SAE exists to capture, and indeed excludes the classical unit-level model in Section~\ref{sec:nerm}. Allowing the $P_d^{Y\mid\mathbf Z}$ to differ arbitrarily is too weak. SAE emerged as a response to data scarcity at the area level: with $n_d$ small, the
area-specific conditional law $P_d^{Y\mid\mathbf Z}$ cannot be learnt from
area $d$'s data alone. The cornerstone of SAE is therefore to ``borrow
strength'' across areas while still accounting for differences between the
area-level laws.

To reconcile these two needs, we draw on the literature on multi-task
learning and domain adaptation with shared representations
\citep{caruana1997multitask,baxter2000model}. There, $D$ related prediction
tasks are learnt jointly by a model with a common body, trained on all tasks,
and task-specific heads kept deliberately small, low-dimensional,
penalised, or both, so that each task benefits from every other task's data
without being forced to share its characteristics. This description closely
resembles classical SAE methods, which we embed in this learning framework.

Before stating the main assumption, we define its ingredients. Let
$\mathcal H$ and $\Delta$ be parameter sets for a shared parameter
$\boldsymbol\eta\in\mathcal H$ and area-specific parameters
$\delta_d\in\Delta$, $d\in\mathcal D$; the sets may be finite-dimensional,
as in Section~\ref{sec:nerm}, or infinite-dimensional, e.g.\ classes of
functions. Furthermore, we define a Markov kernel from $\mathcal Z$ to
$\mathcal Y$ as a map
$K:\mathcal Z\times\mathcal B(\mathcal Y)\to[0,1]$
such that $K(\mathbf z,\cdot)$ is a probability measure on $\mathcal Y$ for
every $\mathbf z\in\mathcal Z$ and $K(\cdot,B)$ is measurable for every
$B\in\mathcal B(\mathcal Y)$, where $\mathcal B(\mathcal Y)$ denotes the
Borel $\sigma$-algebra on $\mathcal Y$. Observe that a conditional distribution
$P_d^{Y\mid\mathbf Z}$ is a Markov kernel, with $K(\mathbf z,B)
=
\Pr(Y\in B\mid\mathbf Z=\mathbf z,G=d)$.
\begin{assumption}[Structured conditional shift]
\label{assump:structured-shift}
Let $P_d=P_d^{Y\mid\mathbf Z}P_d^{\mathbf Z}$ be the area-specific law
defined in Section~\ref{sec:super_pop}, and let the marginal law
$P_d^{\mathbf Z}$ be left unrestricted. Let
$\{
\mathcal K_{\boldsymbol\eta,\delta}:
\boldsymbol\eta\in\mathcal H,\ \delta\in\Delta
\}$ 
be a family of sets of Markov kernels from $\mathcal Z$ to $\mathcal Y$,
indexed by the pair $(\boldsymbol\eta,\delta)$. There exist a shared
parameter $\boldsymbol\eta\in\mathcal H$ and area-specific parameters
$\delta_d\in\Delta$, $d\in\mathcal D$, such that
$P_d^{Y\mid\mathbf Z}\in
\mathcal K_{\boldsymbol\eta,\delta_d}$,
$d\in\mathcal D$. The vector of deviations $(\delta_1,\ldots,\delta_D)$ is jointly constrained
in a sense made precise in Assumption~\ref{assump:constraint}.
\end{assumption}

Intuitively, Assumption~\ref{assump:structured-shift} represents each area's
conditional law through a shared component, indexed by $\boldsymbol\eta$, and an area-specific component, indexed by $\delta_d$. The formulation allows
both complete specifications, where
$\mathcal K_{\boldsymbol\eta,\delta_d}$ is a singleton, for example under Gaussian specifications in Sections~\ref{sec:nerm}--\ref{sec:merf}, and partial specifications, such as estimating equations in Section~\ref{sec:mq}. The shared component can be estimated from the pooled sample of all $n$
observations, providing SAE's ``borrowing of strength''. The area-specific
component captures the deviation of area $d$ from the shared mechanism.
When deviations are jointly regularised, information from other areas can
stabilise their estimation, at the cost of shrinking them towards the common
mechanism and thereby attenuating between-area differences, as with the
random effects in Section~\ref{sec:nerm}; otherwise each $\delta_d$ must be
learnt from the $n_d$ observations of its own area. This distinction between
learning the shared component and stabilising area-specific deviations is
central to SAL. It should not be confused with the rate of estimation of the target $\mu_d$:
for fully residual-corrected estimators, the first-order error is governed by
within-area residual variation and has the $n_d^{-1/2}$ scale under the
assumptions of Section~\ref{sec:salearners}; see also Section~\ref{sec:discussion}.

Covariate shift \citep{shimodaira2000covariateshift} and conditional shift
\citep{Yang_Zhang_Dai_Pan_2020}, introduced at the beginning of this section,
are the boundary cases of Assumption~\ref{assump:structured-shift}. Pure
covariate shift, $P_d^{Y\mid\mathbf Z}=P^{Y\mid\mathbf Z}$ for every $d$,
requires not only that $\Delta$ be a singleton but also that the specification
determine the same kernel, or the same target-relevant features of it, in every area. Unrestricted conditional shift removes any borrowing of strength for the estimation of the deviations, although borrowing continues through the shared component $\boldsymbol\eta$ whenever the latter is identified separately from them. Classical SAE models occupy an intermediate case in
which the deviation is a low-dimensional random effect, often a scalar location shift, constrained through a common distribution. SAL also permits functional deviations, with their structure imposed through a distributional
assumption, regularisation, or a complexity restriction.

\begin{assumption}[Constrained deviations]
\label{assump:constraint}
There exist deviation functionals
$\rho_d:\Delta\to[0,\infty]$, $d\in\mathcal D$,
and finite constants $\{C_D\}_{D\ge1}$ with
$\sup_{D\ge1}C_D/D<\infty$ such that $\sum_{d\in\mathcal D}\rho_d^2(\delta_d)\leq C_D$ or  $E\left\{\sum_{d\in\mathcal D}\rho_d^2(\delta_d)\right\}\leq C_D$ for deterministic or random deviations, respectively. 
\end{assumption}

The functional $\rho_d$ measures the magnitude of the area-specific deviation
under the covariate law of area $d$ and may be chosen uniformly across areas.
A natural choice for functional deviations is
$\rho_d(\delta)=\|\delta\|_{L_2(P_d^{\mathbf Z})}$. Even though $\delta_d$ is
a function of the covariates, square-integrable under $P_d^{\mathbf Z}$, it
acts on the conditional law $P_d^{Y\mid\mathbf Z}$: for additive deviations
it shifts the regression function, $m_d=m_0+\delta_d$. The decomposition into
a shared component and area-specific deviations then requires a centring
convention. We use
$\sum_{d\in\mathcal D}\omega_d\delta_d=0$ for deterministic deviations and
$E(\delta_d)=0$ for random deviations, where $\omega_d=\Pr(G=d)$, and
$N_d/N\to\omega_d$ if the finite populations are generated from a mixture
with area probabilities $(\omega_d)_{d\in\mathcal D}$. In contrast, when
$\Delta$ is a bounded subset of $\mathbb R$ and $\delta_d$ indexes the family
$\{K_{\boldsymbol\eta,\delta}:\delta\in\Delta\}$, its value identifies the
area's position within that family and no centring convention is required. The scaling $\sup_D C_D/D<\infty$ bounds the average magnitude of the
deviations and permits classical random-effects models with nonvanishing
between-area variance $\tau^2$. In particular, if
$E\{\rho_d^2(\delta_d)\}=\tau^2$ for every $d$, then
Assumption~\ref{assump:constraint} holds with $C_D=D\tau^2$, whereas a bound
fixed in $D$ would force $\tau^2\to0$. The additive case covers
Sections~\ref{sec:nerm}--\ref{sec:merf}, where $\delta_d$ is a random
intercept and $\rho_d(\delta)=|\delta|$ is common across areas; the
bounded-index case covers the M-quantile model of Section~\ref{sec:mq}. 

Finally, Assumption~\ref{assump:constraint} restricts the data-generating structure,
not the learning method. The extent of borrowing depends on the learner,
penalty, or hierarchical model used, while learner-class complexity is
controlled separately by Assumption~\ref{assump:entropy_growth}; accordingly,
Assumption~\ref{assump:constraint} is not needed in the proofs of
Section~\ref{sec:salearners}, where learner consistency is imposed directly.

Under Assumption~\ref{assump:structured-shift}, the target
\eqref{eq:functional_map} takes the form
\begin{equation}
\label{eq:target_structured}
\theta_d
=
\Theta\!\left(
P_d^{Y\mid\mathbf Z}P_d^{\mathbf Z}
\right),
\qquad
P_d^{Y\mid\mathbf Z}
\in
\mathcal K_{\boldsymbol\eta,\delta_d},
\end{equation}
which exhibits the two sources of between-area heterogeneity admitted by the
framework: the area deviation $\delta_d$ and the covariate law
$P_d^{\mathbf Z}$. The shared parameter $\boldsymbol\eta$ and the functional
$\Theta$ are common to all areas. If $P_d^{Y\mid\mathbf Z}=K_{\boldsymbol\eta,\delta_d}$, as in the Gaussian examples in Sections~\ref{sec:nerm}--\ref{sec:merf}, \eqref{eq:target_structured} reduces to $\theta_d
=
\Theta\!\left(
K_{\boldsymbol\eta,\delta_d}P_d^{\mathbf Z}
\right)$, so that $(\boldsymbol\eta,\delta_d,P_d^{\mathbf Z})$ determines the target.
Under a partial specification as in Section~\ref{sec:mq}, $\theta_d$ is determined by these quantities
only when $\Theta$ depends on the conditional law through features fixed by
$\mathcal K_{\boldsymbol\eta,\delta_d}$.

For the mean functional $\mu_d = \Theta_{\mathrm{mean}}(P_d) =
E_{P_d^{\mathbf Z}}\{m_d(\mathbf Z)\}$ in \eqref{eq:mean_functional}, the
conditional mean is specified through a subclass
$\mathcal K^{\mathrm{mean}}_{\boldsymbol\eta,\delta} \subseteq
\mathcal K_{\boldsymbol\eta,\delta}$, constructed as follows.

\begin{definition}[Generic regression function]
\label{def:constraint_generic}
Let $\mathcal E_\lambda$ be a class of mean-zero
sub-Gaussian distributions with parameter $\lambda$.
For $\boldsymbol\eta=(m_0,\lambda)$ and
$\delta\in\Delta\subseteq
\bigcap_{d\in\mathcal D}L_2(P_d^{\mathbf Z})$,
where $\Delta$ is a class of measurable functions, define
\[
\mathcal K^{\mathrm{mean}}_{\boldsymbol\eta,\delta}
=
\left\{
K:
\mathcal L_{K(\mathbf z,\cdot)}
\{Y-m_0(\mathbf z)-\delta(\mathbf z)\}
\in\mathcal E_\lambda
\ \text{for $P_d^{\mathbf Z}$-a.e.\ $\mathbf z$
and every $d$}
\right\}.
\]

For area-specific deviations $\delta_d\in\Delta$, the
corresponding regression function is
\begin{equation}
\label{eq:generic_m}
m_d(\mathbf z)
=
m_0(\mathbf z)+\delta_d(\mathbf z)
=
m(\mathbf z,d)
=
m(\mathbf x),
\qquad
\mathbf x=(\mathbf z^\top,d)^\top .
\end{equation}
\end{definition}
Since the distributions in $\mathcal E_\lambda$ are
mean zero, every
$K\in\mathcal K^{\mathrm{mean}}_{\boldsymbol\eta,\delta}$
has conditional mean
$m_0(\mathbf z)+\delta(\mathbf z)$. 
For functional deviations, a natural choice in
Assumption~\ref{assump:constraint} is
$\rho_d(\delta_d)
=\|\delta_d\|_{L_2(P_d^{\mathbf Z})}$,
which measures the deviation under the covariate law of
the corresponding area. Non-additive deviations in the
conditional mean can be represented as
$m(\mathbf z,d)
=f\{m_0(\mathbf z),\delta_d(\mathbf z)\}$
for a known link $f$, with
\eqref{eq:generic_m} as the leading additive case.

\subsection{Worked examples}
\label{sec:running_examples}
In this section we show that some of the most common models used in SAE are special cases of SAL.
\subsubsection{Nested-error regression model (NERM)}
\label{sec:nerm}

Let $\mathbf Z_{dj}=\mathbf W_{dj}$. The unit-level model of
\cite{battese1988error} posits
$Y_{dj}=\mathbf W_{dj}^{\top}\boldsymbol\zeta+\gamma_d+\varepsilon_{dj}$,
where $\gamma_d\sim N(0,\sigma_\gamma^2)$ is an area-specific random effect
and $\varepsilon_{dj}\sim N(0,\sigma_\varepsilon^2)$ is a unit-level error.
The area effects are independent across areas, the unit errors across units,
and the two collections are mutually independent. NERM satisfies
Assumption~\ref{assump:structured-shift} with shared component
$\boldsymbol\eta=(\boldsymbol\zeta,\sigma_\varepsilon^2)$,
area-specific component $\delta_d=\gamma_d$, and $\Delta=\mathbb R$.
Conditionally on the realized $\gamma_d$,
$P_d^{Y\mid\mathbf W=\mathbf w}
=
K_{\boldsymbol\eta,\delta_d}(\cdot\mid\mathbf w)
=
N\!\left(
\mathbf w^\top\boldsymbol\zeta+\gamma_d,
\sigma_\varepsilon^2
\right)$, so $\mathcal K_{\boldsymbol\eta,\delta_d}
=\{K_{\boldsymbol\eta,\delta_d}\}$, while $P_d^{\mathbf Z}$ is unrestricted.
In classical SAE the covariates are often treated as fixed. The area-specific components are constrained through
$\gamma_d\sim N(0,\sigma_\gamma^2)$. With
$\rho_d(\delta)=|\delta|$, $\sum_{d=1}^D E\{\rho_d^2(\gamma_d)\}
=
D\sigma_\gamma^2$, so Assumption~\ref{assump:constraint} holds with
$C_D=D\sigma_\gamma^2$. Under Gaussianity, estimation of the area effects is
equivalent to $\ell_2$-penalised least squares, with penalty strength
proportional to $\sigma_\varepsilon^2/\sigma_\gamma^2$. Finally,
$\mu_d^{ner}
=
\bar{\mathbf w}_d^\top\boldsymbol\zeta+\gamma_d$,
where
$\bar{\mathbf w}_d
=
\int_{\mathcal W}\mathbf w\,dP_d^{\mathbf W}(\mathbf w)$.
This is the conditional area target, and classical BLUP and EBLUP predict it
when the variance components are known and estimated, respectively. With
known parameters, the best linear predictor of $\gamma_d$ based on
$P_{n_d}(Y-\mathbf W^\top\boldsymbol\zeta)$ has shrinkage factor
$\Gamma_d^{ner}
=
\sigma_\gamma^2/
(\sigma_\gamma^2+\sigma_\varepsilon^2/n_d)$.
Thus, from a learning perspective, NERM is a multi-task learner with a linear
shared component and scalar area-specific intercepts shrunk through an
$\ell_2$ penalty, equivalently a Gaussian random-effects distribution.

\subsubsection{Fay--Herriot model}
\label{sec:fh}

At the area level, take one observation per area with
$\mathbf Z_{dj}=\mathbf U_d$ and $Y_{dj}=Y_d$, where
$Y_d=\widehat{\bar y}_d^{(\cdot)}$ is either the H\'ajek or
Horvitz--Thompson estimator in \eqref{eq:Hajek_HT}. The model of
\cite{fay1979estimates} posits 
$Y_d=\mathbf u_d^\top\boldsymbol\zeta+\gamma_d+e_d$, where $\gamma_d\sim N(0,A)$ is an area effect and
$e_d\sim N(0,V_d)$ is a sampling error with known $V_d$. The area effects
and sampling errors are independent across areas and mutually independent. Conditionally on $\gamma_d$, Fay--Herriot satisfies Assumption~\ref{assump:structured-shift} with shared component
$\boldsymbol\eta=\boldsymbol\zeta$, area-specific component $\delta_d=\gamma_d$, $\Delta=\mathbb R$, and
$P_d^{Y\mid\mathbf U=\mathbf u}
=
K_{\boldsymbol\eta,\delta_d}(\cdot\mid\mathbf u)
=
N\!\left(
\mathbf u^\top\boldsymbol\zeta+\gamma_d,V_d
\right)$, where $(V_1,\ldots,V_D)$ enter the kernel family as known constants.
Since $\mathbf Z_{dj}=\mathbf U_d$, the covariate law is degenerate,
$P_d^{\mathbf Z}=\delta_{\mathbf u_d}$. The deviations are constrained as
in Section~\ref{sec:nerm}, with $C_D=DA$. The corresponding area target is $\mu_d^{fh}=\mathbf u_d^\top\boldsymbol\zeta+\gamma_d$. As in NERM, BLUP and EBLUP predict this latent target from the noisy direct estimator $Y_d$. With known model parameters, the best predictor is
$\mathbf u_d^\top\boldsymbol\zeta
+
\Gamma_d^{fh}
(Y_d-\mathbf u_d^\top\boldsymbol\zeta)$,
$\Gamma_d^{fh}=\frac{A}{A+V_d}$. Thus Fay--Herriot has the same shared-plus-shrunk-area-effect structure as
NERM, but operates directly at the area level.

\subsubsection{Mixed-effects random forest}
\label{sec:merf}

Let $\mathbf Z_{dj}=(\mathbf W_{dj}^\top,\mathbf U_{dj}^\top)^\top$.
\cite{merf_aos} extend the nested-error model of
Section~\ref{sec:nerm} by replacing its linear shared component with
$Y_{dj}=m_0(\mathbf Z_{dj})+\gamma_d+\varepsilon_{dj}$, where $m_0$ is an
unknown, potentially nonlinear and high-dimensional function,
$\gamma_d\sim N(0,\sigma_\gamma^2)$, and
$\varepsilon_{dj}\sim N(0,\sigma_\varepsilon^2)$, with the same independence
structure as in NERM. Estimation proceeds iteratively by fitting a
mixed-effects random forest and updating the area effects and variance components \citep[MERF, see][for details]{Hajjem03062014}. Conditionally on $\gamma_d$,
$P_d^{Y\mid\mathbf Z=\mathbf z}
=K_{\boldsymbol\eta,\delta_d}(\cdot\mid\mathbf z)
=N\!\left(m_0(\mathbf z)+\gamma_d,\sigma_\varepsilon^2\right)$,
with shared component
$\boldsymbol\eta=(m_0,\sigma_\varepsilon^2)$ and area-specific component
$\delta_d=\gamma_d$, constrained through the Gaussian distribution as in
Sections~\ref{sec:nerm}--\ref{sec:fh}; $P_d^{\mathbf Z}$ remains
unrestricted. The corresponding mean functional is
$\mu_d^{merf}
=\int_{\mathcal Z}m_0(\mathbf z)\,dP_d^{\mathbf Z}(\mathbf z)+\gamma_d$.
If the covariates are treated as fixed and $P_d^{\mathbf Z}$ is replaced by
their finite-population empirical distribution, the integral becomes
$N_d^{-1}\sum_{j\in\mathcal Q_d}m_0(\mathbf z_{dj})$, with no general
closed-form reduction for nonlinear $m_0$. Thus MERF has the same shared-plus-shrunk-area-effect structure as NERM, but
uses a random forest for the shared component. It illustrates how
Assumption~\ref{assump:structured-shift} accommodates classical mixed models
and flexible machine-learning predictors within the same structure.

\subsubsection{M-quantile regression model}
\label{sec:mq}

\cite{chambers2006m} model between-area variation through area-specific
quantile-like coefficients. For a monotone, bounded $\psi$ and small
$\epsilon>0$, the M-quantile of order $q\in[\epsilon,1-\epsilon]$ of $Y$
given $\mathbf Z=\mathbf z$ is the solution $Q_\psi(\mathbf z;q)$ of
$E\{\psi_q(Y-Q)\mid\mathbf Z=\mathbf z\}=0$, where
$\psi_q(u)=2\psi(u)\{qI(u>0)+(1-q)I(u\le0)\}$
\citep{breckling1988m}. The model assumes
$Q_\psi(\mathbf z;q)=\mathbf z^\top\boldsymbol\zeta_\psi(q)$ for all
$q\in[\epsilon,1-\epsilon]$, with
$q\mapsto\boldsymbol\zeta_\psi(q)$ continuous, and characterises area $d$
by $\delta_d^{\mathrm{MQ}}\in[\epsilon,1-\epsilon]$, the mean of its units'
M-quantile coefficients.\footnote{For unit $(d,j)$, $q_{dj}$ solves
$Q_\psi(\mathbf Z_{dj};q_{dj})=Y_{dj}$.} Unlike Sections~\ref{sec:nerm}--\ref{sec:merf}, the conditional law is only
partially specified through
$P_d^{Y\mid\mathbf Z=\mathbf z}\in
\mathcal K_{\boldsymbol\eta,\delta_d^{\mathrm{MQ}}}
=\{P:E_P[\psi_{\delta_d^{\mathrm{MQ}}}\{Y-\mathbf z^\top
\boldsymbol\zeta_\psi(\delta_d^{\mathrm{MQ}})\}]=0\}$,
while $P_d^{\mathbf Z}$ remains unrestricted. In the notation of
Assumption~\ref{assump:structured-shift},
$\boldsymbol\eta=\{\boldsymbol\zeta_\psi(q):q\in[\epsilon,1-\epsilon]\}$,
with infinite-dimensional $\mathcal H$, is estimated from the pooled sample
by iteratively reweighted least squares over a grid of $q$, while
$\delta_d=\delta_d^{\mathrm{MQ}}$ selects area $d$'s position within the
shared family. Assumption~\ref{assump:constraint} holds with
$\rho_d(\delta)=|\delta-q_0|$, $q_0=1/2$, and
$C_D=D(1/2-\epsilon)^2$. The resulting working prediction function is
$g_d^{\mathrm{MQ}}(\mathbf z)
=\mathbf z^\top\boldsymbol\zeta_\psi(\delta_d^{\mathrm{MQ}})$.
If additionally
$E(Y\mid\mathbf Z=\mathbf z,G=d)=g_d^{\mathrm{MQ}}(\mathbf z)$,
then $g_d^{\mathrm{MQ}}=m_d$ and
$\mu_d^{mq}
=\bar{\mathbf z}_d^\top\boldsymbol\zeta_\psi(\delta_d^{\mathrm{MQ}})$,
where
$\bar{\mathbf z}_d=\int_{\mathcal Z}\mathbf z\,dP_d^{\mathbf Z}(\mathbf z)$.

\section{Small area learners}\label{sec:salearners}

To show SAL in action, we present small area learners of the area mean
functional $\mu_d$ in \eqref{eq:mean_functional}.
Once the asymptotic results
for $\mu_d$ as a target are established, we make the connection to $\bar y_d$ in
\eqref{eq:target_decomposition}. Before delving into the derivations, we
introduce further notation. For a measure $P$ and a function $h$, we write
$Ph = \int h\,dP$, and define the area empirical measures $P_{n_d}h
=
n_d^{-1}\sum_{j\in\mathcal Q_d^{s}}
h(Y_{dj},\mathbf Z_{dj})$ and $P_{N_d}h
=
N_d^{-1}\sum_{j\in\mathcal Q_d}
h(Y_{dj},\mathbf Z_{dj})$ over the sampled and the population units of area $d$, which are i.i.d.\
draws from $P_d$ by Assumption~\ref{assump:design}(iii) and
Assumption~\ref{assump:iid}, respectively.
Since $\mathbf X_{dj}=(\mathbf Z_{dj}^{\top},d)^{\top}$ for every
$j\in\mathcal Q_d$, a function
$g\colon\mathcal X\to\mathbb R$ learned on the full covariate vector is
evaluated, within area $d$, at $g(\mathbf Z_{dj},d)$, and we use the
abbreviation
$g_d(\cdot):=g(\cdot,d)\colon\mathcal Z\to\mathbb R$.
Recall that, by \eqref{eq:mean_functional},
$\mu_d=E_{P_d}(Y)=E_{P_d^{\mathbf Z}}\{m_d(\mathbf Z)\}$,
with $m_d$ as in Definition~\ref{def:constraint_generic}. For any
$g\colon\mathcal X\to\mathbb R$ satisfying
$E_{P_d^{\mathbf Z}}|g_d(\mathbf Z)|<\infty$, one can write
\begin{equation}
\label{eq:expected_decompose}
\mu_d
=
E_{P_d^{\mathbf Z}}\{g_d(\mathbf Z)\}
+
E_{P_d}\{Y-g_d(\mathbf Z)\},
\end{equation}
which holds identically in $g$ and, therefore, does not require $g_d$ to coincide with $m_d$ as in \eqref{eq:area_regression}. This identity motivates the estimators in
Sections~\ref{sec:sample_no_split}--\ref{sec:sal_cf} and \ref{sec:sal_agg}. We first consider the setting in which, in addition to the sampled data
$\mathcal S=\bigcup_d\mathcal S_d$, with
$\mathcal S_d=\{(Y_{dj},\mathbf X_{dj}):j\in\mathcal Q_d^s\}$,
covariates are observed for all nonsampled units,
$\mathcal R_d=\{\mathbf X_{dj}:j\in\mathcal Q_d^r\}$. We use
``learners'' and ``estimators'' interchangeably, reflecting their connection
to semiparametric estimators (Section~\ref{sec:literature_review}); ``small
area learners'' emphasizes that validity is not tied to a particular target
functional or working model. 

%We first consider the setup in which, apart from the sampled data$\mathcal{S}=\bigcup_d\mathcal{S}_d$ of all $n$ observations, with$\mathcal{S}_d=\{(Y_{dj},\mathbf{X}_{dj}):j\in\mathcal{Q}^s_d\}$,auxiliary information is available for \emph{all} population units, i.e., fornon-sampled units we observe$\mathcal{R}_d=\{\mathbf{X}_{dj}:j\in\mathcal{Q}_d^r\}$. In what follows, we use the terms ``learners'' and ``estimators'' interchangeably, reflecting the close connection between our methods and estimators in semiparametric statistics (see Section~\ref{sec:literature_review}). The term ``small area learners'' emphasizes that their validity is not tied to a particular target functional or working model.  

%=====================================================================
%  Section 4.1 (learners without sample splitting), corrected, in order
%=====================================================================

\subsection{Learners without cross-fitting}\label{sec:sample_no_split}

Let $\hat g_n\colon\mathcal X\to\mathbb R$ be estimated (or learned) from
the pooled sample $\mathcal S$ of all $n$ observations. We first present
learners which do not use sample splitting; they draw on results from
empirical process theory for classes of functions whose complexity is
controlled through a bound on their covering numbers (see, for example, Chapter 19 in \citep{vandervaart1998} or Chapter 2 in \citep{vandervaartwellner1996}).

\begin{definition}[Low-complexity class]
\label{def:low_complexity}
Let $(\mathcal G_{c,n})_{n\ge1}$ be a sequence of pointwise measurable
classes of real-valued functions on $\mathcal X$, uniformly bounded
by a constant $B<\infty$, such that for a constant $A\ge e$ and indices $v_n\ge1$,
the bound
$\sup_{\tilde P}
\log N\bigl(
\varepsilon B,\mathcal G_{c,n},L_2(\tilde P)
\bigr)
\le v_n\log(A/\varepsilon)$ holds for every $0<\varepsilon<1$,
%\begin{equation*}
%\sup_{\tilde P}
%\log N\bigl(
%\varepsilon B,\mathcal G_{c,n},L_2(\tilde P)
%\bigr)
%\le
%v_n\log(A/\varepsilon),
%\qquad 
%0<\varepsilon<1,
%\end{equation*}
the supremum being over all finitely discrete probability measures $\tilde P$ on $\mathcal X$, where $N(\cdot)$ denotes the covering number.
\end{definition}

The bound in Definition~\ref{def:low_complexity}
limits the richness of the class $\mathcal G_{c,n}$.
Section~\ref{sec:simulations} shows that it nevertheless
accommodates many commonly used flexible learners, such as uniformly bounded linear, penalised (such as lasso, ridge or elastic net over a dictionary of uniformly bounded basis functions), series,
spline and additive learners, as well as
fixed-complexity trees and ensembles, many of them perform exceptionally well in our simulations and empirical
application; see Section~\ref{sec:simulations}, and Sections 5--6 in the Supplementary Material \citep{RelugaSAL_SM}. Classes satisfying the condition in Definition~\ref{def:low_complexity} are called VC-type
\citep{chernozhukov2014gaussian}; the constants $A$ and $B$ do not depend on
$n$, whereas the index $v_n$ may, and a class not depending on $n$ corresponds
to $v_n\equiv v$. In semiparametric theory, the error incurred by replacing
the population limit of the learner by its pooled estimate $\hat g_n$ is known
as the first-stage or nuisance estimation error. Together with the convergence conditions in Assumption~\ref{assump:learner}, the bound renders the associated empirical-process terms asymptotically negligible. For a class not depending on $n$ this follows from
classical empirical process theory, such a class being $P$-Donsker for every
probability measure $P$ on $\mathcal X$
\citep[Theorem~2.5.2]{vandervaartwellner1996}. When the class is allowed to grow with $n$, the same conclusion follows
from a maximal inequality \citep[Corollary~5.1]{chernozhukov2014gaussian},
provided the complexity index $v_n$ does not grow too fast, as required in
Assumption~\ref{assump:entropy_growth}.

\begin{assumption}[Entropy growth]
\label{assump:entropy_growth} \noindent The index $v_n$ of Definition~\ref{def:low_complexity} satisfies
$n_d^{-1/2}v_n \log n_d \to 0$, which holds automatically when the class does not depend on $n$.
\end{assumption}

The condition in Assumption~\ref{assump:entropy_growth} restricts the growth of learner complexity and holds for the fixed-complexity learners mentioned above and used in Section~\ref{sec:simulations}, Sections~5-6 of the Supplementary Material \citep{RelugaSAL_SM}, for which $v_n=O(1)$.
Assumption~\ref{assump:regime} specifies the asymptotic regime for
$n_d$, $N_d$, and $D_n$ and, together with
Assumption~\ref{assump:entropy_growth}, underlies the convergence results in
Proposition~\ref{prop:normality_family}.

\begin{assumption}[Asymptotic regime]
\label{assump:regime}
Consider a sequence of populations indexed by $n$, with the $n$th population
comprising $D_n$ areas indexed by $\mathcal D_n$. Area
$d\in\mathcal D_n$ has population size $N_d$ and sample size $n_d$, with
$n=\sum_{d\in\mathcal D_n}n_d$. For the area $d$ under consideration:
(i) $n_d\to\infty$; (ii) $n_d/N_d\to0$; and
(iii) either $\sup_nD_n<\infty$ or $D_n\to\infty$.
\end{assumption}

Assumption~\ref{assump:regime} imposes no direct rate on $D_n$; when
$D_n\to\infty$, restrictions on its growth arise indirectly through the
complexity index $v_n$ in Assumption~\ref{assump:entropy_growth} and the
area-specific convergence condition in Assumption~\ref{assump:learner}. The
deviation bound in Assumption~\ref{assump:constraint} plays no part in these
conditions; see the discussion following Assumption~\ref{assump:learner}.
Applied to the $n$th population, it yields $C_{D_n}$, and the results below
are compatible with $C_{D_n}=O(D_n)$, hence with nonvanishing average
heterogeneity across areas. Definition~\ref{def:SAL_family} presents the
general small area learner, obtained from the empirical counterpart of
\eqref{eq:expected_decompose} by weighting its sample residual correction.
Finally, the dependence of $N_d$ and $n_d$ on $n$ is suppressed for
notational simplicity. Definition~\ref{def:SAL_family} presents an estimator of $\mu_d$ in \eqref{eq:mean_functional}.

\begin{definition}[Small area learner]
\label{def:SAL_family}
Let $\hat g_{n,d}=\hat g_n(\cdot,d)$ be learned from the pooled sample
$\mathcal S$, and let $a\in[0,1]$ be a correction weight. The \emph{small area
learner} is
\begin{equation}
\label{eq:sal_family}
\widehat\mu_d(a)
=
P_{N_d}\hat g_{n,d}
+
a\,P_{n_d}(Y-\hat g_{n,d}).
\end{equation}
\end{definition}

Different choices of $a$ recover familiar estimators.

\begin{example}\label{ex:SSE}
For $a=1$, SAL recovers the fully residual-corrected estimator
\citep{zhang_high-dimensional_2022},
$\widehat{\mu}^{sal}_{d}
=
P_{N_d}\hat g_{n,d}+P_{n_d}(Y-\hat g_{n,d})
=
N_d^{-1}\sum_{j\in\mathcal Q_d}\hat g_{n,d}(\mathbf Z_{dj})
+n_d^{-1}\sum_{j\in\mathcal Q_d^s}
\{Y_{dj}-\hat g_{n,d}(\mathbf Z_{dj})\}$.
The learner $\hat g_n$ is fitted on the pooled sample $\mathcal Q^s$, predictions are formed for all units in $\mathcal Q_d$, and the residual correction uses $\mathcal Q_d^s$.
\end{example}

\begin{example}\label{ex:SAE}
For $a=f_d:=n_d/N_d$, SAL recovers the model-based SAE estimator
\citep{rao2015},
$\widehat{\mu}^{sae}_{d}
=
P_{N_d}\hat g_{n,d}
+f_dP_{n_d}(Y-\hat g_{n,d})
=
N_d^{-1}\{\sum_{j\in\mathcal Q_d^s}Y_{dj}
+\sum_{j\in\mathcal Q_d^r}\hat g_{n,d}(\mathbf Z_{dj})\}$.  The learner $\hat g_n$ is learned on the pooled sample $\mathcal Q^s$,
observed outcomes are retained for $\mathcal Q_d^s$, and only units in
$\mathcal Q_d^r$ are imputed.
\end{example}

\begin{example}\label{ex:syn}
For $a=0$, SAL recovers the synthetic estimator \citep{rao2015},
$\widehat{\mu}^{syn}_{d}
=
P_{N_d}\hat g_{n,d}
=
N_d^{-1}\sum_{j\in\mathcal Q_d}
\hat g_{n,d}(\mathbf Z_{dj})$. The learner $\hat g_n$ is learned on the pooled sample $\mathcal Q^s$ and all units in
$\mathcal Q_d$ are imputed.
\end{example}

Under Proposition~\ref{prop:normality_family},
$\widehat\mu^{sae}_d-\widehat\mu^{syn}_d
=f_dP_{n_d}(Y-\hat g_{n,d})=o_p(1)$, whereas
$\widehat\mu^{sal}_d$ is the fully corrected member with $a=1$; the
proposition characterizes the consequences of these correction weights.
Under a fixed-size, self-weighting design satisfying
Assumption~\ref{assump:design}, $p_d=n_d/N_d$ and
$\widehat{\mu}^{sal}_{d}$ coincides with the model-assisted generalized
difference estimator \citep{deville1992calibration}. It also connects to
model-assisted estimation with flexible learners
\citep{an2026agnostic} and, through its residual correction, to orthogonal
estimating equations in debiased machine learning
\citep{chernozhukov2018}; see Section~\ref{sec:literature_review}.

To describe the limit of \eqref{eq:sal_family} we need the population target
of the pooled learner. Let $\varpi_{d,n}=n_d/n$ be the sample proportion of area $d$  to the pooled sample. Recall from \eqref{eq:area_law} that
$P_d^{\mathbf Z}=\mathcal L(\mathbf Z\mid G=d)$ is the conditional law of the
covariates given membership of area $d$. Let $\bar P_n$ be the law of
$\mathbf X=(\mathbf Z^\top,G)^\top$ on $\mathcal X$ defined by $\bar P_n(\mathbf Z\in A,\,G=d)
=
\varpi_{d,n}\,P_d^{\mathbf Z}(A)$, $A\in\mathcal B(\mathcal Z)$, $d\in\mathcal D_n$ that is, the mixture of the area covariate laws with weights given by the composition of the pooled sample. Assumption~\ref{assump:learner} states two properties of the pooled learner
$\hat g_n$ of Definition~\ref{def:SAL_family}, necessary for proving the
asymptotic normality of $\widehat\mu_d(a)$ in \eqref{eq:sal_family} in
Proposition~\ref{prop:normality_family}. First, $\hat g_n$ must eventually
lie in the low-complexity class of Definition~\ref{def:low_complexity}.
Second, its restriction $\hat g_{n,d}$ converges in mean square to a fixed
function in the area under consideration, as required by the
empirical-process arguments in the proof.

\begin{assumption}[Nuisance learner]
\label{assump:learner}
The pooled learner $\hat g_n$ satisfies
$\Pr(\hat g_n\in\mathcal G_{c,n})\to1$, and there exists a function $g^*$, its
population limit, with $g^*_d\in L_2(P_d^{\mathbf Z})$, such that
$r_{n,d}:=\|\hat g_{n,d}-g^*_d\|_{P_d^{\mathbf Z},2}=o_p(1)$. If
$\sup_n v_n=\infty$, where $v_n$ is the uniform entropy exponent of
$\mathcal G_{c,n}$ in Definition~\ref{def:low_complexity}, we additionally require $r_{n,d}\sqrt{v_n\log n_d}=o_p(1)$.
\end{assumption}

Assumption~\ref{assump:learner} requires convergence only to a fixed limit,
and the limits in Propositions~\ref{prop:normality_family} and
\ref{prop:normality_cf_family} depend on $g^*$ only through the mean
$\beta_d$ and variance $\sigma_d^2$ of the population residual in
\eqref{eq:sigma_d}. Let
$\|h\|_{\bar P_n}^2
:=\sum_{d\in\mathcal D_n}\varpi_{d,n}
E_{P_d^{\mathbf Z}}\{h_d(\mathbf Z)^2\}$
denote the pooled $L_2$ norm. For a learner obtained by minimising the
empirical counterpart of the pooled squared-error risk
$\sum_{d\in\mathcal D_n}\varpi_{d,n}E_{P_d}[\{Y-g_d(\mathbf Z)\}^2]$
over a class $\mathcal G_{c,n}$ that is nonempty, closed, and convex in
$L_2(\bar P_n)$, let $g_n^\circ$ denote the corresponding population-risk
minimiser. Since squared-error risk differs from $\|m-g\|_{\bar P_n}^2$
only by a term independent of $g$, and $m$ in \eqref{eq:generic_m} belongs
to $L_2(\bar P_n)$, $g_n^\circ$ is the unique projection of $m$ onto
$\mathcal G_{c,n}$ in that norm
\citep[e.g.][Theorem~5.2]{brezis2011functional}. The
decomposition
$\|\hat g_n-g^*\|_{\bar P_n}
\leq\|\hat g_n-g_n^\circ\|_{\bar P_n}
+\|g_n^\circ-g^*\|_{\bar P_n}$
separates the estimation error from the convergence of the population-risk
minimiser to its limit. If
$m\in\mathcal G_{c,n}$, then $g_n^\circ=m$; hence, under the stated
consistency, correct specification yields $g^*=m$ and
$g_d^*=m_d$ in \eqref{eq:area_regression}. For learners not defined by
empirical risk minimisation, $g^*$ simply denotes the fixed probability
limit required by Assumption~\ref{assump:learner}.

Assumption~\ref{assump:learner} is stated for the area under consideration,
whereas pooled convergence is naturally measured by
$\|\cdot\|_{\bar P_n}$. Since
$\|\hat g_{n,d}-g_d^*\|_{P_d^{\mathbf Z},2}^2
\leq\varpi_{d,n}^{-1}\|\hat g_n-g^*\|_{\bar P_n}^2$,
pooled consistency implies area-wise consistency whenever $\varpi_{d,n}$ is
bounded away from zero. When $\varpi_{d,n}\to0$, the same conclusion follows
if $\|\hat g_n-g^*\|_{\bar P_n}^2=o_p(\varpi_{d,n})$; if
$\sup_n v_n=\infty$, the stronger rate
$o_p\{\varpi_{d,n}/(v_n\log n_d)\}$ also ensures the additional convergence
requirement in Assumption~\ref{assump:learner}.

Assumption~\ref{assump:learner_mean} adds a rate condition on the area mean error to the mean-square convergence required by
Assumption~\ref{assump:learner}.

\begin{assumption}[Learner mean error]
\label{assump:learner_mean}
For the area $d$ under consideration,
$\sqrt{n_d}\,P_d(\hat g_{n,d}-g_d^*)=o_p(1)$.
\end{assumption}
For a regular parametric
learner converging at rate $n^{-1/2}$, Assumption~\ref{assump:learner_mean} holds when
$n_d/n\to0$, but is not generally guaranteed when
$n_d/n$ remains bounded away from zero. For
nonparametric learners, it requires separate
verification. As shown below, the condition is
needed for the root-$n_d$ limit of partially corrected
members of the family in Definition~\ref{def:SAL_family}, but not for the fully corrected estimator.

For each area, let
$\varepsilon_d^*:=Y-g_d^*(\mathbf Z)$ denote
the limiting population residual, such that %with mean and variance
\begin{equation}
\label{eq:sigma_d}
\beta_d:=E_{P_d}(\varepsilon_d^*)
=\mu_d-E_{P_d^{\mathbf Z}}\{g_d^*(\mathbf Z)\},
\qquad
\sigma_d^2:=\operatorname{Var}_{P_d}(\varepsilon_d^*)=E_{P_d}
\bigl[
\{\varepsilon_d^*-E_{P_d}(\varepsilon_d^*)\}^2
\bigr].
\end{equation}
Both quantities are finite when
$E_{P_d}(Y^2)<\infty$ and
$g_d^*\in L_2(P_d^{\mathbf Z})$.
The mean residual $\beta_d$ measures the
area-level approximation bias and vanishes whenever $g_d^*$ in Assumption~\ref{assump:learner} reproduces the area mean of
$m_d$ in \eqref{eq:area_regression}, in particular when $g_d^*=m_d$
$P_d^{\mathbf Z}$-almost surely.

\begin{proposition}[Asymptotic normality]
\label{prop:normality_family}
Let $\hat\mu_d(a_n)$ be the estimator in
\eqref{eq:sal_family}, evaluated at a
deterministic sequence $a_n\in[0,1]$ with
$a_n\to a\in[0,1]$. Suppose that
$\sigma_d^2>0$, $E_{P_d}(Y^2)<\infty$, and
Assumptions~\ref{assump:design},
\ref{assump:iid}, \ref{assump:entropy_growth},
\ref{assump:regime} and
\ref{assump:learner} hold. Then:
\begin{enumerate}
\item[(i)]
$\hat\mu_d(a_n)-\mu_d
\xrightarrow{p}(a-1)\beta_d$;
\item[(ii)]
if $a_n\equiv1$, then $\sqrt{n_d}
(\hat\mu_d^{sal}-\mu_d)
\xrightarrow{d}N(0,\sigma_d^2)$;
\item[(iii)]
if $a\in(0,1]$, then under
Assumption~\ref{assump:learner_mean},
\begin{equation*}
\label{eq:clt_family}
\sqrt{n_d}
\{\hat\mu_d(a_n)-\mu_d-(a_n-1)\beta_d\}
\xrightarrow{d}N(0,a^2\sigma_d^2);
\end{equation*}
\item[(iv)]
if $a_n\to0$ and, in place of
Assumption~\ref{assump:learner_mean},
$P_d(\hat g_{n,d}-g_d^*)
=O_p(n^{-\gamma})$
for some $\gamma\in(0,1/2]$, then $\hat\mu_d(a_n)-\mu_d-(a_n-1)\beta_d=O_p\left(
a_n n_d^{-1/2}
+N_d^{-1/2}
+n^{-\gamma}
\right)$.
\end{enumerate}
\end{proposition}

Proposition~\ref{prop:normality_family} distinguishes $\hat\mu_d^{sal}$ from the other family members. By part (ii), $\hat\mu_d^{sal}$ is consistent and asymptotically normal around $\mu_d$ \emph{without} requiring $g_d^*=m_d$, with $m_d$ as in \eqref{eq:area_regression}, or Assumption~\ref{assump:learner_mean}, because the residual correction of $\hat\mu_d^{sal}$ applied in $\mathcal{Q}_d^s$ eliminates both the limiting bias $(a-1)\beta_d$ and the corresponding area mean learner-estimation term. In contrast, the model-based and synthetic estimators in Examples~\ref{ex:SAE} and \ref{ex:syn} converge to $\mu_d-\beta_d$ and are therefore consistent for $\mu_d$ only when $\beta_d=0$. The cross-fitted family $\hat\mu_d^{cf}(a)$ in \eqref{eq:cf_family} has the same first-order bias--variance structure, as discussed in Section~\ref{sec:sal_cf}.

\subsection{Cross-fitted small area learners}
\label{sec:sal_cf}

The low-complexity requirement in Definition~\ref{def:low_complexity},
together with Assumption~\ref{assump:entropy_growth}, restricts the pooled
learner in Definition~\ref{def:SAL_family} to classes of fixed or suitably
controlled complexity. Cross-fitting \citep{chernozhukov2018} removes the low-complexity requirement by training the learner on one part of the sample and evaluating it on another, so that under the sampling and fold-allocation conditions stated below, and conditionally on $\mathcal F_{\mathcal P}$ defined in Assumption~\ref{assump:iid}, the fitted function is independent of the observations on which it is evaluated. Widely used in semiparametric and
machine-learning-based inference, and more recently in design-based survey
statistics \citep{an2026agnostic}, cross-fitting therefore
extends SAL to learners such as deep forests, boosting with growing depth, and
neural networks.

\begin{definition}[Cross-fitted small area learner]
\label{def:SAL_cf_family}
Let $\mathcal I_1,\dots,\mathcal I_L$ be a partition of the pooled sampled
index set $\mathcal Q^s=\bigcup_d\mathcal Q^s_d$ into $L\ge2$ folds, with $L$
fixed. For each fold $\ell$, let
$\hat g_n^{(-\ell)}\colon\mathcal X\to\mathbb R$ be the pooled learner trained
on the sampled data outside $\mathcal I_\ell$, and write
$\hat g_{n,d}^{(-\ell)}:=\hat g_n^{(-\ell)}(\cdot,d)$ and
$\bar g_{n,d}:=L^{-1}\sum_{\ell=1}^L\hat g_{n,d}^{(-\ell)}$. Let
$\hat g^{\mathrm{oof}}_{n,d}$ denote the out-of-fold learner, defined at the
sampled units of area $d$ by
$\hat g^{\mathrm{oof}}_{n,d}(\mathbf Z_{dj}):=\hat g^{(-\ell)}_{n,d}(\mathbf Z_{dj})$
for $j\in\mathcal Q_d^s\cap\mathcal I_\ell$ and $\ell=1,\dots,L$. For a
correction weight $a\in[0,1]$, the cross-fitted small area learner is
\begin{equation}
\label{eq:cf_family}
\widehat\mu^{cf}_d(a)
=
P_{N_d}\bar g_{n,d}
+
a\,P_{n_d}(Y-\hat g^{\mathrm{oof}}_{n,d}).
\end{equation}
\end{definition}

The first term of \eqref{eq:cf_family} uses the fold-averaged fit
$\bar g_{n,d}$, while the second computes each sampled unit's residual from
the fit that did not use its own fold. The values $a=1$, $a=f_d$ and $a=0$
give the cross-fitted counterparts of the estimators in
Examples~\ref{ex:SSE}--\ref{ex:syn}, and we write
$\widehat\mu^{cf}_d:=\widehat\mu^{cf}_d(1)$. Assumptions~\ref{assump:cf_learner} and \ref{assump:cf_learner_mean} replace Definition~\ref{def:low_complexity} and
Assumptions~\ref{assump:entropy_growth}, \ref{assump:learner} and
\ref{assump:learner_mean} at the expense of sample splitting. The population limit $g^*$ plays the same role as in Assumption~\ref{assump:learner}.

\begin{assumption}[Cross-fitted nuisance learner]
\label{assump:cf_learner}
Conditionally on $\mathcal F_{\mathcal P}$ defined in Assumption~\ref{assump:iid}, the sampling indicators are independent of the population pairs. The fold allocation is independent of the observed values, conditional on the area labels and sampling indicators. For the area $d$ under consideration, the folds are stratified so that
$|n_{d\ell}-n_d/L|\le1$ for $\ell=1,\dots,L$, where
$n_{d\ell}=|\mathcal Q_d^s\cap\mathcal I_\ell|$, and the sampling fraction
satisfies $n_d^{3/2}/N_d\to0$. There exist
$g^*\colon\mathcal X\to\mathbb R$ and a constant $B<\infty$ with
$|g^*_d|\le B$ and $|\hat g^{(-\ell)}_{n,d}|\le B$ for $\ell=1,\dots,L$, such
that $\|\hat g_{n,d}^{(-\ell)}-g_d^*\|_{P_d^{\mathbf Z},2}=o_p(1)$,
$\ell=1,\dots,L$.
\end{assumption}

\begin{lemma}[Population averaging of the fold-averaged learner]
\label{lemma:cf_population_averaging}
Under Assumptions~\ref{assump:design}, \ref{assump:iid} and
\ref{assump:cf_learner}, the fold-averaged learner satisfies $(P_{N_d}-P_d)(\bar g_{n,d}-g_d^*)
=
O_p\left(\frac{n_d}{N_d}+N_d^{-1/2}\right)$ and $\sqrt{n_d}\,(P_{N_d}-P_d)(\bar g_{n,d}-g_d^*)
=o_p(1)$.
%\begin{equation}
%\label{eq:cf_population_averaging}
%(P_{N_d}-P_d)(\bar g_{n,d}-g_d^*)=
%O_p\left(\frac{n_d}{N_d}+N_d^{-1/2}\right),
%\quad
%\sqrt{n_d}\,(P_{N_d}-P_d)(\bar g_{n,d}-g_d^*)
%=o_p(1).
%\end{equation}
\end{lemma}

Lemma~\ref{lemma:cf_population_averaging} controls the first term of
\eqref{eq:cf_family} evaluated over the full population. Cross-fitting
provides independence between each fitted learner and the observations in its evaluation fold. This independence can be extended to the full-population average by evaluating sampled units at their own out-of-fold fits and nonsampled units at $\bar g_{n,d}$, replacing $P_{N_d}\bar g_{n,d}$ with
$N_d^{-1}\{\sum_{j\in\mathcal Q_d^s}
\hat g^{\mathrm{oof}}_{n,d}(\mathbf Z_{dj})
+\sum_{j\in\mathcal Q_d^r}
\bar g_{n,d}(\mathbf Z_{dj})\}$.
Under this construction, each evaluation is independent of the data used to produce the corresponding fit, allowing the condition $n_d^{3/2}/N_d\to0$ in
Assumption~\ref{assump:cf_learner} to be replaced by $n_d/N_d\to0$, as already required by Assumption~\ref{assump:regime}(ii). Alternatively, when retaining the original formulation $P_{N_d}\bar g_{n,d}$, this rate condition provides the additional control needed for the sampled part of the population average, whose units also contribute to training $\bar g_{n,d}$.

\begin{assumption}[Cross-fitted learner mean error]
\label{assump:cf_learner_mean}
For the area $d$ under consideration,
$\sqrt{n_d}\,P_d(\hat g^{(-\ell)}_{n,d}-g^*_d)=o_p(1)$ for
$\ell=1,\dots,L$.
\end{assumption}

Assumption~\ref{assump:cf_learner_mean} is the counterpart of Assumption~\ref{assump:learner_mean} and is used for the centred root-$n_d$ normal limit when $a\in(0,1)$; it is not required for the fully corrected estimator with $a_n\equiv1$. Cross-fitting does not dispense with it: the area mean
of the estimation error is a population quantity, unaffected by the
independence between a fitted learner and its evaluation fold. Proposition~\ref{prop:normality_cf_family} shows normality of $\widehat\mu^{cf}_d(a_n)$ in \eqref{eq:cf_family}.

\begin{proposition}[Asymptotic normality of the cross-fitted learner]
\label{prop:normality_cf_family}
Let $\hat\mu_d^{cf}(a_n)$ be the estimator in
\eqref{eq:cf_family}, evaluated at a deterministic
sequence $a_n\in[0,1]$ with $a_n\to a\in[0,1]$.
Suppose that $\sigma_d^2>0$,
$E_{P_d}(Y^2)<\infty$, and
Assumptions~\ref{assump:design},
\ref{assump:iid}, \ref{assump:regime}
and \ref{assump:cf_learner} hold. Then:
\begin{enumerate}

\item[(i)]
$\hat\mu_d^{cf}(a_n)-\mu_d
\xrightarrow{p}(a-1)\beta_d$;

\item[(ii)]
if $a_n\equiv1$, then
$\sqrt{n_d}(\hat\mu_d^{cf}-\mu_d)
\xrightarrow{d}N(0,\sigma_d^2)$;

\item[(iii)]
if $a\in(0,1]$, then under
Assumption~\ref{assump:cf_learner_mean},
\begin{equation*}
\sqrt{n_d}
\{\hat\mu_d^{cf}(a_n)-\mu_d
-(a_n-1)\beta_d\}
\xrightarrow{d}N(0,a^2\sigma_d^2);
\end{equation*}

\item[(iv)]
if $a_n\to0$ and, in place of
Assumption~\ref{assump:cf_learner_mean},
$P_d(\bar g_{n,d}-g_d^*)
=O_p(n^{-\gamma})$
for some $\gamma\in(0,1/2]$, then $\hat\mu_d^{cf}(a_n)-\mu_d
-(a_n-1)\beta_d
=
O_p\left(
a_n n_d^{-1/2}
+N_d^{-1/2}
+\frac{n_d}{N_d}
+n^{-\gamma}
\right)$.
%\begin{equation*}
%\hat\mu_d^{cf}(a_n)-\mu_d
%-(a_n-1)\beta_d
%=
%O_p\left(
%a_n n_d^{-1/2}
%+N_d^{-1/2}
%+\frac{n_d}{N_d}
%+n^{-\gamma}
%\right).
%\end{equation*}
\end{enumerate}
\end{proposition}

Propositions~\ref{prop:normality_family} and \ref{prop:normality_cf_family}
yield the same first-order limits but differ in how learner-estimation error
is controlled: by the low-complexity condition of
Assumption~\ref{assump:learner} without sample splitting, and by conditioning
on the training sample under cross-fitting. The population-averaging
remainder is controlled by Lemma~1.4 in the Supplementary Material \citep{RelugaSAL_SM} in the first case and by Lemma~\ref{lemma:cf_population_averaging} in the second, while
stratification keeps $n_{d\ell}/n_d$ within $n_d^{-1}$ of $1/L$, so that the
fold-weight imbalance is asymptotically negligible. The two propositions therefore establish the same
first-order bias--variance structure for the standard and cross-fitted SAL
families in Definitions~\ref{def:SAL_family} and~\ref{def:SAL_cf_family}.
Ignoring finite-population and learner-estimation terms, the common
bias--variance approximation is
$\operatorname{AMSE}_d(a\mid\beta_d)
=(1-a)^2\beta_d^2+a^2\sigma_d^2/n_d$, minimised at the oracle weight
$a_d^\star=\beta_d^2/(\beta_d^2+\sigma_d^2/n_d)$.
If $\beta_d=0$, then $a_d^\star=0$ and
$\operatorname{AMSE}_d(a\mid\beta_d)=0$ at the synthetic endpoint in
Example~\ref{ex:syn}; the actual comparison is then determined by the omitted
learner-estimation and finite-population terms. If $\beta_d\neq0$ is fixed,
$a_d^\star\to1$, so the fully corrected estimator in Example~\ref{ex:SSE} is
asymptotically optimal within this first-order criterion, whereas the SAE
weight $f_d=n_d/N_d$ in Example~\ref{ex:SAE} tends to zero under
Assumption~\ref{assump:regime}(ii). More generally, if
$\sqrt{n_d}\beta_d\to c\in(0,\infty)$, then
$a_d^\star\to c^2/(c^2+\sigma_d^2)$, while $a_d^\star\to1$ only when
$\sqrt{n_d}|\beta_d|\to\infty$. The weight $a_d^\star$ is oracle because it
depends on the unknown approximation error $\beta_d$.

For either SAL family, the limiting bias can also be removed directly.
Writing $\hat\mu_d^{(\cdot)}(a)$ for either the standard or cross-fitted
estimator in Definitions~\ref{def:SAL_family} and
\ref{def:SAL_cf_family}, respectively, define the corresponding area-level
residual mean by
\begin{equation}
\label{eq:resid_g_hat}
\bar r_{n_d}^{(\cdot)}
:=
P_{n_d}(Y-\hat g_{n,d}^{(\cdot)}),
\qquad
\hat g_{n,d}^{(\cdot)}
\in
\{\hat g_{n,d},\hat g_{n,d}^{\mathrm{oof}}\}.
\end{equation}
The quantity in \eqref{eq:resid_g_hat} is the sample residual analogue of
$\beta_d$. For either estimator, the identity
$\hat\mu_d^{(\cdot)}(a_n)
+(1-a_n)\bar r_{n_d}^{(\cdot)}
=\hat\mu_d^{(\cdot)}(1)$
holds exactly for every sample realization. Thus, for example, adding
$(1-f_d)\bar r_{n_d}^{(\cdot)}$ to the SAE estimator in
Example~\ref{ex:SAE}, where $f_d=n_d/N_d$, transforms it exactly into the
fully corrected estimator of Example~\ref{ex:SSE}. Its residual-corrected
version is therefore the $a_n\equiv1$ estimator covered by part~(ii) of
Propositions~\ref{prop:normality_family} and
\ref{prop:normality_cf_family}.

This removal of the limiting bias comes at the cost of an increased residual-variance component. For a fixed $a\in(0,1)$, under the conditions of part~(iii) of the respective proposition, the uncorrected member has limiting bias $(a-1)\beta_d$ relative to $\mu_d$ and first-order variance $a^2\sigma_d^2/n_d$, whereas the corrected estimator has zero limiting bias and first-order variance $\sigma_d^2/n_d$. For the SAE choice $a_n=f_d$, where $f_d$ is often treated as fixed as opposed to $f_d\to0$ under Assumption~\ref{assump:regime}(ii), the corresponding residual component is $f_d^2\sigma_d^2/n_d$ before and $\sigma_d^2/n_d$ after residual correction,
an increase by the factor $1/f_d^2$. This factor concerns only the residual component,
not the total variance, which may also depend at leading order on learner-estimation and finite-population terms. This bias--variance trade-off is
particularly visible in empirical analyses in Sections~\ref{sec:simulations}, and Sections~5--6 of the Supplementary Material \citep{RelugaSAL_SM}.

The endpoint $a_n\to0$ requires separate attention. For the synthetic
estimator, $a_n\equiv0$, while for the SAE estimator $a_n=f_d\to0$ under
Assumption~\ref{assump:regime}(ii). Under the additional learner-rate
conditions imposed there, parts~(iv) of
Propositions~\ref{prop:normality_family} and
\ref{prop:normality_cf_family} therefore give rate statements rather than
nondegenerate root-$n_d$ limit laws: the residual-mean fluctuation vanishes
with $a_n$, leaving the rate to depend on finite-population averaging and
learner-estimation error. A sharper limit law requires additional assumptions
on the learner and the relative rates.

%Thus, the residual mean supplies the bias correctionwithout requiring a separate model for $\beta_d$.
%%%%%%%%%%%%%%%%%%%%%%%%%%%%%%%
%%%%%%% Finite-population parameter
%%%%%%%%%%%%%%%%%%%%%%%%%%%%%%%%%%%%%
\subsection{Finite-population parameter}
\label{sec:finite_samp_mean}

The preceding results also extend to the finite-population mean $\bar Y_d$ 
under the negligible sampling fraction condition in
Assumption~\ref{assump:regime}.

\begin{corollary}[Finite-population mean]
\label{cor:finite_pop}
Let $\bar Y_d$ denote the superpopulation-random
analogue of $\bar y_d$ in
\eqref{eq:target_decomposition}, and let
$\hat\mu_d^{(\cdot)}(a_n)$ denote either $\hat\mu_d(a_n)$ or
$\hat\mu_d^{cf}(a_n)$, as defined in
Definitions~\ref{def:SAL_family} and
\ref{def:SAL_cf_family}, respectively.
Under the assumptions of
Proposition~\ref{prop:normality_family} or
\ref{prop:normality_cf_family}, respectively, their conclusions remain valid with $\bar Y_d$ in place of $\mu_d$. In particular, for $a\in(0,1]$, under the respective conditions of parts (ii) and (iii), $\sqrt{n_d}
\{\hat\mu_d^{(\cdot)}(a_n)-\bar Y_d
-(a_n-1)\beta_d\}
\xrightarrow{d}N(0,a^2\sigma_d^2)$.
%\begin{equation}
%\label{eq:clt_finite}
%\sqrt{n_d}
%\{\hat\mu_d^{(\cdot)}(a_n)-\bar Y_d
%-(a_n-1)\beta_d\}
%\xrightarrow{d}N(0,a^2\sigma_d^2).
%\end{equation}
\end{corollary}

The result follows from
$\hat\mu_d^{(\cdot)}(a_n)-\bar Y_d
=\{\hat\mu_d^{(\cdot)}(a_n)-\mu_d\}-(\bar Y_d-\mu_d)$.
Since $\bar Y_d-\mu_d=O_p(N_d^{-1/2})$, the second term
is negligible at the $n_d^{-1/2}$ scale under
Assumption~\ref{assump:regime}. Thus, the superpopulation and finite-population means are
asymptotically indistinguishable at this scale, so replacing $\mu_d$ by $\bar Y_d$ in the centred estimation error leaves its root-$n_d$ limiting distribution unchanged. Accordingly, the conclusions established for the superpopulation mean carry over unchanged to $\bar Y_d$: the model-based and synthetic estimators in Examples~\ref{ex:SAE} and \ref{ex:syn} remain sensitive to the approximation error $\beta_d$, whereas the fully residual-corrected estimator in Example~\ref{ex:SSE} remains consistent and asymptotically normal under the respective standard or cross-fitted assumptions, providing robustness to learner misspecification in the sense that it does not require $g_d^*=m_d$ with $g_d^*$ in Assumption~\ref{assump:learner} and $m_d$ as in \eqref{eq:area_regression}.

The transfer from $\mu_d$ to $\bar Y_d$ uses the negligible sampling fraction to ensure that $\sqrt{n_d}(\bar Y_d-\mu_d)=o_p(1)$. When this condition fails, $\bar Y_d-\mu_d$, whose variance is of order $N_d^{-1}$, may contribute at first order. Variance estimation for the finite-population mean must then account for this term and its covariance with $\hat\mu_d^{(\cdot)}(a_n)-\mu_d$, requiring a separate finite-population analysis rather than a direct application of Corollary~\ref{cor:finite_pop}.

%%%%%%%%%%%%%%%%%%%%%%%%%%%%%%%%%%%%%
%%% Variance estimation
%%%%%%%%%%%%%%%%%%%%%%%%%%%%%%%%%%%%

\subsection{Variance estimation and confidence intervals}
\label{sec:variance}

Write $\hat\mu_d^{(\cdot)}(a)$ for either the standard or the cross-fitted
estimator and $\bar r_{n_d}^{(\cdot)}$ for the corresponding area residual
mean of \eqref{eq:resid_g_hat}. The variance $\sigma_d^2$ in
\eqref{eq:sigma_d} is estimated by the sample variance of the fitted area
residuals,
\begin{equation}
\label{eq:sigma_hat}
\hat\sigma_d^2
:=\frac{1}{n_d-1}\sum_{j\in\mathcal Q_d^s}
\bigl(\hat\varepsilon_{dj}-\bar r_{n_d}^{(\cdot)}\bigr)^2,
\end{equation}
where $\hat\varepsilon_{dj}=Y_{dj}-\hat g_{n,d}(\mathbf Z_{dj})$ without
sample splitting and
$\hat\varepsilon_{dj}=Y_{dj}-\hat g_{n,d}^{(-\ell)}(\mathbf Z_{dj})$ for
$j\in\mathcal Q_d^s\cap\mathcal I_\ell$ under cross-fitting. Within either
family $\hat\sigma_d^2$ does not depend on the correction weight, which
enters the estimator but not the residuals.

\begin{corollary}[Variance estimation and confidence intervals]
\label{cor:variance}
Under the assumptions of
Proposition~\ref{prop:normality_family}
or Proposition~\ref{prop:normality_cf_family}
for the fully corrected estimator, $\hat\sigma_d^2\xrightarrow{p}\sigma_d^2$. For any deterministic $a_n\in[0,1]$, let
$\tilde\mu_d^{(\cdot)}(a_n)
:=\hat\mu_d^{(\cdot)}(a_n)
+(1-a_n)\bar r_{n_d}^{(\cdot)}$.
Then, for $\alpha\in(0,1)$,
\begin{equation}
\label{eq:ci}
\Pr\left\{
\left|\tilde\mu_d^{(\cdot)}(a_n)-\mu_d\right|
\leq z_{1-\alpha/2}
\frac{\hat\sigma_d}{\sqrt{n_d}}
\right\}
\longrightarrow 1-\alpha,
\end{equation}
where $\hat\sigma_d=(\hat\sigma_d^2)^{1/2}$ and $z_{1-\alpha/2}$ denotes the
$(1-\alpha/2)$-quantile of the standard normal distribution. The conclusion also holds with $\mu_d$
replaced by $\bar Y_d$ under the assumptions
of Corollary~\ref{cor:finite_pop}.
\end{corollary}

As shown in the preceding section,
$\tilde\mu_d^{(\cdot)}(a_n)=\hat\mu_d^{(\cdot)}(1)$
holds exactly for every $a_n$. Hence the interval in
\eqref{eq:ci} is the common fully corrected interval $[
\hat\mu_d^{(\cdot)}(1)
\pm
z_{1-\alpha/2}\hat\sigma_d/ \sqrt{n_d}]$. The estimator in Example~\ref{ex:SSE} therefore requires no
additional adjustment, whereas the SAE and synthetic estimators in Examples~\ref{ex:SAE}-\ref{ex:syn}  are first mapped to the fully corrected member by the residual correction described
above. This construction requires neither
Assumption~\ref{assump:learner_mean} nor
Assumption~\ref{assump:cf_learner_mean}, because inference is based on the
$a_n\equiv1$ member covered by part~(ii) of the corresponding proposition.
In contrast, the interval in
\eqref{eq:ci} is not centred at the estimator $\hat\mu_d^{(\cdot)}(a_n)$. Inference for that estimator must account for its approximation bias and, when relevant, for additional
learner-estimation variability.  Finally, for inference on $\bar Y_d$, multiplying
the interval radius by
$(1-n_d/N_d)^{1/2}$
is asymptotically irrelevant under
Assumption~\ref{assump:regime}(ii).
Justifying a finite-population correction
when the sampling fraction is nonnegligible
requires a separate analysis; see
Section~\ref{sec:finite_samp_mean}.

\section{Approximate small area learners}
\label{sec:sal_agg}

The small area learners in
Sections~\ref{sec:sample_no_split}--\ref{sec:sal_cf}
require unit-level auxiliary information for the entire
finite population, typically obtained from censuses or
administrative registers. Sections~\ref{sec:agg_no_cf}
and \ref{sec:agg_cf} develop alternative estimator
families without and with cross-fitting, respectively,
for settings in which population auxiliary information
is available only in aggregated form, for example because
of confidentiality or data-access restrictions. Specifically, suppose that the survey sample
$\mathcal{S}_d=\{(Y_{dj},\mathbf{X}_{dj}):j\in\mathcal{Q}^s_d\}$ is available
as in Section~\ref{sec:salearners}, but that the auxiliary information on the
non-sampled units is restricted to the area covariate mean, that is
$\bar{\mathcal R}_d=\{\bar{\mathbf Z}_d\}$ in place of
$\mathcal{R}_d=\{\mathbf{X}_{dj}:j\in\mathcal{Q}_d^r\}$, where
$\bar{\mathbf Z}_d:=N_d^{-1}\sum_{j\in\mathcal Q_d}\mathbf Z_{dj}$. Since the
contextual covariates are constant within the area,
$\mathbf U_{dj}=\mathbf u_d$ and 
$\bar{\mathbf Z}_d=(\bar{\mathbf W}_d^\top,\mathbf u_d^\top)^\top$, so that
aggregation affects the unit-level block alone. The corresponding point in
$\mathcal X$ is $\bar{\mathbf X}_d=(\bar{\mathbf Z}_d^\top,d)^\top$. We use
the abbreviations $g_d^{\mathrm{agg}}:=g(\bar{\mathbf Z}_d,d)$ and
$\hat g_{n,d}^{\mathrm{agg}}:=\hat g_n(\bar{\mathbf Z}_d,d)
=\hat g_{n,d}(\bar{\mathbf Z}_d)$ for evaluation at the area covariate mean.
The prediction functions and their limits are understood to be defined on the
convex hull of the area covariate support, including the covariate means at
which they are evaluated.

\subsection{Learners without cross-fitting}
\label{sec:agg_no_cf}

We first present approximate small area
learners which do not use sample splitting.

\begin{definition}[Aggregated small area learner]
\label{def:SAL_agg}
For a correction weight $a\in[0,1]$, define
\begin{equation}
\label{eq:general_predictor_a}
\hat\mu_d^{\mathrm{agg}}(a)
=
\hat g^{\mathrm{agg}}_{n,d}
+
aP_{n_d}(Y-\hat g_{n,d}).
\end{equation}
\end{definition}

The aggregated family in Definition~\ref{def:SAL_agg} replaces $P_{N_d}\hat g_{n,d}$ of Definition~\ref{def:SAL_family}
by $\hat g^{\mathrm{agg}}_{n,d}$, while
retaining the same weighted residual correction.
In particular, $a=1$ gives the 
aggregated version of the fully corrected learner in Example~\ref{ex:SSE}. Aggregation may shift the probability limit of the fully corrected estimator away from the target $\mu_d$. Let
$\bar{\mathbf z}_d:=E_{P_d^{\mathbf Z}}(\mathbf Z)$
and define
\begin{equation}
\label{eq:mu_agg}
\mu_d^{\mathrm{agg}}
:=
g_d^*(\bar{\mathbf z}_d)
+
E_{P_d}\{Y-g_d^*(\mathbf Z)\},
\qquad
b_d
:=
\mu_d^{\mathrm{agg}}-\mu_d
=
g_d^*(\bar{\mathbf z}_d)
-
E_{P_d^{\mathbf Z}}\{g_d^*(\mathbf Z)\},
\end{equation}
the aggregated version of $\mu_d$ in \eqref{eq:expected_decompose}. The signed aggregation bias $b_d$ need not vanish
with increasing sample size. It is zero when
$g_d^*$ is affine on the convex hull of the area
covariate support, although it may also vanish
for particular covariate distributions when
$g_d^*$ is nonlinear. When the area distribution
varies with $n$, we denote the corresponding
population aggregation bias by $b_{d,n}$. For a general correction weight, the estimator
also inherits the limiting residual-weight bias
$(a-1)\beta_d$ of the original family.
Consequently, its total asymptotic bias for
$\mu_d$ is $b_d+(a-1)\beta_d$.

To establish asymptotic results, note that the
mean-square consistency required by
Assumption~\ref{assump:learner} does not control
the learner's evaluation at $\bar{\mathbf z}_d$.
We therefore introduce an additional assumption
controlling the aggregation contrast and the
effect of replacing the superpopulation covariate
mean with its finite-population counterpart.

\begin{assumption}[Aggregation stability without sample splitting]
\label{assump:pointwise}
For the area $d$ under consideration, suppose that: (i)
$\sqrt{n_d}
\{\hat g_{n,d}(\bar{\mathbf z}_d)
-g_d^*(\bar{\mathbf z}_d)
-P_d(\hat g_{n,d}-g_d^*)\}=o_p(1)$;
(ii)
there exist a neighbourhood $U_d$ of
$\bar{\mathbf z}_d$ and a constant
$\kappa_d<\infty$ such that, with probability
tending to one, $\hat g_{n,d}$ is
$\kappa_d$-Lipschitz on $U_d$;
(iii)
$E_{P_d^{\mathbf Z}}\|\mathbf Z\|^2<\infty$.
\end{assumption}

Condition (i) holds identically when the estimation error
$\hat g_{n,d}-g_d^*$ is affine. More generally, it requires the aggregation
contrast of this error to be $o_p(n_d^{-1/2})$. For a sufficiently regular
parametric learner with contrast error $O_p(n^{-1/2})$, this holds when
$n_d/n\to0$. Conditions (ii) and (iii) connect the superpopulation covariate
mean $\bar{\mathbf z}_d$, at which (i) is stated, with the area covariate
mean $\bar{\mathbf Z}_d$, at which the estimator is evaluated: the second
moment in (iii) gives $\bar{\mathbf Z}_d-\bar{\mathbf z}_d=O_p(N_d^{-1/2})$,
and the local Lipschitz property in (ii) transfers that order to the fitted
values. Condition (ii) holds for locally Lipschitz learners with suitably
bounded parameters, but is not automatic for nonparametric learners, such as
tree-based methods, whose fitted functions may be discontinuous near the area
covariate mean. Proposition~\ref{prop:normality_agg} shows that aggregation introduces the
additional bias term $b_d$, distinct from the residual-weight bias
$(a-1)\beta_d$. To characterise when this additional bias is small,
Lemma~\ref{lemma:aggregation_bias} relates $b_d$ to the within-area covariate
dispersion and the regularity of $g_d^*$.

\begin{proposition}[Asymptotic normality under aggregation]
\label{prop:normality_agg}
Let $\hat\mu_d^{\mathrm{agg}}(a_n)$ be the estimator in
\eqref{eq:general_predictor_a}, evaluated at a
deterministic sequence $a_n\in[0,1]$ with
$a_n\to a\in[0,1]$. Suppose that
$\sigma_d^2>0$, $E_{P_d}(Y^2)<\infty$, and
Assumptions~\ref{assump:design},
\ref{assump:iid}, \ref{assump:entropy_growth},
\ref{assump:regime}, \ref{assump:learner}, and
\ref{assump:pointwise} hold. Then:
\begin{enumerate}
\item[(i)]
$\hat\mu_d^{\mathrm{agg}}(a_n)
-\mu_d^{\mathrm{agg}}
\xrightarrow{p}(a-1)\beta_d$;

\item[(ii)]
if $a_n\equiv1$, then $\sqrt{n_d}
\{\hat\mu_d^{\mathrm{agg}}(1)
-\mu_d^{\mathrm{agg}}\}
\xrightarrow{d}N(0,\sigma_d^2)$;
\item[(iii)]
if $a\in(0,1]$, then under
Assumption~\ref{assump:learner_mean},
\begin{equation*}
\label{eq:clt_agg}
\sqrt{n_d}
\bigl\{
\hat\mu_d^{\mathrm{agg}}(a_n)
-\mu_d^{\mathrm{agg}}
-(a_n-1)\beta_d
\bigr\}
\xrightarrow{d}N(0,a^2\sigma_d^2);
\end{equation*}

\item[(iv)]
if $a_n\to0$ and, instead of
Assumption~\ref{assump:learner_mean},
$P_d(\hat g_{n,d}-g_d^*)
=O_p(n^{-\gamma})$
for $\gamma\in(0,1/2]$, then $\hat\mu_d^{\mathrm{agg}}(a_n)
-\mu_d^{\mathrm{agg}}
-(a_n-1)\beta_d
=
O_p\left(
a_n n_d^{-1/2}
+N_d^{-1/2}
+n^{-\gamma}
\right)
+o_p(n_d^{-1/2})$.
%\begin{equation*}
%\label{eq:rate_agg_synthetic}
%\hat\mu_d^{\mathrm{agg}}(a_n)
%-\mu_d^{\mathrm{agg}}
%-(a_n-1)\beta_d
%=
%O_p\left(
%a_n n_d^{-1/2}
%+N_d^{-1/2}
%+n^{-\gamma}
%\right)
%+o_p(n_d^{-1/2}).
%\end{equation*}
\end{enumerate}
Consequently,
$\hat\mu_d^{\mathrm{agg}}(a_n)-\mu_d
\xrightarrow{p}b_d+(a-1)\beta_d$.
\end{proposition}

\begin{lemma}[Aggregation-bias bounds]
\label{lemma:aggregation_bias}
Let $\mathcal C_d$ be the convex hull of the support of
$P_d^{\mathbf Z}$, and suppose
$E_{P_d^{\mathbf Z}}\|\mathbf Z\|^2<\infty$.
Define
$\Sigma_{Z,d}:=\operatorname{Var}_{P_d^{\mathbf Z}}(\mathbf Z)$.
Then the aggregation bias $b_d$ in \eqref{eq:mu_agg} satisfies the following properties:
(i) if $g_d^*$ is $\kappa_d^*$-Lipschitz on $\mathcal C_d$, then
$|b_d|\leq\kappa_d^*\{\operatorname{tr}(\Sigma_{Z,d})\}^{1/2}$;
(ii) if $g_d^*$ is twice continuously differentiable on
$\mathcal C_d$ and
$\|\nabla^2g_d^*(\mathbf z)\|\leq H_d$ for every $\mathbf z\in\mathcal C_d$,
where $\|\cdot\|$ denotes the operator norm, then
$|b_d|\leq(H_d/2)\operatorname{tr}(\Sigma_{Z,d})$;
(iii) if $g_d^*$ is convex on $\mathcal C_d$, then $b_d\leq0$, and if
it is concave on $\mathcal C_d$, then $b_d\geq0$.
\end{lemma}

For piecewise constant learners, bounds on $|b_d|$ analogous to those in Lemma~\ref{lemma:aggregation_bias} require additional conditions on the partition. If $C$ denotes the cell containing $\bar{\mathbf z}_d$, $P_d^{\mathbf Z}(C)\geq1-\epsilon$ and $|g_d^*|\leq B$, then $|b_d|\leq2B\epsilon$, so that concentration in that cell is sufficient for a small aggregation bias, although cancellation across cells may also occur.

\subsection{Learners with cross-fitting}
\label{sec:agg_cf}

Cross-fitting can also be applied when only
aggregated population covariates are available.
The population covariate mean is computed once,
while the learner is trained separately across folds.

\begin{definition}[Aggregated cross-fitted small area learner]
\label{def:SAL_agg_cf}
Let $\hat g_{n,d}^{(-\ell)}$,
$\bar g_{n,d}$ and
$\hat g_{n,d}^{\mathrm{oof}}$ be as in
Definition~\ref{def:SAL_cf_family}, and let $\bar g_{n,d}^{\mathrm{agg}} = \bar g_{n,d}(\bar{\mathbf Z}_d)$. For $a\in[0,1]$, define
\begin{equation}
\label{eq:cf_predictor_a}
\hat\mu_d^{cf,\mathrm{agg}}(a)
=
\bar g_{n,d}^{\mathrm{agg}}
+
aP_{n_d}(Y-\hat g_{n,d}^{\mathrm{oof}}).
\end{equation}
\end{definition}

The population target $\mu_d^{\mathrm{agg}}$
and aggregation bias $b_d$ remain those in
\eqref{eq:mu_agg}. As without cross-fitting,
additional control is needed at the area
covariate mean.

\begin{assumption}[Aggregation stability under cross-fitting]
\label{assump:cf_pointwise}
For the area $d$ under consideration: (i)
$\sqrt{n_d}
\{\bar g_{n,d}(\bar{\mathbf z}_d)
-g_d^*(\bar{\mathbf z}_d)
-P_d(\bar g_{n,d}-g_d^*)\}
=o_p(1)$; (ii)
there exist a neighbourhood $U_d$ of
$\bar{\mathbf z}_d$ and a constant
$\kappa_d<\infty$ such that, with probability
tending to one, $\bar g_{n,d}$ is
$\kappa_d$-Lipschitz on $U_d$; (iii)
$E_{P_d^{\mathbf Z}}\|\mathbf Z\|^2<\infty$.
\end{assumption}

Condition (i) controls the aggregation contrast
of the fold-averaged estimation error, while
condition (ii) controls evaluation at the observed
population covariate mean. The latter holds,
in particular, when every fold-specific learner
is locally Lipschitz with the same bound.

\begin{proposition}[Asymptotic normality under cross-fitted aggregation]
\label{prop:normality_cf_agg}
Let $\hat\mu_d^{cf,\mathrm{agg}}(a_n)$ be the estimator in
\eqref{eq:cf_predictor_a}, evaluated at a
deterministic sequence $a_n\in[0,1]$ with
$a_n\to a\in[0,1]$. Suppose that
$\sigma_d^2>0$, $E_{P_d}(Y^2)<\infty$, and
Assumptions~\ref{assump:design},
\ref{assump:iid}, \ref{assump:regime},
and \ref{assump:cf_pointwise} hold.
Suppose also that all conditions of
Assumption~\ref{assump:cf_learner} hold except for
$n_d^{3/2}/N_d\to0$, which is not required here. Then:
\begin{enumerate}
\item[(i)]
$\hat\mu_d^{cf,\mathrm{agg}}(a_n)
-\mu_d^{\mathrm{agg}}
\xrightarrow{p}(a-1)\beta_d$;
\item[(ii)]
if $a_n\equiv1$, then
$\sqrt{n_d}
\bigl\{
\hat\mu_d^{cf,\mathrm{agg}}(1)
-\mu_d^{\mathrm{agg}}
\bigr\}
\xrightarrow{d}N(0,\sigma_d^2)$;
\item[(iii)]
if $a\in(0,1]$, then under
Assumption~\ref{assump:cf_learner_mean},
\[
\sqrt{n_d}
\bigl\{
\hat\mu_d^{cf,\mathrm{agg}}(a_n)
-\mu_d^{\mathrm{agg}}
-(a_n-1)\beta_d
\bigr\}
\xrightarrow{d}N(0,a^2\sigma_d^2);
\]

\item[(iv)]
if $a_n\to0$ and, instead of
Assumption~\ref{assump:cf_learner_mean}, $P_d(\bar g_{n,d}-g_d^*)
=O_p(n^{-\gamma})$
for  $\gamma\in(0,1/2]$, then
$\hat\mu_d^{cf,\mathrm{agg}}(a_n)
-\mu_d^{\mathrm{agg}}
-(a_n-1)\beta_d
=
O_p\left(
a_n n_d^{-1/2}
+N_d^{-1/2}
+n^{-\gamma}
\right)
+o_p(n_d^{-1/2})$.
\end{enumerate}
Hence,
$\hat\mu_d^{cf,\mathrm{agg}}(a_n)-\mu_d
\xrightarrow{p}b_d+(a-1)\beta_d$,
and the bounds in
Lemma~\ref{lemma:aggregation_bias}
are unchanged.
\end{proposition}

Unlike the cross-fitted estimator with unit-level population covariates, the
aggregated family does not require Lemma~\ref{lemma:cf_population_averaging},
and hence not the sampling-fraction condition of
Assumption~\ref{assump:cf_learner}, because the fold-averaged learner is
evaluated at a single point rather than averaged over all population units.

\begin{corollary}[Finite-population mean under aggregation]
\label{cor:finite_pop_agg}
Let $\hat\mu_d^{(\cdot),\mathrm{agg}}(a_n)$
denote either $\hat\mu_d^{\mathrm{agg}}(a_n)$ in
Definition~\ref{def:SAL_agg} or
$\hat\mu_d^{cf,\mathrm{agg}}(a_n)$ in
Definition~\ref{def:SAL_agg_cf}.
Under the assumptions of the corresponding
proposition, its conclusions remain valid with
$\bar Y_d+b_d$ in place of
$\mu_d^{\mathrm{agg}}$. In particular,
\[
\hat\mu_d^{(\cdot),\mathrm{agg}}(a_n)-\bar Y_d
\xrightarrow{p}b_d+(a-1)\beta_d.
\]
For $a\in(0,1]$, under the respective conditions
of parts (ii) and (iii),
\begin{equation}
\label{eq:clt_finite_agg}
\sqrt{n_d}
\bigl\{
\hat\mu_d^{(\cdot),\mathrm{agg}}(a_n)
-\bar Y_d-b_d-(a_n-1)\beta_d
\bigr\}
\xrightarrow{d}N(0,a^2\sigma_d^2).
\end{equation}
The rate conclusion in part (iv) likewise
continues to hold with
$\mu_d^{\mathrm{agg}}$ replaced by
$\bar Y_d+b_d$.
\end{corollary}

%\textcolor{blue}{I suggest to substitute the following sentence with 'Thus, consistency for $\bar Y_d$ requires the combined aggregation andresidual-weight bias $b_d+(a-1)\beta_d$ to vanish. The fully correctedinterval targets $\mu_d^{\mathrm{agg}}$, rather than $\mu_d$ when$b_d\neq0$. Moreover, subtracting the natural plug-in estimate
%$\hat b_d=\hat g_{n,d}^{\mathrm{agg}}-P_{n_d}\hat g_{n,d}$ from$\hat\mu_d^{\mathrm{agg}}(1)$ yields the direct estimator $P_{n_d}Y$exactly; the analogous identity holds under cross-fitting. Confidenceintervals and further implications are given in Section~1.3.1 of theSupplementary Material \citep{RelugaSAL_SM}.'}

Thus, consistency for $\bar Y_d$ requires the combined aggregation and residual-weight bias $b_d+(a-1)\beta_d$ to vanish. Since the construction of the confidence intervals for $\mu_d^{\mathrm{agg}}$ is analogous to the construction developed  for $\mu_d$  in Corollary~\ref{cor:variance}, the statement and the proof are deferred to Section~1.3.1 of the
Supplementary Material \citep{RelugaSAL_SM}.

\section{Discussion}\label{sec:discussion}
We discuss several aspects of the proposed estimators that
complement the preceding theoretical results.  Other points, such as discussion on the hierarchical
correction weights, details on low complexity class of learners and the sampling design can be found in Section~2 of the
Supplementary Material \citep{RelugaSAL_SM}, whereas Section~3 therein provides  a practical procedure for the selection of learners in empirical studies.

\emph{On area-level small area learners}.
Suppose that only area-level data are available. For each area $d$, let $\mathbf u_d$ denote the available auxiliary variables and $Y_d$ a direct estimator defined in \eqref{eq:Hajek_HT}, with known sampling variance $V_d$. Given the pooled data $\mathcal S=\{(Y_d,\mathbf u_d):d\in\mathcal D\}$, a natural area-level analogue of \eqref{eq:sal_family} is
$\hat\mu_d(a)=aY_d+(1-a)\hat g_D(\mathbf u_d)$,
where $\hat g_D$ is a shared regression learner fitted across areas. The choices $a=1$ and $a=0$ yield the direct and synthetic estimators, respectively, while intermediate weights combine direct estimation with predictions learned across areas. Using the Fay--Herriot shrinkage factor from Section~\ref{sec:fh}, we obtain the flexible estimator
$\hat\mu_d^{sal,fh}=\hat\Gamma_dY_d+(1-\hat\Gamma_d)\hat g_D(\mathbf u_d)$.
The conventional Fay--Herriot EBLUP is recovered when the regression component and shrinkage parameters are estimated under the standard linear Fay--Herriot model.

Despite the common algebraic form, this area-level construction is not covered by the unit-level asymptotic theory developed above. With only one direct estimate per area, learning the shared regression component naturally relies on increasing the number of areas $D$, rather than on increasing within-area sample sizes $n_d$. When the sampling variances $V_d$ remain nonvanishing, area-specific prediction error need not vanish as $D\to\infty$. Inference therefore requires a separate prediction mean squared error analysis accounting for estimation of the shared learner and shrinkage parameters. Covariate-powered empirical Bayes \citep{ignatiadis2019covariate} combines flexible prediction and shrinkage in a related setting, but developing an inferential theory for the proposed area-level extension is beyond the scope of this paper.

%\subsection{On augmented estimation and Neyman orthogonality}
%\label{sec:causal}
%\label{sec:orthogonality}
\emph{On augmented estimation and Neyman orthogonality}.
The fully corrected SAL estimator resembles the augmented
regression-imputation estimators used in semiparametric
inference. With $D=2$, the areas may be viewed as two observed groups, and $\hat\mu_d^{sal}$ combines the population average of fitted outcomes with the mean residual of the group's sampled units. Unlike causal treatment-effect estimation, however, the target here is a descriptive group mean, and no treatment-assignment or overlap assumptions are required.
The connection is made precise by the population
moment corresponding to \eqref{eq:sal_family}, $U_d(\theta,g)=P_d^{\mathbf Z}g_d +aP_d(Y-g_d)-\theta$.
%\begin{equation}
%\label{eq:moment_family}
%U_d(\theta,g)
%=
%P_d^{\mathbf Z}g_d
%+aP_d(Y-g_d)-\theta.
%\end{equation}
For $a=1$, $U_d(\mu_d,g)=0$ for every integrable
$g$. Moreover, its Gateaux derivative is $D_gU_d(\theta,g)[h]=(1-a)P_d^{\mathbf Z}h_d$,
%\begin{equation}
%\label{eq:gateaux}
%D_gU_d(\theta,g)[h]
%=
%(1-a)P_d^{\mathbf Z}h_d,
%\end{equation}
which vanishes in every direction precisely when
$a=1$. Thus, the fully corrected member is Neyman
orthogonal at every $g$, not merely at $g^*$
\citep{chernozhukov2018}. This explains why its
asymptotic normality does not require
Assumption~\ref{assump:learner_mean}:
the learner's area mean error cancels exactly.
For fixed weights $a\neq1$, that error generally
enters at first order.
%Without sample splitting, the required control of the remaining empirical-process term follows from Definition~\ref{def:low_complexity} together with Assumptions~\ref{assump:entropy_growth} and\ref{assump:learner}; cross-fittingpermits more flexible learners under the correspondingconvergence conditions. 
The same residual-correction principle underlies model-assisted generalized difference estimation, as illustrated in
Example~\ref{ex:SSE}, and is related to recent
work on model-assisted estimation with machine
learning for population-level means \citep{an2026agnostic}.

%\subsection{On borrowing strength and the $\sqrt{n_d}$ rate}\label{rem:rate}

\emph{On borrowing strength and the $\sqrt{n_d}$ rate}. The $\sqrt{n_d}$ normalisation in
Propositions~\ref{prop:normality_family} and
\ref{prop:normality_cf_family} is consistent
with classical SAE theory. In particular,
\cite{prasad1990estimation} establish the
second-order MSE decomposition of the EBLUP
under NERM,
$\operatorname{MSE}(\hat\mu_d^{eblup})
=g_{1d}+g_{2d}+g_{3d}+o(D^{-1})$.
Here $g_{1d}$ is the leading prediction-error
component, of order $O(n_d^{-1})$ when $n_d$
grows, while $g_{2d}$ and $g_{3d}$ account
for estimating the regression parameters and
variance components, respectively, and are
of order $O(D^{-1})$ under the classical
large-$D$ regularity conditions. An analogous separation of error sources
holds for the fully corrected SAL estimator.
The proof of Proposition~\ref{prop:normality_family} gives
$\hat\mu_d^{sal}-\mu_d
=(P_{n_d}-P_d)(Y-g_d^*)
+R_{n,d}^{\mathrm{learn}}+R_{N,d}^{\mathrm{pop}}$,
where the first term captures area residual variation,
$R_{n,d}^{\mathrm{learn}}=o_p(n_d^{-1/2})$
is the nuisance-learning error, and
$R_{N,d}^{\mathrm{pop}}=O_p(N_d^{-1/2})$
is the finite-population averaging error.
The leading term has variance
$\sigma_d^2/n_d$, while the remaining contributions are asymptotically negligible
under the stated assumptions. Thus, as in the Prasad--Rao decomposition, area-specific variation determines the leading error and pooled estimation contributes through additional terms, although the two decompositions concern different asymptotic regimes. Borrowing strength in SAL therefore acts primarily through $\sigma_d^2=\operatorname{Var}_{P_d}\{Y-g_d^*(\mathbf Z)\}$ rather than through the first-order rate: pooling can improve the prediction function and thereby reduce this variance, which attains its minimum at $g_d^*=m_d$. The learner convergence and complexity conditions control the remaining estimation error. Cross-fitting admits more flexible learners without the entropy restriction while retaining area-wise consistency. Pooling could also help to improve the estimation of variance, for example by modeling $q(\mathbf z,d)=E_{P_d}[\{Y-g_d^*(\mathbf Z)\}^2\mid\mathbf Z=\mathbf z]$ across areas and using $\sigma_d^2=P_d^{\mathbf Z}q(\cdot,d)-\beta_d^2$. The latter may stabilize area-specific variance estimates, but its theory depends on the across-area structure and is left for future work.

\emph{On uniform statements}.
The results in 
Sections~\ref{sec:sample_no_split}-\ref{sec:sal_cf} and \ref{sec:agg_no_cf}-\ref{sec:agg_cf} are area-wise:
they concern a fixed area $d$ along the sequence
of populations. Uniform results over $\mathcal D_n$,
such as bounds on
$\max_d|\hat\mu_d-\mu_d|$, simultaneous confidence
intervals, or multiple testing procedures, would
require uniform versions of the learner assumptions,
suitable moment and nondegeneracy conditions,
and rates accounting for the growth of $D_n$.
In particular, they require sufficient growth
of the smallest area sample size and uniform
control of learner errors and empirical-process
or cross-fitted remainders. Uniform bias control
is additionally needed for partially corrected
estimators, while simultaneous inference may
require accounting for dependence across areas.
We leave these extensions for future work.

\section{Simulation study}
\label{sec:simulations}
To test the empirical performance of methods introduced in Sections~\ref{sec:sample_no_split}-\ref{sec:sal_cf} and \ref{sec:agg_no_cf}-\ref{sec:agg_cf}, we carried out extensive simulation studies reflecting the framework in Sections~\ref{sec:super_pop}--\ref{sec:foundations}. Below we summarise the design and report selected results. The data-generating mechanisms, details on the simulation design, and complementary results are given in Sections~4-5
of the Supplementary Material \citep{RelugaSAL_SM}.

\emph{Design}.
\label{sec:sim_design}
Following the superpopulation framework of Sections~\ref{sec:super_pop}--\ref{sec:foundations}, each Monte Carlo replication generates a finite population from the area laws $P_d$ and draws a sample from it. All estimators are scored against the realized finite-population means $\bar y_d$. There are $p=10$ unit-level and $q=2$ contextual covariates, generated from Gaussian latent variables with scenario-specific transformations. Outcomes follow $Y_{dj}=m_d(\mathbf Z_{dj})+e_{dj}$ with $m_d=m_0+\delta_d$ with $m_d$ introduced in Definition~\ref{def:constraint_generic}. Within each area, $n_d$ of $N_d=1000$ units are drawn by simple random sampling, a design that is noninformative given the covariates and self-weighting within areas following Assumption~\ref{assump:design}. The five scenarios are specified in Section 5.1 of the Supplementary Material \citep{RelugaSAL_SM}. Section~\ref{sec:simulations} reports results for S2, which combines a smooth nonlinear regression function, nonlinear covariate dependence, and deviations with both location and functional components. Each scenario is crossed with $D\in\{25,50\}$ and $n_d\in\{10,20\}$, giving twenty setups with $500$ replications each. We compare three fitting schemes: no sample splitting (Definitions~\ref{def:SAL_family} and~\ref{def:SAL_agg}), and $L=5$-fold cross-fitting (Definitions~\ref{def:SAL_cf_family} and~\ref{def:SAL_agg_cf}) with stratification, i.e., folds formed within areas (``cf\_s'') or pooled across areas (``cf\_p''). The cross-fitting propositions assume stratified folds; pooled folds are included as an empirical comparison. The nuisance function is learned on $\mathbf x=(\mathbf z^\top,d)^\top$, so the area indicator is a feature, using a library of parametric, nonparametric and machine-learning methods described in Section 5.3 of the Supplementary Material \citep{RelugaSAL_SM}. We also report the direct Horvitz-Thompson estimator~\eqref{eq:Hajek_HT}.

% ============================================================
% Simulation results -- revised
% Section references in words are marked  % REF: .
% ============================================================

\emph{Results}.
Table~\ref{tab:top3_S2} reports, for each $(D,n_d)$ cell of Scenario~S2, the three combinations of learner, fitting scheme and estimator with the smallest RMSE, the best fully corrected estimator defined in Example~\ref{ex:SSE} when it is not among them, and the direct estimator. Results for the remaining scenarios are reported in Section 5.5 in the Supplementary Material~\citep{RelugaSAL_SM}.

The best SAL estimators reduce the RMSE of the direct estimator by $43$--$49\%$. The smallest RMSE is attained by the predictor $\hat\mu_d^{sae}$ of Example~\ref{ex:SAE} with cross-fitted flexible BART; pooled and stratified folds differ by at most $0.2\%$. The best fully corrected estimator $\hat\mu_d^{sal}$ in Example~\ref{ex:SSE}, obtained in every cell by the low-complexity spline learner without sample splitting, still reduces the RMSE of the direct estimator by $39$--$43\%$, so borrowing strength pays off even after residual correction. The size of the gain is governed by the residual variance $\sigma_d^2$ of Section~\ref{sec:variance}, that is, by how much within-area variability the covariates explain. The ranking reflects the bias--variance trade-off of Section~\ref{sec:sal_cf}. %For a given learner and fitting scheme, $\hat\mu_d(f_d)$ and $\hat\mu_d(1)$ differ by $(1-f_d)\bar r_{n_d}$, so dropping the residual correction removes its variance contribution but leaves the limiting bias $(f_d-1)\beta_d$. The uncorrected predictor is therefore preferable when $\beta_d^2$ is small relative to $\sigma_d^2/n_d$, so that the optimal weight $a_d^\star=\beta_d^2/(\beta_d^2+\sigma_d^2/n_d)$ is close to zero. This may occur when a flexible learner, given the area indicator, captures most area-level departures. Since the two best configurations use different learners, their $4$--$12\%$ RMSE gap does not isolate the cost of residual correction.

Doubling $n_d$ from $10$ to $20$ reduces the RMSE of the predictor $\hat\mu_d^{sae} = \hat\mu_d(f_d)$ in Example~\ref{ex:SAE}, the best fully corrected estimator $\hat\mu_d^{sal}=\hat\mu_d(1)$ in Example~\ref{ex:SSE} and the direct estimator by $27$--$33\%$, close to the reduction of $1-1/\sqrt2\approx29\%$ implied by $\sqrt{n_d}$  limits of Propositions~\ref{prop:normality_family} and~\ref{prop:normality_cf_family}. The relative gain over the direct estimator is broadly stable. Increasing $D$ from $25$ to $50$ leaves the direct estimator, which uses only the sample of area~$d$, essentially unchanged. It improves $\hat\mu_d(1)$ in Example~\ref{ex:SSE}estimator and the best predictor $\hat\mu_d(f_d)$ in Example~\ref{ex:SAE} by only $1$--$4\%$ and $3$--$8\%$, respectively. This agrees with the theory: learner-estimation error is asymptotically negligible at the root-$n_d$ scale for the fully corrected estimator but may enter at leading order for the uncorrected predictor. The theory does not, however, guarantee monotone gains in $D$.

In Scenario~S2, the intervals of Corollary~\ref{cor:variance} for the fully corrected estimator attain coverage $0.915$--$0.945$, approaching the nominal level as $n_d$ grows; the direct estimator covers $0.941$--$0.947$. Undercoverage at $n_d=10$ may partly reflect downward bias of the in-sample residual variance, which out-of-fold residuals avoid. The naive interval centred at $\hat\mu_d(f_d)$ has half-width $t_{n_d-1,0.975}f_d\hat\sigma_d n_d^{-1/2}$, only $f_d\approx1$--$2\%$ of the fully corrected one, and covers in $2$--$4\%$ of replications. It ignores both the approximation bias and the learner-estimation variability, which may enter at leading order when $a_n=f_d\to0$ (part~(iv) of Propositions~\ref{prop:normality_family} and~\ref{prop:normality_cf_family}). Bias-correcting the predictor $\hat\mu_d(f_d)$ in Example~\ref{ex:SAE} recovers $\hat\mu_d(1)$ and its valid interval. Finally, in the construction of the interval, we replaced the standard normal quantile with the corresponding Student's $t$ quantile to account for finite samples.

The results do not favour either low-complexity learners or cross-fitting, which should be treated as complimentary. Cross-fitting makes flexible learners such as BART and gradient boosting available, and these give the most accurate uncorrected predictors. For the fully corrected estimator, which supports inference, the spline learner without sample splitting combines near-nominal coverage with competitive RMSE.

\begin{table}[htbp]
\centering\small
\caption{Scenario S2: per design cell, the three learner--fitting scheme--estimator combinations with the smallest RMSE, the best fully corrected estimator if not among them, and the direct estimator. Absolute bias and RMSE are $\times100$; gap is the percentage RMSE excess over the cell's best row. Coverage is of nominal $95\%$ $t_{n_d-1}$ intervals; for \texttt{sae\_naive}, see the naive interval in Section~4 in the Supplementary Material \citep{RelugaSAL_SM}.}
\label{tab:top3_S2}
\begin{tabular}{llllrrrr}
\toprule
cell & scheme & learner & estimator & abs.\ bias & RMSE & gap (\%) & coverage \\
\midrule
S2\_D25\_n10 & cf\_p & bart\_flexible & sae\_naive & 2.24 & 35.34 & 0.0 & 0.020 \\
S2\_D25\_n10 & cf\_s & bart\_flexible & sae\_naive & 2.23 & 35.39 & 0.2 & 0.022 \\
S2\_D25\_n10 & cf\_p & bart & sae\_naive & 2.69 & 37.90 & 7.2 & 0.023 \\ \cmidrule(lr){2-8}
S2\_D25\_n10 & sal & spline & sal & 0.99 & 37.97 & 7.5 & 0.920 \\
S2\_D25\_n10 & --- & --- & direct & 2.32 & 62.22 & 76.1 & 0.941 \\ \midrule
S2\_D25\_n20 & cf\_p & bart\_flexible & sae\_naive & 0.83 & 24.92 & 0.0 & 0.038 \\
S2\_D25\_n20 & cf\_s & bart\_flexible & sae\_naive & 0.82 & 24.94 & 0.1 & 0.038 \\
S2\_D25\_n20 & cf\_s & spline & sae\_naive & 1.06 & 25.83 & 3.6 & 0.035 \\ \cmidrule(lr){2-8}
S2\_D25\_n20 & sal & spline & sal & 1.06 & 25.85 & 3.7 & 0.939 \\
S2\_D25\_n20 & --- & --- & direct & 1.95 & 44.25 & 77.6 & 0.947 \\ \midrule
S2\_D50\_n10 & cf\_s & bart\_flexible & sae\_naive & 1.71 & 32.55 & 0.0 & 0.023 \\
S2\_D50\_n10 & cf\_p & bart\_flexible & sae\_naive & 1.69 & 32.57 & 0.1 & 0.023 \\
S2\_D50\_n10 & cf\_s & spline & sae\_naive & 1.27 & 36.48 & 12.1 & 0.019 \\ \cmidrule(lr){2-8}
S2\_D50\_n10 & sal & spline & sal & 1.27 & 36.51 & 12.2 & 0.940 \\
S2\_D50\_n10 & --- & --- & direct & 2.52 & 63.62 & 95.5 & 0.941 \\ \midrule
S2\_D50\_n20 & cf\_s & bart\_flexible & sae\_naive & 0.98 & 23.74 & 0.0 & 0.035 \\
S2\_D50\_n20 & cf\_p & bart\_flexible & sae\_naive & 0.99 & 23.74 & 0.0 & 0.038 \\
S2\_D50\_n20 & cf\_s & spline & sae\_naive & 0.90 & 25.56 & 7.6 & 0.036 \\ \cmidrule(lr){2-8}
S2\_D50\_n20 & sal & spline & sal & 0.89 & 25.57 & 7.7 & 0.945 \\
S2\_D50\_n20 & --- & --- & direct & 1.64 & 44.17 & 86.1 & 0.944 \\
\bottomrule
\end{tabular}
\end{table}

\section{Conclusions and outlook}\label{sec:conclusions}

This manuscript introduced \emph{small area learning} (SAL), a learner-agnostic framework that formulates small area estimation
as a multi-target semi-supervised learning problem rather than tying inference to a particular working model. Borrowing strength is represented through shared and area-specific structure, while the
outcome mechanism is treated as a nuisance component that may be estimated by
parametric, semiparametric, or flexible learning methods. We developed
SAL estimators with and without cross-fitting, established area-specific
asymptotic theory and inference, and extended the framework to settings in
which only aggregated population covariates are available. Classical model-based estimator $\hat\mu_d^{sae}$ in Example~\ref{ex:SAE} and synthetic estimator $\hat\mu_d^{syn}$ in Example~\ref{ex:syn} emerge
as particular cases within this formulation.

The results also make explicit that borrowing strength involves a
bias--variance trade-off. Estimators using only partial residual correction ($\hat\mu_d^{sae}$ in Example~\ref{ex:SAE})
can have smaller MSE than the fully corrected SAL estimator ($\hat\mu_d^{sal}$ in Example~\ref{ex:SSE}) when the
approximation error is small, and this behavior is visible in simulations in Section~\ref{sec:simulations} and Section 5 of the Supplementary Material \citep{RelugaSAL_SM}. The fully corrected estimator in Example~\ref{ex:SSE} instead removes the residual-weight bias and provides inference robust to the misspecification of working model under conditions in Sections~\ref{sec:salearners}-\ref{sec:sal_agg}. On the other hand, the bias of estimators with aggregation introduced in Definition~\ref{def:SAL_agg} remains and must be addressed in the inference on the original area mean $\mu_d$ in \eqref{eq:expected_decompose}. Thus SAL should not be interpreted as
an estimator that necessarily minimizes MSE in every setting, but as a general
construction that separates prediction from inference and remains applicable
across a broad class of working models. In the simulations in Section~\ref{sec:simulations} and Section 5 in the Supplementary Material \citep{RelugaSAL_SM}, the best fully corrected estimators $\mu^{sal}_d$ have moderately higher RMSE than the best uncorrected predictors $\hat\mu_d^{sae}$, while retaining substantial gains over direct estimation and coverage approaching the nominal level across all simulation scenarios.

The framework opens several directions for further research. One is the
development of learning algorithms designed specifically for the multi-area
structure of SAL, rather than importing generic prediction methods. In particular, deriving the learner convergence conditions from explicit assumptions on the shared component and
the constrained deviations remains a direction for future research:
the present inferential theory assumes these conditions directly. A second is to extend the theory beyond means to nonlinear area-level functionals and to develop corresponding learner-selection criteria and inferential
procedures. Further empirical work across different surveys, outcomes, and
forms of auxiliary information is also needed to understand which learners
and correction schemes perform well in different SAE environments. Finally,
alternative asymptotic regimes deserve separate study. These include
many-area empirical-Bayes formulations related to
\cite{ignatiadis2019covariate}, as well as classical
unit-level mixed models with random area effects, where a joint asymptotic
analysis of prediction, shrinkage, variance-component estimation, and
second-order MSE would provide a complementary perspective to the
area-specific first-order theory developed here.

%%%%%%%%%%%%%%%%%%%%%%%%%%%%%%%%%%%%%%%%%%%%%%
%% Funding information, if any,             %%
%% should be provided in the                %%
%% funding section.                         %%
%%%%%%%%%%%%%%%%%%%%%%%%%%%%%%%%%%%%%%%%%%%%%%
\begin{funding}
Katarzyna Reluga is the first and corresponding author of this manuscript. The first author is  was supported by the Swiss National Science Foundation through the project P2GEP2-195898 and by resources of the Center for Targeted Machine Learning, University of California at Berkeley and the HPC at the Humboldt University of Berlin. The second author was supported by the `MAPPE’ project, Programma `PE GRINS -- Growing Resilinet, Inclusive and Sustainable’ (cod. PE0000018 CUP: J33C22002910001)
\end{funding}

%

%%%%%%%%%%%%%%%%%%%%%%%%%%%%%%%%%%%%%%%%%%%%%%%%%%%%%%%%%%%%%
%%                  The Bibliography                       %%
%%                                                         %%
%%  imsart-???.bst  will be used to                        %%
%%  create a .BBL file for submission.                     %%
%%                                                         %%
%%  Note that the displayed Bibliography will not          %%
%%  necessarily be rendered by Latex exactly as specified  %%
%%  in the online Instructions for Authors.                %%
%%                                                         %%
%%  MR numbers will be added by VTeX.                      %%
%%                                                         %%
%%  Use \cite{...} to cite references in text.             %%
%%                                                         %%
%%%%%%%%%%%%%%%%%%%%%%%%%%%%%%%%%%%%%%%%%%%%%%%%%%%%%%%%%%%%%

\newpage

\newpage

\begin{center}
  \Large\bfseries Supplement to ``Small area learning''

\end{center}

This supplement provides additional theoretical and empirical material for ``Small area learning''. Section~\ref{sec:theory} contains proofs of the results in Sections~4--5 of the main manuscript, together with
variance estimation and confidence intervals under aggregation. Section~\ref{sec:discussion} discusses hierarchical correction weights, learner complexity, sampling design, and the scope of area-wise inference.
Section~\ref{sec:choice_learner} describes the learner-selection criterion and its interpretation. Sections~\ref{app:generation} and~\ref{app:design} provide the data-generating mechanisms, simulation design, and complementary results. Section~\ref{sec:data_analysis} presents the application of our new methodology to study income distribution across Italian provinces. 

\setcounter{section}{0}

\section{Proofs of theoretical results and further theoretical developments}\label{sec:theory}

%%%%%%%%%%%%%%%%%%%%%%%%%%%%%%%%%%%%%%%%%%%%%%%%%%%%%%%%%%%%%%%%%%%%%%%%%%%%%%%
% Proofs for Section 4.1
%%%%%%%%%%%%%%%%%%%%%%%%%%%%%%%%%%%%%%%%%%%%%%%%%%%%%%%%%%%%%%%%%%%%%%%%%%%%%%% Intro
%%%%%%%%%%%%%%%%%%%%%%%%%
\subsection{Proofs of results in Section 4 of the main document}
Proposition~4.1 in the manuscript is the main theoretical result of Section~4.1. Its
proof rests on the exact decomposition \eqref{eq:family_decomposition},
which splits $\widehat\mu_d(a_n)-\mu_d$ into the deterministic bias
$(a_n-1)\beta_d$, a centered average over the sample of the residuals
$Y-g^*_d(\mathbf Z)$ scaled by $a_n$, and four remainders. To establish asymptotic properties of each term in the decomposition, we need to provide proofs of three additional lemmas and one corollary. Lemma~\ref{lemma:boundness_g_star} establishes the boundedness of $g^*_d$, Corollary~\ref{cor:sigma_finite} the finiteness of $\sigma_d^2
=\operatorname{Var}_{P_d}(\varepsilon_d^*)$ in \eqref{eq:sigma_d}, Lemma~\ref{lemma:maximal_inequality} establishes the
maximal inequality \eqref{eq:maximal_inequality} for the class of differences
between the members of $\mathcal G_{c,n}$ and $g^*_d$ and Lemma~\ref{lemma:empirical_process} the negligibility of the empirical process terms. 

Recall from the main document that, for each area,
$\varepsilon_d^*:=Y-g_d^*(\mathbf Z)$ denotes
the limiting population residual, with mean and variance given by 
\begin{equation}
\label{eq:sigma_d}
\beta_d:=E_{P_d}(\varepsilon_d^*)
=\mu_d-E_{P_d^{\mathbf Z}}\{g_d^*(\mathbf Z)\},
\qquad
\sigma_d^2:=\operatorname{Var}_{P_d}(\varepsilon_d^*)=E_{P_d}
\bigl[
\{\varepsilon_d^*-E_{P_d}(\varepsilon_d^*)\}^2
\bigr].
\end{equation}
In addition, throughout the proofs, we write
\begin{equation*}\label{eq:add_terms}
\Delta_{n,d}:=\hat g_{n,d}-g^*_d,
\qquad
\mathbb G_m:=\sqrt m\,(P_m-P_d),
\qquad
\mathcal E_n:=\{\hat g_n\in\mathcal G_{c,n}\},
\end{equation*}
where $\Delta_{n,d}$ is the deviation of the learner from its population limit on area $d$, $\mathbb G_m$ is the empirical process based on $m$ observations from $P_d$, evaluated below at the sample $m=n_d$ and the population $m=N_d$ of area $d$, and $\mathcal E_n$ is the event that the pooled learner belongs to the low-complexity class. By the definition of the area law in the main document, $P_d=
\mathcal L\{(Y,\mathbf Z)\mid G=d\}
=
P_d^{Y\mid\mathbf Z}\,P_d^{\mathbf Z}$, $P_d$ acts as $P^{\mathbf Z}_d$ on functions of $\mathbf Z$ alone. By
Definition~2 in the main document, the class $\mathcal G_{c,n}$ is uniformly bounded by $B$ and satisfies the uniform entropy bound with constant $A$ and
exponent $v_n$, and by Assumption~7 in the main document both $\Pr(\mathcal E_n)\to1$ and $r_{n,d}=o_p(1)$.

%%%%%%%%%%%%%%%%%%%%%%%%%%%%%
%%%%%%%%% Auxiliary lemma 1
%%%%%%%%%%%%%%%%%%%%%%%%%%%%%%%
Lemma~\ref{lemma:boundness_g_star} and  Corollary~\ref{cor:sigma_finite} help to establish the finiteness of
$\sigma^2_d$, necessary in the proof of Proposition~4.1 in the main document. The strict positivity of $\sigma^2_d$ is a hypothesis of that proposition.

\begin{lemma}[Boundedness of the population limit]
\label{lemma:boundness_g_star}
Under Definition~2 and Assumption~7 in the main document,
$|g^*_d|\le B$ holds $P^{\mathbf Z}_d$-almost surely.
\end{lemma}

\begin{proof}
Fix $\epsilon>0$ and let
$A_\epsilon:=\{\mathbf z:|g^*_d(\mathbf z)|>B+\epsilon\}$, a deterministic set
since $g^*_d$ is deterministic. On $\mathcal E_n$ we have
$|\hat g_{n,d}|\le B$, hence $|\Delta_{n,d}|>\epsilon$ on $A_\epsilon$ and
therefore
\[
r^2_{n,d}
\;\ge\;
\int_{A_\epsilon}\Delta^2_{n,d}\,\mathrm dP^{\mathbf Z}_d
\;\ge\;
\epsilon^2\,P^{\mathbf Z}_d(A_\epsilon).
\]
If $P^{\mathbf Z}_d(A_\epsilon)>0$, then
$\Pr\{r_{n,d}\ge\epsilon\,P^{\mathbf Z}_d(A_\epsilon)^{1/2}\}
\ge\Pr(\mathcal E_n)\to1$, contradicting $r_{n,d}=o_p(1)$. Hence $P^{\mathbf Z}_d(A_\epsilon)=0$ for every $\epsilon>0$, and
$\{|g^*_d|>B\}=\bigcup_{k\in\mathbb N}A_{1/k}$ is a countable union of null
sets, so that $P^{\mathbf Z}_d(|g^*_d|>B)=0$, which proves $|g^*_d|\le B$.
\end{proof}

%%%%%%%%%%%%%%%%%%%%%%%%%%%%%
%%%%%%%%% Auxiliary corollary
%%%%%%%%%%%%%%%%%%%%%%%%%%%%%%%

\begin{corollary}
\label{cor:sigma_finite}
Let $\varepsilon_d^*:=Y-g^*_d(\mathbf Z)$ and let
$\sigma_d^2=\operatorname{Var}_{P_d}(\varepsilon_d^*)$ be as in
\eqref{eq:sigma_d}. If $|g^*_d|\le B$ holds $P^{\mathbf Z}_d$-almost surely
for some constant $B<\infty$ and $E_{P_d}(Y^2)<\infty$, then
$\sigma^2_d\le E_{P_d}\{(\varepsilon^*_d)^2\}<\infty$. The first condition
holds under Definition~2 and Assumption~7 in the main document by
Lemma~\ref{lemma:boundness_g_star}, and under Assumption~9 in the main
document by hypothesis.
\end{corollary}

\begin{proof}
Since $|g^*_d|\le B$ almost surely, $(\varepsilon^*_d)^2\le2Y^2+2B^2$
almost surely, and the right-hand side is $P_d$ integrable because
$E_{P_d}(Y^2)<\infty$. The variance $\sigma^2_d$ of $\varepsilon^*_d$ does not
exceed $E_{P_d}\{(\varepsilon^*_d)^2\}$ and is therefore finite.
\end{proof}

%%%%%%%%%%%%%%%%%%%%%%%%%%%%%
%%%%%%%%% Auxiliary lemma 2 corrected
%%%%%%%%%%%%%%%%%%%%%%%%%%%%%%%

\begin{lemma}[Maximal inequality for the class of differences]
\label{lemma:maximal_inequality}
Let Definition~2 and Assumption~7 in the main document hold, write
$g_d:=g(\cdot,d)$ and $\mathcal G_{c,n,d}:=\{g_d:\ g\in\mathcal G_{c,n}\}$,
and, for $0<\eta\le B$, let
\begin{equation}
\label{eq:difference_class}
\mathcal H_n(\eta)
:=
\bigl\{
g_d-g^*_d:\ g\in\mathcal G_{c,n},\
\|g_d-g^*_d\|_{P^{\mathbf Z}_d,2}\le\eta
\bigr\}.
\end{equation}
Then $\mathcal H_n(\eta)$ admits the constant envelope $2B$ and satisfies the
uniform entropy bound of Definition~2 in the main document with the same
constant $A$ and exponent $v_n$ relative to that envelope, and for every
$m\in\mathbb N$,
\begin{equation}
\label{eq:maximal_inequality}
E\bigl\|\mathbb G_m\bigr\|_{\mathcal H_n(\eta)}
\lesssim
\eta\sqrt{v_n\log(AB/\eta)}
+
\frac{v_nB}{\sqrt m}\,\log(AB/\eta),
\end{equation}
with a constant depending on neither $n$, $m$ nor $\eta$.
\end{lemma}

\begin{proof}
First, by applying Lemma~\ref{lemma:boundness_g_star}, the bound $|g|\le B$
for $g\in\mathcal G_{c,n}$ in Definition~2 in the main document and the
triangle inequality $|g_d-g^*_d|\le2B$,  we can observe that that $\mathcal H_n(\eta)$
admits the constant envelope $2B$.

Second, let us bound the entropy of the restricted class in Definition 2 in the main document.  For a finitely discrete probability
measure $Q$ on $\mathcal Z$, let $\tilde P:=Q\otimes\delta_d$, a finitely
discrete probability measure on $\mathcal X$. Then
$\|g_d-g'_d\|_{L_2(Q)}=\|g-g'\|_{L_2(\tilde P)}$ for all
$g,g'\in\mathcal G_{c,n}$, so every $L_2(\tilde P)$-cover of
$\mathcal G_{c,n}$ induces an $L_2(Q)$-cover of $\mathcal G_{c,n,d}$ of the
same radius and cardinality, and Definition~2 in the main document yields
\begin{equation}
\label{eq:entropy_restricted}
\sup_Q\log N\bigl(\epsilon B,\mathcal G_{c,n,d},L_2(Q)\bigr)
\le v_n\log(A/\epsilon),
\qquad 0<\epsilon<1,
\end{equation}
the supremum being over finitely discrete $Q$ on $\mathcal Z$. Since
covering numbers are nonincreasing in the radius,
\eqref{eq:entropy_restricted} extends to $\epsilon=1$. Therefore, we have
$N(B,\mathcal G_{c,n,d},L_2(Q))\le N(\epsilon B,\mathcal G_{c,n,d},L_2(Q))
\le(A/\epsilon)^{v_n}$ for every $\epsilon\in(0,1)$, and letting
$\epsilon\uparrow1$ gives the bound $A^{v_n}$.

Third, let us study the entropy of 
 of $\mathcal H_n(\eta)$. The class $\mathcal H_n(\eta)$ is a
subclass of the translate
$\mathcal G_{c,n,d}-g^*_d:=\{g_d-g^*_d:\ g\in\mathcal G_{c,n}\}$. Covering
numbers are invariant under translation of a class by a fixed function,
since translating the centre of each ball by $-g^*_d$ turns a cover of
$\mathcal G_{c,n,d}$ into a cover of $\mathcal G_{c,n,d}-g^*_d$ of the same
radius and cardinality, and they do not increase when the class is replaced
by a subclass or the radius is enlarged
\citep[see Section~19.2 in][]{vandervaart1998}. Fix a finitely discrete $Q$
on $\mathcal Z$ and $\epsilon\in(0,1)$, and set
$\epsilon':=\min(2\epsilon,1)$, so that $B\epsilon'\le2B\epsilon$ and
$\epsilon'\ge\epsilon$. Then
\begin{align*}
N\bigl(2B\epsilon,\mathcal H_n(\eta),L_2(Q)\bigr)
&\le
N\bigl(B\epsilon',\mathcal H_n(\eta),L_2(Q)\bigr)
\le
N\bigl(B\epsilon',\mathcal G_{c,n,d}-g^*_d,L_2(Q)\bigr)\\
&=
N\bigl(B\epsilon',\mathcal G_{c,n,d},L_2(Q)\bigr)
\le
(A/\epsilon')^{v_n}
\le
(A/\epsilon)^{v_n},
\end{align*}
where the first inequality follows from $B\epsilon'\le2B\epsilon$, the
second from $\mathcal H_n(\eta)\subseteq\mathcal G_{c,n,d}-g^*_d$, the
equality from translation invariance, the third inequality from
\eqref{eq:entropy_restricted} and its extension to $\epsilon'=1$, and the
last from $\epsilon'\ge\epsilon$. Since $\|2B\|_{Q,2}=2B$, this is the
uniform entropy bound relative to the envelope of $\mathcal H_n(\eta)$, with
the same constant $A$ and exponent $v_n$. The argument does not depend on
$\eta$, so the same bound holds for $\mathcal H_n(2\eta)$ and for every
subclass of it.

Fourth, by Definition~2 in the main document,
$\mathcal G_{c,n}$ is pointwise measurable, and restriction to
$\mathcal Z\times\{d\}$ and translation by $-g^*_d$ preserve this property,
so there is a countable $\mathcal G^0\subseteq\mathcal G_{c,n,d}-g^*_d$ such
that every $h\in\mathcal G_{c,n,d}-g^*_d$ is the pointwise limit of a
sequence $(h_k)$ in $\mathcal G^0$. Let $h\in\mathcal H_n(\eta)$. Since
$|h_k|\le2B$, dominated convergence gives
$\|h_k\|_{P^{\mathbf Z}_d,2}\to\|h\|_{P^{\mathbf Z}_d,2}\le\eta$ and
$\mathbb G_mh_k\to\mathbb G_mh$, so that
$h_k\in\mathcal H^0_n(2\eta):=\mathcal G^0\cap\mathcal H_n(2\eta)$ for all
large $k$. Hence
\begin{equation}
\label{eq:countable_domination}
\|\mathbb G_m\|_{\mathcal H_n(\eta)}
\le
\|\mathbb G_m\|_{\mathcal H^0_n(2\eta)},
\end{equation}
where the right-hand side is a supremum over a countable class and therefore
measurable.

Finally,  we apply Corollary~5.1 of
\cite{chernozhukov2014gaussian} to the countable class
$\mathcal H^0_n(2\eta)$; the corollary bounds $E\|\mathbb G_m\|_{\mathcal F}$
for a class $\mathcal F$ in terms of its uniform entropy bound, its envelope
and the $L_2(P_d)$ radius of its members. The envelope $2B$ and the uniform
entropy bound with constants $A$ and $v_n$ have just been verified for
$\mathcal H^0_n(2\eta)$. Since $\|h\|_{P^{\mathbf Z}_d,2}\le2\eta$ for
$h\in\mathcal H^0_n(2\eta)$ and $2\eta\le2B=\|2B\|_{P_d,2}$ by $\eta\le B$,
the corollary applies with radius $2\eta$ and envelope $2B$ and yields

\begin{align*}
    &E\|\mathbb G_m\|_{\mathcal H^0_n(2\eta)}
\lesssim
2\eta\sqrt{v_n\log\{A\cdot2B/(2\eta)\}}
+\frac{v_n\,2B}{\sqrt m}\log\{A\cdot2B/(2\eta)\}\\
&=
2\eta\sqrt{v_n\log(AB/\eta)}+\frac{2v_nB}{\sqrt m}\log(AB/\eta).
\end{align*}
The logarithm is at least one because
$AB/\eta\ge A\ge e$. Together with \eqref{eq:countable_domination} this
proves \eqref{eq:maximal_inequality}.
\end{proof}

%%%%%%%%%%%%%%%%%%%%%%%%%%%%
%%% Auxiliary lemma
%%%%%%%%%%%%%%%%%%%%%%%%%%%%

\begin{lemma}[Negligibility of the empirical process terms]
\label{lemma:empirical_process}
Under Assumptions~1, ~2, ~5, ~6 and ~7 in the main document, we have
\[
(P_{n_d}-P_d)\Delta_{n,d}=o_p\bigl(n_d^{-1/2}\bigr),
\qquad
(P_{N_d}-P_d)\Delta_{n,d}=o_p\bigl(N_d^{-1/2}\bigr).
\]
\end{lemma}

\begin{proof}
We first construct a deterministic radius
$\eta_n^0\downarrow0$ such that
$\Pr(r_{n,d}>\eta_n^0)\to0$.
This is needed because
Lemma~\ref{lemma:maximal_inequality} applies to the
deterministic localized class
$\mathcal H_n(\eta_n^0)$, whereas
$\Delta_{n,d}$ is itself random. We use the standard
fact that if $X_n=o_p(1)$, then there exists a
deterministic sequence $\epsilon_n\to0$ such that
$\Pr(X_n>\epsilon_n)\to0$. We now apply this fact in the
two cases $\sup_n v_n<\infty$ and  $\sup_n v_n=\infty$ distinguished by Assumption~7 in the main document.

If $\sup_n v_n<\infty$, apply this fact to $X_n=r_{n,d}$, let $\epsilon_n\to0$
be the resulting sequence and set $\eta_n^0:=\max(\epsilon_n,n_d^{-1})$. If
$\sup_n v_n=\infty$, apply it to $X_n=r_{n,d}\sqrt{v_n\log n_d}$, which is
$o_p(1)$ by Assumption~7 in the main document, let $c_n\to0$ be the resulting
sequence and set $\eta_n^0:=\max\{c_n/\sqrt{v_n\log n_d},\,n_d^{-1}\}$; then
$\Pr(r_{n,d}>\eta_n^0)\le\Pr(X_n>c_n)\to0$, and $\eta_n^0\to0$ because
$v_n\ge1$ and $\log n_d\to\infty$. In either case, $\eta_n^0$ is
deterministic and satisfies
\begin{equation}
\label{eq:radius_properties}
\eta_n^0\to0,
\qquad
\eta_n^0\ge n_d^{-1},
\qquad
\Pr(r_{n,d}>\eta_n^0)\to0.
\end{equation}

The first relation in \eqref{eq:radius_properties} $\eta^0_n\to0$ implies $\eta_n^0\le B$ for all sufficiently large $n$. The second gives $1/\eta_n^0\le n_d$ and hence $AB/\eta_n^0\le ABn_d$, so that
\begin{equation}\label{eq:log_bound}
\log(AB/\eta_n^0)
\le \log n_d+\log(AB)
\le 2\log n_d
\end{equation}
for all sufficiently large $n$, since $n_d\to\infty$.
Thus the conditions required to apply
Lemma~\ref{lemma:maximal_inequality} at
$\eta=\eta_n^0$ are satisfied.

We next verify that the bound in
\eqref{eq:maximal_inequality} vanishes at
$\eta=\eta_n^0$. Consider the second term of \eqref{eq:maximal_inequality} on the right hand side.
Using the logarithmic bound in \eqref{eq:log_bound},
$v_nBm^{-1/2}\log(AB/\eta_n^0)$ is bounded, up to a
constant, by
$v_nBm^{-1/2}\log n_d$, which tends to zero for
$m=n_d$ by Assumption~5 in the main document
and hence also for $m=N_d\ge n_d$.

For the first term \eqref{eq:maximal_inequality} on the right hand side, suppose first that
$\sup_n v_n=:v<\infty$. Then
\[
\eta_n^0
\sqrt{v_n\log(AB/\eta_n^0)}
\le
\sqrt v\,\eta_n^0
\sqrt{\log(AB/\eta_n^0)}
\to0,
\]
because
$\eta\mapsto\eta\sqrt{\log(AB/\eta)}$
tends to zero as %$\eta$ tends to zero from the right 
$\eta\downarrow0$ and
$\eta_n^0\to0$.
If instead $\sup_n v_n=\infty$, then
\[
\eta_n^0\sqrt{v_n\log(AB/\eta_n^0)}
\lesssim
\eta_n^0\sqrt{v_n\log n_d}
\le
\max\bigl\{
c_n,\,
n_d^{-1}\sqrt{v_n\log n_d}
\bigr\}
=o(1),
\]
where the first inequality follows from \eqref{eq:log_bound}, the second from
the definition of $\eta_n^0$ as a maximum, and the final equality because $c_n\to0$ and
Assumption~5 in the main document implies
$v_n\log n_d=o(n_d^{1/2})$.
Hence both terms in the bound of
Lemma~\ref{lemma:maximal_inequality} vanish, and therefore
\begin{equation}
\label{eq:process_negligible}
E\bigl\|\mathbb G_{n_d}\bigr\|_{\mathcal H_n(\eta_n^0)}
\to0,
\qquad
E\bigl\|\mathbb G_{N_d}\bigr\|_{\mathcal H_n(\eta_n^0)}
\to0.
\end{equation}
It remains to transfer these bounds from the deterministic
class $\mathcal H_n(\eta_n^0)$ to the random function
$\Delta_{n,d}$. On the event
$\mathcal E_n\cap\{r_{n,d}\leq\eta_n^0\}$,
one has
$\Delta_{n,d}\in\mathcal H_n(\eta_n^0)$ and hence
\[
\sqrt{n_d}\,
\bigl|(P_{n_d}-P_d)\Delta_{n,d}\bigr|
\leq
\|\mathbb G_{n_d}\|_{\mathcal H_n(\eta_n^0)}.
\]
Therefore, for every $\epsilon>0$,
\[
\Pr\Bigl\{
\sqrt{n_d}\,
\bigl|(P_{n_d}-P_d)\Delta_{n,d}\bigr|>\epsilon
\Bigr\}
\leq
\Pr(\mathcal E_n^c)
+
\Pr(r_{n,d}>\eta_n^0)
+
\Pr\Bigl\{
\|\mathbb G_{n_d}\|_{\mathcal H_n(\eta_n^0)}
>\epsilon
\Bigr\}.
\]
All three terms tend to zero: the first two by
Assumption~7 in the main document and
\eqref{eq:radius_properties}, and the third by
\eqref{eq:process_negligible} and Markov's inequality.
Hence
$(P_{n_d}-P_d)\Delta_{n,d}
=o_p(n_d^{-1/2})$.

The same argument, with $\mathbb G_{N_d}$ and
$\sqrt{N_d}$ in place of $\mathbb G_{n_d}$ and
$\sqrt{n_d}$, gives
$(P_{N_d}-P_d)\Delta_{n,d}
=o_p(N_d^{-1/2})$.
Importantly, this domination uses a deterministic
localized class and therefore does not require
$\Delta_{n,d}$ to be independent of the population
covariates, even though
$\mathcal Q_d^s\subset\mathcal Q_d$.
\end{proof}

%%%%%%%%%%%%%%%%%%%%%%%
% Proposition 4.1
%%%%%%%%%%%%%%%%%%%%%%%%%%%%

\begin{proof}[Proof of Proposition~4.1 in the main document]
We are finally ready to prove
Proposition~4.1 by bounding
each term in the following decomposition.

Let us start from the decomposition of $\widehat\mu_d(a_n)-\mu_d$. Observe that adding and subtracting $P_{N_d}g^*_d$ and $P_{n_d}g^*_d$ on the right-hand side of the expression for the small area learner
$\widehat\mu_d(a)$ in Definition~3 in the main document, centring the resulting averages $P_{N_d}g^*_d$, $P_{N_d}\Delta_{n,d}$,
$P_{n_d}(Y-g^*_d)$ and $P_{n_d}\Delta_{n,d}$ at their $P_d$-expectations, and using $\mu_d=P_dg^*_d+\beta_d$ with $\beta_d=P_d(Y-g^*_d)$ yields the following identity

\begin{align}
\widehat\mu_d(a_n)-\mu_d
&=\underbrace{(a_n-1)\beta_d}_{T_0}
+\underbrace{(P_{N_d}-P_d)g^*_d}_{T_1}
+\underbrace{a_n\,(P_{n_d}-P_d)(Y-g^*_d)}_{T_2}
\nonumber\\
&+\underbrace{(P_{N_d}-P_d)\Delta_{n,d}}_{T_{3,N}}
-\underbrace{a_n\,(P_{n_d}-P_d)\Delta_{n,d}}_{T_{3,n}}
+\underbrace{(1-a_n)\,P_d\Delta_{n,d}}_{T_4}.
\label{eq:family_decomposition}
\end{align}
Conditional on the $\sigma$-field generated by the area-specific
components $\mathcal F_{\mathcal P}$ in Assumption~2 in the main document, the term $T_0$ is fixed; the remaining five terms are treated in subsequent paragraphs.

\emph{The finite-population term $T_1$.}
Conditional on $\mathcal F_{\mathcal P}$,
the function $g_d^*$ is fixed and bounded by
Lemma~\ref{lemma:boundness_g_star}. In addition,
$(\mathbf Z_{dj})_{j\in\mathcal Q_d}$ are i.i.d. $P_d^{\mathbf Z}$ by Assumption~2 in the main document. Hence $T_1$ is a centred average of $N_d$ i.i.d.\
variables whose variance $\operatorname{Var}_{P^{\mathbf Z}_d}\{g^*_d(\mathbf Z)\}$
is finite, being at most $B^2$, so that $T_1=O_p(N_d^{-1/2})$ by Chebyshev's inequality and $\sqrt{n_d}\,T_1=(n_d/N_d)^{1/2}O_p(1)=o_p(1)$ by
Assumption~6(ii) in the main document.

\emph{The empirical process terms $T_{3,n}$ and $T_{3,N}$.}
By Lemma~\ref{lemma:empirical_process},
$(P_{n_d}-P_d)\Delta_{n,d}=o_p(n_d^{-1/2})$.
Since $a_n$ is deterministic and $|a_n|\leq1$,
$T_{3,n}
=
a_n\,o_p(n_d^{-1/2})$,
and hence $\sqrt{n_d}T_{3,n}=o_p(1)$.
The same lemma gives
$T_{3,N}=o_p(N_d^{-1/2})$, and since
$n_d\leq N_d$, it follows that
$\sqrt{n_d}T_{3,N}=o_p(1)$.

\emph{The learner mean term $T_4$.} By the Cauchy--Schwarz inequality,
$|P_d\Delta_{n,d}|\le r_{n,d}=o_p(1)$, so that $T_4=o_p(1)$ for every sequence
$a_n$ in $[0,1]$, since $|1-a_n|\le1$. Three refinements correspond to the cases in
parts (ii)--(iv) of the proposition. If $a_n\equiv1$, then $T_4=0$
identically. If $a \in (0,1)$ under Assumption~8 in the main document,
$\sqrt{n_d}\,T_4=(1-a_n)\sqrt{n_d}\,P_d\Delta_{n,d}=o_p(1)$, since
$|1-a_n|\le1$. If instead $P_d\Delta_{n,d}=O_p(n^{-\gamma})$, then
$T_4=O_p(n^{-\gamma})$, again because $|1-a_n|\le1$.

\emph{The leading term $T_2$.} The sampled pairs
$(Y_{dj},\mathbf Z_{dj})_{j\in\mathcal Q^s_d}$ are i.i.d.\ $P_d$ by
Assumption~1(iii) in the main document, so that the quantity
$(P_{n_d}-P_d)(Y-g^*_d)$, which equals
\[
\frac1{n_d}\sum_{j\in\mathcal Q^s_d}
\bigl\{\varepsilon^*_{dj}-E_{P_d}(\varepsilon^*_d)\bigr\},
\quad
\varepsilon^*_{dj}=Y_{dj}-g^*_d(\mathbf Z_{dj}),
\]
is a centred average of $n_d$ i.i.d.\ variables with variance $\sigma^2_d$,
finite by Corollary~\ref{cor:sigma_finite} and strictly positive by the
hypothesis of Proposition~4.1 in the main document; the centring is by
$E_{P_d}(\varepsilon^*_d)=\beta_d$ in \eqref{eq:sigma_d}, which is not assumed to vanish. By the Lindeberg--L\'evy central limit theorem,
$\sqrt{n_d}(P_{n_d}-P_d)(Y-g_d^*)
\xrightarrow{d}N(0,\sigma_d^2)$, so that
\begin{equation}
\label{eq:T2_order}
T_2=a_nO_p(n_d^{-1/2}),
\quad\text{and}\quad
\sqrt{n_d}\,T_2
\xrightarrow{d}N(0,a^2\sigma_d^2),
\end{equation}
where the second assertion follows by Slutsky's lemma and $a_n\to a$.

\emph{Conclusion.} By the second to fifth steps, the five stochastic
terms in \eqref{eq:family_decomposition} are
$o_p(1)$, using $n_d\to\infty$ for $T_2$. Hence
$\widehat\mu_d(a_n)-\mu_d=(a_n-1)\beta_d+o_p(1)$, which converges in probability to $(a-1)\beta_d$ which proves Proposition~4.1(i) in the main document.

To obtain Proposition~4.1(ii)--(iii) in the main document, the second and third steps give
$\sqrt{n_d}\,(T_1+T_{3,N}-T_{3,n})=o_p(1)$. If $a_n\equiv1$, then
$T_0=T_4=0$ identically by the first and fourth steps as explained above, so that equations
\eqref{eq:family_decomposition}, \eqref{eq:T2_order} with $a=1$ and Slutsky's
lemma yield the limit in (ii). If $a\in(0,1)$, then
Assumption~8 in the main document and the fourth step give
$\sqrt{n_d}\,T_4=o_p(1)$, and \eqref{eq:family_decomposition} reads
$\sqrt{n_d}\,\bigl\{\widehat\mu_d(a_n)-\mu_d-(a_n-1)\beta_d\bigr\}
=
\sqrt{n_d}\,T_2+o_p(1)$, so that \eqref{eq:T2_order} and Slutsky's lemma give the limit in (iii), whose
variance $a^2\sigma^2_d$ is strictly positive.

For Proposition~4.1(iv) in the main document,
we collect the orders obtained above:
$T_1=O_p(N_d^{-1/2})$ by the second step,
$T_2=a_nO_p(n_d^{-1/2})$ by the fifth step,
$T_{3,N}=o_p(N_d^{-1/2})$ and
$T_{3,n}=a_no_p(n_d^{-1/2})$ by the third step,
and $T_4=O_p(n^{-\gamma})$ by the fourth step
under the rate condition imposed there.
Substituting these five orders in
\eqref{eq:family_decomposition}, after removing
$T_0=(a_n-1)\beta_d$, gives the stated rate.
 No
single normalisation is asserted in this case: the limit in
\eqref{eq:T2_order} is degenerate at zero when $a=0$, so $\sqrt{n_d}$ is no
longer the relevant scaling, and which of $a_nn_d^{-1/2}$, $N_d^{-1/2}$ and
$n^{-\gamma}$ dominates is not determined by Assumptions~6 and ~7 in the main document.
\end{proof}

%%%%%%%%%%%%%%%%%%%%%%%%%%%%%%%%%%%%%%%%%%%%%%%%%%%%%%%%%%%%%%%%%%%%%%%%%%%%%%%
% Proofs for Section 4.2
%%%%%%%%%%%%%%%%%%%%%%%%%%%%%%%%%%%%%%%%%%%%%%%%%%%%%%%%%%%%%%%%%%%%%%%%%%%%%%%
\subsection{Proofs for Section~4.2 of the main document}
Proposition~4.2 in the main document is the central theoretical result of Section~4.2 therein. Its proof depends on a decomposition of $\widehat\mu^{cf}_d(a_n)-\mu_d$ in \eqref{eq:cf_family_decomposition} with similar terms as for Proposition~4.1 in the main document. We first establish the proof of Lemma~4.1 in the main document, and afterwards, after introducing some additional notation, we move to discussing each term of the decomposition in \eqref{eq:cf_family_decomposition}.

%%%%%%%%%%%%%%%%%%%%%%%%%%%%%
%%%%%%%% Proof of population averaging lemma
%%%%%%%%%%%%%%%%%%%%%%%%%%%%%%%

\begin{proof}[Proof of Lemma~4.1 in the main document]
Let $h:=\bar g_{n,d}-g_d^*$. Splitting $P_{N_d}$ into its
sampled and nonsampled parts gives the exact decomposition
\[
(P_{N_d}-P_d)h
=
\frac{n_d}{N_d}(P_{n_d}-P_d)h
+
\frac{N_d-n_d}{N_d}(P_{N_d-n_d}-P_d)h,
\]
where $P_{N_d-n_d}$ denotes the empirical measure over the
nonsampled units of area $d$.

Consider first the sampled part. Assumption~9 in the main document 
gives $|\hat g^{(-\ell)}_{n,d}|\le B$ for
$\ell=1,\dots,L$ and $|g_d^*|\le B$, so that
$|h|\le2B$ by the triangle inequality. Consequently,
\[
|(P_{n_d}-P_d)h|
\le P_{n_d}|h|+P_d|h|
\le4B,
\]
and hence $(n_d/N_d)(P_{n_d}-P_d)h
=
O_p\left(n_d/N_d\right)$.
The crude bound is sufficient here. In contrast to the
out-of-fold evaluations considered in the cross-fitted
sample term, $\bar g_{n,d}$ is not independent of the full
area sample: for a unit in fold $\ell$, only
$\hat g^{(-\ell)}_{n,d}$ excludes that unit, whereas the
remaining fold-specific learners may have been trained on it.
Thus the foldwise conditional-independence argument does not
apply directly to $(P_{n_d}-P_d)h$.

Consider next the nonsampled part. Let $\mathcal F_s$ denote the
$\sigma$-field generated by the sampling indicators, the fold allocation and
the sampled pairs of all areas, which together determine $\bar g_{n,d}$.
Conditionally on $\mathcal F_s$, the function $h$ is fixed, while the
nonsampled covariates $(\mathbf Z_{dj})_{j\in\mathcal Q^r_d}$ are i.i.d.\
$P_d^{\mathbf Z}$ and independent of $\mathcal F_s$: the population pairs are
i.i.d.\ within area $d$ and independent across areas by Assumption~2 in the
main document, and the sampling indicators and the fold allocation are
independent of the population values by Assumption~9 in the main document.
Therefore,
$$
E\left[
\left\{(P_{N_d-n_d}-P_d)h\right\}^2
\mid\mathcal F_s
\right]
\le
\frac{\|h\|_{P_d^{\mathbf Z},2}^2}{N_d-n_d}.
$$
Moreover, since
$$
\|h\|_{P_d^{\mathbf Z},2}
=
\left\|
\frac1L\sum_{\ell=1}^L\Delta_\ell
\right\|_{P_d^{\mathbf Z},2}
\le
\frac1L\sum_{\ell=1}^L
\|\Delta_\ell\|_{P_d^{\mathbf Z},2},
$$
Assumption~9 in the main document and the fact that $L$ is fixed imply
$\|h\|_{P_d^{\mathbf Z},2}=o_p(1)$. Since $\|h\|_{P_d^{\mathbf Z},2}$ is
$\mathcal F_s$ measurable, the conditional Chebyshev inequality gives, for
every $\epsilon,\eta>0$,
\[
\Pr\Bigl\{
\sqrt{N_d-n_d}\,\bigl|(P_{N_d-n_d}-P_d)h\bigr|>\epsilon
\Bigr\}
\le
\Pr\bigl\{\|h\|^2_{P_d^{\mathbf Z},2}>\eta\bigr\}
+\frac{\eta}{\epsilon^2}.
\]
For every fixed $\eta>0$ the first term tends to zero; taking the limit and
then letting $\eta\downarrow0$ gives
\[
(P_{N_d-n_d}-P_d)h
=
o_p\bigl((N_d-n_d)^{-1/2}\bigr).
\]
Combining the sampled and nonsampled bounds proves the rate in Lemma~4.1 in the main document.

Finally, multiplying the two contributions by $\sqrt{n_d}$
gives
\[
\sqrt{n_d}\,\frac{n_d}{N_d}
O_p(1)
=
\frac{n_d^{3/2}}{N_d}O_p(1)
=
o_p(1)
\]
for the sampled part, while the nonsampled part satisfies
\[
\sqrt{n_d}\,o_p(N_d^{-1/2})
=
\left(\frac{n_d}{N_d}\right)^{1/2}o_p(1)
=
o_p(1).
\]
Adding the two terms yields the final result.
\end{proof}

%%%%%%%%%%%%%
%%%%%%%% Proof normality cross fitted estimator
%%%%%%%%%%%%%%%%%%%%%%%%%%%%%%%

\begin{proof}[Proof of Proposition~4.2 in the main document]
For $\ell=1,\dots,L$ let $n_{d\ell}:=|\mathcal Q_d^s\cap\mathcal I_\ell|$ and
define the fold empirical measures
\begin{equation*}
P_{n_{d\ell}}h
:=
\frac{1}{n_{d\ell}}
\sum_{j\in\mathcal Q_d^s\cap\mathcal I_\ell}
h(Y_{dj},\mathbf Z_{dj}),\;\text{ so that }\; P_{n_d}
=
\sum_{\ell=1}^L\frac{n_{d\ell}}{n_d}P_{n_{d\ell}}
\quad\text{ and }\quad
\frac{n_{d\ell}}{n_d}\rightarrow\frac{1}{L},
\end{equation*}
where the convergence in the last term follows from the stratification $|n_{d\ell}-n_d/L|\le1$ in
Assumption~9 in the main document and $n_d\to\infty$. Throughout the proof, we will use the following notation:
$$
\Delta_\ell:=\hat g^{(-\ell)}_{n,d}-g^*_d,
\qquad
\bar\Delta_{n,d}:=\bar g_{n,d}-g^*_d=\frac{1}{L}\sum_{\ell=1}^L\Delta_\ell .
$$

\emph{The decomposition.} By adding and subtracting $P_{N_d}g^*_d$ and
$\sum_\ell(n_{d\ell}/n_d)P_{n_{d\ell}}g^*_d=P_{n_d}g^*_d$ in the cross-fitted version of the small area learner $\widehat\mu^{cf}_d(a)$ in Definition~4 in the main document, and using $\mu_d=P_dg^*_d+\beta_d$ with
$\beta_d=P_d(Y-g^*_d)$, we obtain the exact identity
\begin{align}
\widehat\mu^{cf}_d(a_n)-\mu_d
={}&\underbrace{(a_n-1)\beta_d}_{T_0}
+\underbrace{(P_{N_d}-P_d)g^*_d}_{T_1}
+\underbrace{a_n\,(P_{n_d}-P_d)(Y-g^*_d)}_{T_2}
+\underbrace{(P_{N_d}-P_d)\bar\Delta_{n,d}}_{T_{3,N}}
\nonumber\\
&-\underbrace{a_n\sum_{\ell=1}^L\frac{n_{d\ell}}{n_d}
(P_{n_{d\ell}}-P_d)\Delta_\ell}_{T_{3,n}}
+\underbrace{\sum_{\ell=1}^L
\Bigl(\frac{1}{L}-a_n\frac{n_{d\ell}}{n_d}\Bigr)P_d\Delta_\ell}_{T_4}.
\label{eq:cf_family_decomposition}
\end{align}
The term $T_0$ is deterministic; the remaining five are treated in subsequent paragraphs.

\emph{The finite-population term $T_1$.} The function $g^*_d$ is
deterministic and bounded by $B$ by Assumption~9 in the main document, and
$(\mathbf Z_{dj})_{j\in\mathcal Q_d}$ are i.i.d.\ $P_d^{\mathbf Z}$ by
Assumption~2. Hence $T_1$ is a centred average of $N_d$ i.i.d.\ variables whose variance is at most $B^2$, so that
$T_1=O_p(N_d^{-1/2})$ by Chebyshev's inequality and
$\sqrt{n_d}\,T_1=(n_d/N_d)^{1/2}O_p(1)=o_p(1)$ by
Assumption~6(ii) in the manuscript.

\emph{The fold terms $T_{3,n}$.} Fix $\ell$ and let $\mathcal F_{-\ell}$ denote the $\sigma$-field generated by the sampling indicators, the fold allocation and the sampled pairs of all areas outside $\mathcal I_\ell$. Conditionally on $\mathcal F_{-\ell}$, the function $\Delta_\ell$ is fixed, because $\hat g^{(-\ell)}_{n,d}$ is trained on the sampled data outside $\mathcal I_\ell$, while the pairs indexed by $\mathcal Q_d^s\cap\mathcal I_\ell$ are i.i.d.\ $P_d$ and independent of $\mathcal F_{-\ell}$: the sampled pairs of area $d$ are i.i.d.\ $P_d$ by
Assumption~1(iii) in the main document, the collections of distinct areas are independent by Assumption~2 in the main document, and the sampling
indicators and the fold allocation are independent of the observed values
by Assumption~9 in the main document.
Hence
$$
E\Bigl[
\bigl\{(P_{n_{d\ell}}-P_d)\Delta_\ell\bigr\}^2
\,\Bigm|\,\mathcal F_{-\ell}
\Bigr]
\le
\frac{\|\Delta_\ell\|^2_{P_d^{\mathbf Z},2}}{n_{d\ell}} .
$$
Since $\|\Delta_\ell\|_{P_d^{\mathbf Z},2}$ is $\mathcal F_{-\ell}$
measurable, for every $\epsilon,\eta>0$,
\[
\Pr\Bigl\{
\sqrt{n_{d\ell}}\,
\bigl|(P_{n_{d\ell}}-P_d)\Delta_\ell\bigr|>\epsilon
\Bigr\}
\le
\Pr\bigl\{\|\Delta_\ell\|^2_{P_d^{\mathbf Z},2}>\eta\bigr\}
+
\frac{\eta}{\epsilon^2}
\]
by the conditional Chebyshev inequality. For every fixed
$\eta>0$, the first term tends to zero by
Assumption~9 in the main document. Taking the limit and then letting $\eta\downarrow0$ therefore gives
$\sqrt{n_{d\ell}}\,
(P_{n_{d\ell}}-P_d)\Delta_\ell=o_p(1)$.
Since the stratification in Assumption~9 in the main document gives
$n_{d\ell}\ge n_d/L-1\to\infty$, the preceding argument
applies to every fold. Moreover,
\[
\bigl|\sqrt{n_d}\,T_{3,n}\bigr|
\le
a_n\sum_{\ell=1}^L
\left(\frac{n_{d\ell}}{n_d}\right)^{1/2}
\bigl|
\sqrt{n_{d\ell}}\,
(P_{n_{d\ell}}-P_d)\Delta_\ell
\bigr|
\le
a_n\sum_{\ell=1}^L
\bigl|
\sqrt{n_{d\ell}}\,
(P_{n_{d\ell}}-P_d)\Delta_\ell
\bigr|,
\]
where the first inequality follows from the triangle inequality, $a_n\ge0$
and the identity
$\sqrt{n_d}\,(n_{d\ell}/n_d)=(n_{d\ell}/n_d)^{1/2}\sqrt{n_{d\ell}}$, 
and the second from $n_{d\ell}\le n_d$.
Since $L$ is fixed and each term in the final sum is
$o_p(1)$, it follows that
\[
T_{3,n}
=
a_n\,o_p(n_d^{-1/2}),
\]
and hence, in particular,
$\sqrt{n_d}\,T_{3,n}=o_p(1)$ because $0\le a_n\le1$.

\emph{The population term $T_{3,N}$.} By Lemma~4.1 in the main document,
whose conditions are among those of Proposition~4.2 in the main document,
$T_{3,N}=O_p(n_d/N_d+N_d^{-1/2})=o_p(1)$ and $\sqrt{n_d}\,T_{3,N}=o_p(1)$,
the latter using $n_d^{3/2}/N_d\to0$ in Assumption~9 in the main document.

\emph{The learner mean term $T_4$.} Writing
\[
\frac1L-a_n\frac{n_{d\ell}}{n_d}
=
\frac{1-a_n}{L}
+
a_n\left(\frac1L-\frac{n_{d\ell}}{n_d}\right),
\]
we obtain
\begin{equation}
\label{eq:cf_T4_split}
T_4
=
(1-a_n)P_d\bar\Delta_{n,d}
+
a_n\sum_{\ell=1}^L
\left(\frac1L-\frac{n_{d\ell}}{n_d}\right)
P_d\Delta_\ell .
\end{equation}
By the Cauchy--Schwarz inequality and
Assumption~9 in the main document,
$|P_d\Delta_\ell|
\le\|\Delta_\ell\|_{P_d^{\mathbf Z},2}=o_p(1)$
for every $\ell$. Since $L$ is fixed,
$P_d\bar\Delta_{n,d}=o_p(1)$. Moreover, the stratification
in Assumption~9 in the main document gives
$|1/L-n_{d\ell}/n_d|\le n_d^{-1}$, and hence
\[
\left|
a_n\sum_{\ell=1}^L
\left(\frac1L-\frac{n_{d\ell}}{n_d}\right)
P_d\Delta_\ell
\right|
\le
a_nLn_d^{-1}
\max_{1\le\ell\le L}|P_d\Delta_\ell|
=
a_n\,o_p(n_d^{-1}).
\]
It follows that $T_4=o_p(1)$. If $a_n\equiv1$, the first term in
\eqref{eq:cf_T4_split} vanishes identically, while the
second is $o_p(n_d^{-1})$. Thus
$\sqrt{n_d}\,T_4=o_p(1)$ without
Assumption~10 in the main document. Nevertheless, if Assumption~10 in the main document holds, then
$\sqrt{n_d}\,P_d\Delta_\ell=o_p(1)$ for every
$\ell$, and hence
$\sqrt{n_d}\,P_d\bar\Delta_{n,d}=o_p(1)$ because
$L$ is fixed. Together with the bound on the second term
above, this gives $\sqrt{n_d}\,T_4=o_p(1)$. Finally, under the rate condition in part~(iv),
$P_d\bar\Delta_{n,d}=O_p(n^{-\gamma})$, and therefore
\[
T_4
=
(1-a_n)O_p(n^{-\gamma})
+
a_n\,o_p(n_d^{-1}).
\]

\emph{The leading term $T_2$.} The sampled pairs
$(Y_{dj},\mathbf Z_{dj})_{j\in\mathcal Q^s_d}$ are i.i.d.\ $P_d$ by
Assumption~1(iii) in the main document, so that
$(P_{n_d}-P_d)(Y-g_d^*)$ equals
\[
\frac{1}{n_d}\sum_{j\in\mathcal Q_d^s}
\{\varepsilon^*_{dj}-E_{P_d}(\varepsilon^*_d)\},
\qquad
\varepsilon^*_{dj}=Y_{dj}-g_d^*(\mathbf Z_{dj}),
\]
and is therefore a centred average of $n_d$ i.i.d.\ variables with finite
and strictly positive variance $\sigma_d^2$, by
Corollary~\ref{cor:sigma_finite} and the hypothesis of
Proposition~4.2 in the main document. The centring is by
$E_{P_d}(\varepsilon^*_d)=\beta_d$, which is not assumed to vanish. By the
Lindeberg--L\'evy central limit theorem,
\[
\sqrt{n_d}\,(P_{n_d}-P_d)(Y-g_d^*)
\xrightarrow{d}N(0,\sigma_d^2).
\]
Consequently,
$T_2=a_nO_p(n_d^{-1/2})$, and since $a_n\to a$, Slutsky's lemma gives
$\sqrt{n_d}\,T_2\xrightarrow{d}N(0,a^2\sigma_d^2)$.

\emph{Conclusion.} By the preceding steps,
$T_1$, $T_2$  $T_{3,N}$, $T_{3,n}$, and $T_4$ are all $o_p(1)$, using $n_d\to \infty$ for $T_2$.
Hence \eqref{eq:cf_family_decomposition} gives
\[
\widehat\mu^{cf}_d(a_n)-\mu_d
=
(a_n-1)\beta_d+o_p(1),
\]
which converges in probability to $(a-1)\beta_d$ and proves
Proposition~4.2(i) in the main document.

For part~(ii), suppose that $a_n\equiv1$. Then $T_0=0$, while the bounds on $T_{3,N}$ and $T_4$ give
$\sqrt{n_d}(T_{3,N}+T_4)=o_p(1)$ without
Assumption~10 in the main document. Together with bounds on $T_{1}$ and $T_{3,n}$, we have
\[
\sqrt{n_d}\,(T_1+T_{3,N}-T_{3,n}+T_4)=o_p(1).
\]
Since $a=1$, the limit of $\sqrt{n_d}\,T_2$ is 
$\sqrt{n_d}\,T_2\xrightarrow{d}N(0,\sigma_d^2)$. Therefore,
\eqref{eq:cf_family_decomposition} and Slutsky's lemma yield
\[
\sqrt{n_d}\,
(\widehat\mu^{cf}_d-\mu_d)
\xrightarrow{d}N(0,\sigma_d^2),
\]
which proves part~(ii).

For part~(iii), Assumption~10 in the main document and the bound on $T_4$ 
give $\sqrt{n_d}\,T_4=o_p(1)$. Together with bounds on $T_1$, $T_{3, N}$ and $T_{3, n}$ therefore imply
\[
\sqrt{n_d}\,(T_1+T_{3,N}-T_{3,n}+T_4)=o_p(1).
\]
Subtracting $T_0=(a_n-1)\beta_d$ from
\eqref{eq:cf_family_decomposition} gives
\[
\sqrt{n_d}\,
\bigl\{
\widehat\mu^{cf}_d(a_n)-\mu_d-(a_n-1)\beta_d
\bigr\}
=
\sqrt{n_d}\,T_2+o_p(1).
\]
The limit on $\sqrt{n}_d T_2$ and Slutsky's lemma then give
\[
\sqrt{n_d}\,
\bigl\{
\widehat\mu^{cf}_d(a_n)-\mu_d-(a_n-1)\beta_d
\bigr\}
\xrightarrow{d}
N(0,a^2\sigma_d^2),
\]
which proves part~(iii).

Finally, for part~(iv), the bounds obtained above give
$T_1=O_p(N_d^{-1/2})$,
$T_2=a_nO_p(n_d^{-1/2})$,
$T_{3,n}=a_no_p(n_d^{-1/2})$, and
\[
T_4
=
(1-a_n)O_p(n^{-\gamma})
+
a_n\,o_p(n_d^{-1}).
\]
Moreover, Lemma~4.1 in the main document gives
\[
T_{3,N}
=
O_p\left(
\frac{n_d}{N_d}+N_d^{-1/2}
\right).
\]
Since $a_no_p(n_d^{-1})$ is absorbed by
$a_nO_p(n_d^{-1/2})$ and $|1-a_n|\leq 1$, substituting these bounds in
\eqref{eq:cf_family_decomposition} and subtracting
$T_0=(a_n-1)\beta_d$ yields
\[
\widehat\mu^{cf}_d(a_n)-\mu_d-(a_n-1)\beta_d
=
O_p\left(
a_n n_d^{-1/2}
+N_d^{-1/2}
+\frac{n_d}{N_d}
+n^{-\gamma}
\right),
\]
which proves part~(iv).
\end{proof}

%%%%%%%%%%%%%%%%%%%%%%%%%%%%%%%%%%%%%%%%%%%%
%%% Corollary 4.2 Finite population corr
%%%%%%%%%%%%%%%%%%%%%%%%%%%%%%%%%%%%%%%%%%%%
\begin{proof}[Proof of Corollary~4.2 in the main document]
For either estimator,
$\hat\mu_d^{(\cdot)}(a_n)-\bar Y_d
=\{\hat\mu_d^{(\cdot)}(a_n)-\mu_d\}-(\bar Y_d-\mu_d)$.
Since $\bar Y_d-\mu_d=(P_{N_d}-P_d)Y$ is a centred average of $N_d$ i.i.d.\
outcomes by Assumption~2 in the main document, with finite variance because
$E_{P_d}(Y^2)<\infty$, Chebyshev's inequality yields
$\bar Y_d-\mu_d=O_p(N_d^{-1/2})$, and Assumption~6(ii) in the main document
gives $\sqrt{n_d}(\bar Y_d-\mu_d)=o_p(1)$. Consequently, the convergence in
probability in part~(i) of Propositions~4.1 and~4.2 in the main document
carries over to $\bar Y_d$, and so do the $\sqrt{n_d}$-normalised limit laws
by Slutsky's lemma: under the conditions of part~(ii), where $a_n\equiv1$,
or of part~(iii), where $a\in(0,1]$ and Assumption~8 or~10 in the main
document holds, respectively,
$\sqrt{n_d}\{\hat\mu_d^{(\cdot)}(a_n)-\bar Y_d-(a_n-1)\beta_d\}
\xrightarrow{d}N(0,a^2\sigma_d^2)$. Finally, the rates in part~(iv) of
Propositions~4.1 and~4.2 in the main document are preserved by the triangle
inequality, since the additional term is $O_p(N_d^{-1/2})$, an order already
included in those rates.
\end{proof}

%\begin{proof}[Proof of Corollary~4.2 in the main document]
%For either estimator,$\hat\mu_d^{(\cdot)}(a_n)-\bar Y_d=\{\hat\mu_d^{(\cdot)}(a_n)-\mu_d\}-(\bar Y_d-\mu_d)$.Since $\bar Y_d-\mu_d=(P_{N_d}-P_d)Y$ isa centred average of $N_d$ i.i.d.\ outcomes,the finite-second-moment assumption and Chebyshev'sinequality yield$\bar Y_d-\mu_d=O_p(N_d^{-1/2})$.By Assumption~6(ii),$\sqrt{n_d}(\bar Y_d-\mu_d)=o_p(1)$.Consequently, consistency and the$\sqrt{n_d}$-normalized limit laws ofPropositions~4.1 and~4.2 in the main document carry over to$\bar Y_d$ by Slutsky's lemma, giving$\sqrt{n_d}\{\hat\mu_d^{(\cdot)}(a_n)-\bar Y_d-(a_n-1)\beta_d\}\xrightarrow{d}N(0,a^2\sigma_d^2)$ when $a_n\equiv1$.Finally, the rate in part (iv) ofProposition~4.1 in the main document is preserved by the triangle inequality, since the additional term is$O_p(N_d^{-1/2})$, an order already included in that rate.\end{proof}

%%%%%%%%%%%%%%%%%%%%%%%%%%%%%%%%%%%%%%%%%%%%
%%% Corollary 4.3 Variance
%%%%%%%%%%%%%%%%%%%%%%%%%%%%%%%%%%%%%%%%%%%%

\begin{proof}[Proof of Corollary~4.3 in the main document]
Let $\varepsilon_{dj}^*:=Y_{dj}-g_d^*(\mathbf Z_{dj})$ and write
$\hat\varepsilon_{dj}=\varepsilon_{dj}^*-\Delta_{dj}$, where
$\Delta_{dj}:=\hat g_{n,d}(\mathbf Z_{dj})-g^*_d(\mathbf Z_{dj})$ without
sample splitting and
$\Delta_{dj}:=\hat g^{(-\ell)}_{n,d}(\mathbf Z_{dj})-g^*_d(\mathbf Z_{dj})$
for $j\in\mathcal Q_d^s\cap\mathcal I_\ell$ under cross-fitting. Writing
$P_{n_d}\Delta^2:=n_d^{-1}\sum_{j\in\mathcal Q^s_d}\Delta_{dj}^2$ and
analogously for the other sample averages below,
\[
\hat\sigma_d^2
=\frac{n_d}{n_d-1}
\bigl\{P_{n_d}\hat\varepsilon^2-(P_{n_d}\hat\varepsilon)^2\bigr\},
\qquad
P_{n_d}\hat\varepsilon^2
=P_{n_d}\varepsilon^{*2}-2P_{n_d}(\varepsilon^*\Delta)+P_{n_d}\Delta^2,
\qquad
P_{n_d}\hat\varepsilon=P_{n_d}\varepsilon^*-P_{n_d}\Delta .
\]
By Corollary~\ref{cor:sigma_finite}, $E_{P_d}\{(\varepsilon^*_d)^2\}<\infty$,
so the law of large numbers gives
$P_{n_d}\varepsilon^{*2}\xrightarrow{p}E_{P_d}\{(\varepsilon^*_d)^2\}$ and
$P_{n_d}\varepsilon^*\xrightarrow{p}\beta_d$. By the Cauchy--Schwarz
inequality,
$|P_{n_d}(\varepsilon^*\Delta)|
\le(P_{n_d}\varepsilon^{*2})^{1/2}(P_{n_d}\Delta^2)^{1/2}$ and
$|P_{n_d}\Delta|\le(P_{n_d}\Delta^2)^{1/2}$. Hence, since
$n_d/(n_d-1)\to1$,
$\hat\sigma_d^2\xrightarrow{p}
E_{P_d}\{(\varepsilon^*_d)^2\}-\beta_d^2=\sigma_d^2$
as soon as $P_{n_d}\Delta^2=o_p(1)$, which we verify in the two settings.

Without sample splitting, consider the class
$\mathcal H^2_n:=\{h^2:\ h\in\mathcal G_{c,n,d}-g^*_d\}$, with
$\mathcal G_{c,n,d}$ as in Lemma~\ref{lemma:maximal_inequality}. Its members
are bounded by $4B^2$, and since
$|h^2-h'^2|=|h-h'|\,|h+h'|\le4B\,|h-h'|$, every $L_2(Q)$-cover of
$\mathcal G_{c,n,d}-g^*_d$ of radius $\epsilon B$ induces an $L_2(Q)$-cover
of $\mathcal H^2_n$ of radius $4B^2\epsilon$ with the same cardinality.
Together with the entropy bound for $\mathcal G_{c,n,d}-g^*_d$ established in
the proof of Lemma~\ref{lemma:maximal_inequality}, this gives
$\sup_Q\log N(4B^2\epsilon,\mathcal H^2_n,L_2(Q))\le v_n\log(A/\epsilon)$
for $0<\epsilon\le1$, the supremum over finitely discrete $Q$ on
$\mathcal Z$ and the case $\epsilon=1$ following by monotonicity in the
radius. Therefore, the uniform entropy bound of Definition~2 in the main document holds
for $\mathcal H^2_n$ relative to the envelope $4B^2$. The class
$\mathcal H^2_n$ is pointwise measurable as the image of a pointwise
measurable class under a continuous map. Corollary~5.1 of
\cite{chernozhukov2014gaussian} with envelope $4B^2$ and radius
$\sigma=4B^2$ therefore gives
\[
E\|\mathbb G_{n_d}\|_{\mathcal H^2_n}
\lesssim
B^2\sqrt{v_n\log A}+\frac{v_nB^2}{\sqrt{n_d}}\log A .
\]
On $\mathcal E_n$ we have $\Delta^2_{n,d}\in\mathcal H^2_n$ and hence
$\sqrt{n_d}\,|(P_{n_d}-P_d)\Delta^2_{n,d}|
\le\|\mathbb G_{n_d}\|_{\mathcal H^2_n}$, so that, for every $\epsilon>0$,
by Markov's inequality,
\[
\Pr\bigl\{|(P_{n_d}-P_d)\Delta^2_{n,d}|>\epsilon\bigr\}
\le
\Pr(\mathcal E_n^c)
+\Pr\bigl\{\|\mathbb G_{n_d}\|_{\mathcal H^2_n}>\epsilon\sqrt{n_d}\bigr\}
\lesssim
\Pr(\mathcal E_n^c)
+\frac{B^2}{\epsilon}
\Bigl(\sqrt{\frac{v_n}{n_d}}+\frac{v_n}{n_d}\Bigr)\log A
\to0,
\]
by Assumption~7 in the main document and because Assumption~5 in the main
document implies $v_n/n_d\to0$. Since
$P_d\Delta^2_{n,d}=r^2_{n,d}=o_p(1)$ by Assumption~7 in the main document,
it follows that $P_{n_d}\Delta^2=o_p(1)$.

Under cross-fitting, fix $\ell$ and let $\mathcal F_{-\ell}$ be as in the
proof of Proposition~4.2 in the main document. Conditionally on
$\mathcal F_{-\ell}$, the function $\Delta_\ell=\hat g^{(-\ell)}_{n,d}-g^*_d$
is fixed and the covariates indexed by $\mathcal Q_d^s\cap\mathcal I_\ell$
are i.i.d.\ $P_d^{\mathbf Z}$ and independent of $\mathcal F_{-\ell}$, so
that $E(P_{n_{d\ell}}\Delta_\ell^2\mid\mathcal F_{-\ell})
=\|\Delta_\ell\|^2_{P_d^{\mathbf Z},2}$, which is $o_p(1)$ by Assumption~9
in the main document. The conditional Markov inequality gives, for every
$\epsilon,\eta>0$,
$\Pr(P_{n_{d\ell}}\Delta_\ell^2>\epsilon)
\le\Pr(\|\Delta_\ell\|^2_{P_d^{\mathbf Z},2}>\eta)+\eta/\epsilon$,
and letting first $n\to\infty$ and then $\eta\downarrow0$ yields
$P_{n_{d\ell}}\Delta_\ell^2=o_p(1)$. Since
$P_{n_d}\Delta^2=\sum_{\ell=1}^L(n_{d\ell}/n_d)P_{n_{d\ell}}\Delta_\ell^2$
and $L$ is fixed, $P_{n_d}\Delta^2=o_p(1)$. This proves
$\hat\sigma_d^2\xrightarrow{p}\sigma_d^2$ in both settings.

Finally, by the definition of
$\bar r_{n_d}^{(\cdot)}:=P_{n_d}(Y-\hat g_{n,d}^{(\cdot)})$ with
$\hat g_{n,d}^{(\cdot)}\in\{\hat g_{n,d},\hat g_{n,d}^{\mathrm{oof}}\}$, we
have $\tilde\mu_d^{(\cdot)}(a_n)=\hat\mu_d^{(\cdot)}(1)$ exactly. By
part~(ii) of Proposition~4.1 or~4.2 in the main document,
$\sqrt{n_d}\{\hat\mu_d^{(\cdot)}(1)-\mu_d\}\xrightarrow{d}N(0,\sigma_d^2)$,
and since $\hat\sigma_d\xrightarrow{p}\sigma_d>0$, Slutsky's lemma gives
$\sqrt{n_d}\{\hat\mu_d^{(\cdot)}(1)-\mu_d\}/\hat\sigma_d
\xrightarrow{d}N(0,1)$, from which equation (16) in Corollary~4.3 in the main document follows because the
limiting distribution is continuous. The statement with $\bar Y_d$ in place
of $\mu_d$ follows in the same way from Corollary~4.2 in the main document.
\end{proof}

%%%%%%%%%%%%%%%%%%%%%%%%%%%%%%%%%%%%%%%%%%%%%%%%%%%%%%%%%%%%%%%%%%%%%%%%%%%%%%%
% Proofs Section 5: aggregated learners 
%%%%%%%%%%%%%%%%%%%%%%%%%%%%%%%%%%%%%%%%%%%%%%%%%%%%%%%%%%%%%%%%%%%%%%%%%%%%%%%

\subsection{Proofs and derivations for Section~5 in the main document}

%%%%%%%%%%%%%%%%%%%%%%%%%
%%%% Normality under aggregation, no cross-fitting
%%%%%%%%%%%%%%%

\begin{proof}[Proof of Proposition~5.1 in the main document]

Write $\Delta_{n,d}:=\hat g_{n,d}-g_d^*$
and $\varepsilon_d^*:=Y-g_d^*(\mathbf Z)$.
To facilitate comparison with the proofs of
Propositions~4.1 and~4.2 in the main document, we retain the same notation for corresponding remainder terms. There are no $T_1$ or $T_{3,N}$ terms here because the aggregated estimator does not contain the unit-level population average $P_{N_d}\hat g_{n,d}$. Instead, the effect of replacing the population covariate distribution by its aggregate enters through the terms $T_5$ and $T_6$ below.

Subtracting $\mu_d^{\mathrm{agg}}$
from the aggregated small area learner $\hat\mu_d^{\mathrm{agg}}(a)$ in Definition~5 in the main document gives

\begin{align}
&\hat\mu_d^{\mathrm{agg}}(a_n)
-\mu_d^{\mathrm{agg}}
=
\underbrace{(a_n-1)\beta_d}_{T_0}
+
\underbrace{a_n(P_{n_d}-P_d)\varepsilon_d^*}_{T_2}
-
\underbrace{a_n(P_{n_d}-P_d)\Delta_{n,d}}_{T_{3,n}}
\nonumber\\
&\qquad+
\underbrace{
\hat g_{n,d}(\bar{\mathbf Z}_d)
-\hat g_{n,d}(\bar{\mathbf z}_d)}_{T_5}
+
\underbrace{
\Delta_{n,d}(\bar{\mathbf z}_d)
-P_d\Delta_{n,d}}_{T_6}
+
\underbrace{
(1-a_n)P_d\Delta_{n,d}}_{T_4}.
\label{eq:T_decomposition_agg}
\end{align}

\emph{The residual term $T_2$.}
The sampled pairs $(Y_{dj},\mathbf Z_{dj})_{j\in\mathcal Q_d^s}$ are i.i.d.\
$P_d$ by Assumption~1(iii), with
$E_{P_d}(\varepsilon_d^*)=\beta_d$ and
$\operatorname{Var}_{P_d}(\varepsilon_d^*)=\sigma_d^2$, finite by
Corollary~\ref{cor:sigma_finite}. The Lindeberg--L\'evy central limit theorem
therefore gives
\[
\sqrt{n_d}\,(P_{n_d}-P_d)\varepsilon_d^*
\xrightarrow{d}N(0,\sigma_d^2),
\]
the limiting variance being strictly positive by the hypothesis of
Proposition~5.1 in the main document. Hence $T_2=a_nO_p(n_d^{-1/2})$ and, since $a_n\to a$, Slutsky's lemma gives
$\sqrt{n_d}\,T_2\xrightarrow{d}N(0,a^2\sigma_d^2)$.

\emph{The empirical-process term $T_{3,n}$.}
By Lemma~\ref{lemma:empirical_process},
$(P_{n_d}-P_d)\Delta_{n,d}=o_p(n_d^{-1/2})$.
Hence
$T_{3,n}=a_n o_p(n_d^{-1/2})$, and in particular
$\sqrt{n_d}\,T_{3,n}=o_p(1)$ since $a_n\leq1$.

\emph{The learner-mean term $T_4$.}
By the Cauchy--Schwarz inequality and Assumption~7 in the main document, $|P_d\Delta_{n,d}|\le r_{n,d}=o_p(1)$, so that $T_4=o_p(1)$ for every sequence $a_n$ in $[0,1]$, since $|1-a_n|\le1$.

\emph{The mean-approximation term $T_5$.}
By Assumptions~2 and~11(iii),
$\bar{\mathbf Z}_d-\bar{\mathbf z}_d=O_p(N_d^{-1/2})$, so that
$\bar{\mathbf Z}_d\in U_d$ with probability tending to one and the local
Lipschitz condition in Assumption~11(ii) in the main document applies. Hence $T_5=O_p(N_d^{-1/2})=o_p(n_d^{-1/2})$ by Assumption~6(ii) in the main document.

\emph{The aggregation-contrast term $T_6$.}
Assumption~11(i) in the main document gives $\sqrt{n_d}T_6=o_p(1)$.

\emph{Conclusion.}
By the preceding steps, $T_2$, $T_{3,n}$, $T_4$, $T_5$ and $T_6$ are all
$o_p(1)$. Hence \eqref{eq:T_decomposition_agg} gives
\[
\hat\mu_d^{\mathrm{agg}}(a_n)-\mu_d^{\mathrm{agg}}
=
(a_n-1)\beta_d+o_p(1),
\]
which converges in probability to $(a-1)\beta_d$ and proves part~(i).

For part~(ii), let $a_n\equiv1$. Then $T_0=T_4=0$ identically, while
$T_{3,n}$, $T_5$ and $T_6$ are $o_p(n_d^{-1/2})$ by the preceding steps.
Since the residual-term argument gives
$\sqrt{n_d}\,T_2\xrightarrow{d}N(0,\sigma_d^2)$ when $a_n=1$,
\eqref{eq:T_decomposition_agg} and Slutsky's lemma yield the asserted
limit.

For part~(iii), let $a\in(0,1]$. Assumption~8 in the main document gives
$T_4=(1-a_n)P_d\Delta_{n,d}=o_p(n_d^{-1/2})$, while $T_{3,n}$, $T_5$ and
$T_6$ are $o_p(n_d^{-1/2})$ by the preceding steps. After subtracting
$(a_n-1)\beta_d$ from \eqref{eq:T_decomposition_agg} and using
$\sqrt{n_d}\,T_2\xrightarrow{d}N(0,a^2\sigma_d^2)$, Slutsky's lemma proves
part~(iii).

Finally, suppose that $a_n\to0$ and
$P_d\Delta_{n,d}=O_p(n^{-\gamma})$
for some $\gamma\in(0,1/2]$. The preceding bounds give $T_2=a_nO_p(n_d^{-1/2})$, $T_{3,n}=a_no_p(n_d^{-1/2})$,
$T_5=O_p(N_d^{-1/2})$,
$T_6=o_p(n_d^{-1/2})$, while $T_4=(1-a_n)P_d\Delta_{n,d}
=O_p(n^{-\gamma})$. After subtracting $(a_n-1)\beta_d$ and substituting these bounds into
\eqref{eq:T_decomposition_agg}, we obtain
\[
\hat\mu_d^{\mathrm{agg}}(a_n)-\mu_d^{\mathrm{agg}}
-(a_n-1)\beta_d
=
O_p\!\left(
a_n n_d^{-1/2}
+N_d^{-1/2}
+n^{-\gamma}
\right)
+
o_p(n_d^{-1/2}),
\]
which proves part~(iv). The corresponding statements relative to $\mu_d$ follow from
$\mu_d^{\mathrm{agg}}-\mu_d=b_d$.
\end{proof}

%%%%%%%%%%%%%%%%%%%%%%%%%%%
%%%%% Bias Lemma
%%%%%%%%%%%%%%%%%%%%%%%

\begin{proof}[Proof of Lemma~5.1 in the main document]
Since $\mathcal C_d$ is convex and contains the support of
$P_d^{\mathbf Z}$, it contains $\bar{\mathbf z}_d$ and, for every
$\mathbf z$ in the support, the segment joining $\bar{\mathbf z}_d$ and
$\mathbf z$, so the regularity assumed on $\mathcal C_d$ applies along these
segments. For (i), the Lipschitz property, Jensen's inequality and
$E_{P_d^{\mathbf Z}}\|\mathbf Z-\bar{\mathbf z}_d\|^2
=\operatorname{tr}(\Sigma_{Z,d})$ give
\[
|b_d|
=\bigl|E_{P_d^{\mathbf Z}}\{g_d^*(\bar{\mathbf z}_d)-g_d^*(\mathbf Z)\}\bigr|
\le\kappa_d^*E_{P_d^{\mathbf Z}}\|\mathbf Z-\bar{\mathbf z}_d\|
\le\kappa_d^*\bigl\{E_{P_d^{\mathbf Z}}\|\mathbf Z-\bar{\mathbf z}_d\|^2\bigr\}^{1/2}
=\kappa_d^*\{\operatorname{tr}(\Sigma_{Z,d})\}^{1/2}.
\]
For (ii), Taylor's theorem along the segment gives
$g_d^*(\mathbf Z)=g_d^*(\bar{\mathbf z}_d)
+\nabla g_d^*(\bar{\mathbf z}_d)^\top(\mathbf Z-\bar{\mathbf z}_d)
+R(\mathbf Z)$ with
$|R(\mathbf Z)|\le(H_d/2)\|\mathbf Z-\bar{\mathbf z}_d\|^2$. The linear term
has expectation zero, so $b_d=-E_{P_d^{\mathbf Z}}\{R(\mathbf Z)\}$ and
$|b_d|\le(H_d/2)\operatorname{tr}(\Sigma_{Z,d})$. For (iii), Jensen's
inequality gives
$g_d^*(\bar{\mathbf z}_d)\le E_{P_d^{\mathbf Z}}\{g_d^*(\mathbf Z)\}$ when
$g_d^*$ is convex, i.e.\ $b_d\le0$, and the reverse inequality when it is
concave.
\end{proof}

%%%%%%%%%%%%%%%%%%%%%%%%%%%%%%
%%%% Cross-fitted with aggregation normality
%%%%%%%%%%%%%%%%%%%%%%%%%%%%%%%
\begin{proof}[Proof of Proposition~5.2 in the main document]
For $\ell=1,\ldots,L$, write
$w_{d\ell}:=n_{d\ell}/n_d$,
$\Delta_{n,d}^{(-\ell)}
:=\hat g_{n,d}^{(-\ell)}-g_d^*$,
and
$\bar\Delta_{n,d}
:=L^{-1}\sum_{\ell=1}^L
\Delta_{n,d}^{(-\ell)}$. To facilitate comparison with the proof of Proposition~4.2 in the main document, we retain the same notation for corresponding terms. The finite-population terms $T_1$ and
$T_{3,N}$ from that proof are absent here because the aggregated imputation component involves a single evaluation at $\bar{\mathbf Z}_d$ rather
than a population average. The aggregation instead introduces the additional terms $T_5$ and $T_6$ below.

%Subtracting $\mu_d^{\mathrm{agg}}$ from the aggregated cross-fitted small area learner $\hat\mu_d^{cf,\mathrm{agg}}(a)$ in Definition~6 in the main document gives
Adding and subtracting $\bar g_{n,d}(\bar{\mathbf z}_d)$,
$g^*_d(\bar{\mathbf z}_d)$ and
$\sum_\ell w_{d\ell}P_{n_{d\ell}}g^*_d=P_{n_d}g^*_d$ in the aggregated cross-fitted small area learner $\hat\mu_d^{cf,\mathrm{agg}}(a)$ in
Definition~6 in the main document, centring $P_{n_d}(Y-g^*_d)$,
$P_{n_{d\ell}}\Delta^{(-\ell)}_{n,d}$, $\ell=1,\dots,L$, and
$\bar\Delta_{n,d}(\bar{\mathbf z}_d)$ at their $P_d$-expectations, subtracting $\mu_d^{\mathrm{agg}}$ and using $\beta_d=P_d(Y-g^*_d)$ gives
\begin{align}
&\hat\mu_d^{cf,\mathrm{agg}}(a_n)
-\mu_d^{\mathrm{agg}}
=
\underbrace{(a_n-1)\beta_d}_{T_0}
+
\underbrace{
a_n(P_{n_d}-P_d)\varepsilon_d^*}_{T_2}
-
\underbrace{
a_n\sum_{\ell=1}^L w_{d\ell}
(P_{n_{d\ell}}-P_d)
\Delta_{n,d}^{(-\ell)}}_{T_{3,n}}
\nonumber\\
&\qquad+
\underbrace{
(1-a_n)P_d\bar\Delta_{n,d}
+
a_n\sum_{\ell=1}^L
(1/L-w_{d\ell})
P_d\Delta_{n,d}^{(-\ell)}}_{T_4}
\nonumber\\
&\qquad+
\underbrace{
\bar g_{n,d}(\bar{\mathbf Z}_d)
-\bar g_{n,d}(\bar{\mathbf z}_d)}_{T_5}
+
\underbrace{
\bar\Delta_{n,d}(\bar{\mathbf z}_d)
-P_d\bar\Delta_{n,d}}_{T_6}.
\label{eq:T_decomposition_cf_agg}
\end{align}

\emph{The fold term $T_{3,n}$.}
As in the proof of Proposition~4.2 in the main document, conditional on each fold's training data the evaluation observations are independent of the fitted learner. Conditional variance control and fold-specific $L^2$-consistency therefore give $\sqrt{n_d}T_{3,n}=o_p(1)$, since $L$ is fixed, and hence $\sqrt{n_d}T_{3,n} = o_p(1)$ because $0\leq a_n \leq 1$.

\emph{The learner-mean term $T_4$.}
Since
$P_d\bar\Delta_{n,d}
=L^{-1}\sum_{\ell=1}^L
P_d\Delta_{n,d}^{(-\ell)}$,
the term $T_4$ can equivalently be written as $T_4
=
\sum_{\ell=1}^L
(1/L-a_nw_{d\ell})
P_d\Delta_{n,d}^{(-\ell)}$, which is the same learner-mean term as $T_4$
in the proof of Proposition~4.2 in the main document, written here in a form adapted to the aggregated estimator. Fold-specific $L^2$-consistency implies
$P_d\bar\Delta_{n,d}=o_p(1)$.
Moreover, fold stratification gives
$|1/L-w_{d\ell}|\leq n_d^{-1}$, and hence
$\sum_{\ell=1}^L
(1/L-w_{d\ell})
P_d\Delta_{n,d}^{(-\ell)}
=o_p(n_d^{-1})$ by Cauchy--Schwarz inequality and fixed $L$.
Consequently, $T_4=o_p(1)$.

\emph{The mean-approximation term $T_5$.}
By Assumption~2 and~12(iii),
$\bar{\mathbf Z}_d-\bar{\mathbf z}_d
=O_p(N_d^{-1/2})$.
The local Lipschitz condition in
Assumption~12(ii)
therefore gives
$T_5=O_p(N_d^{-1/2})
=o_p(n_d^{-1/2})$
under Assumption~6(ii).

\emph{The aggregation-contrast term $T_6$.}
Assumption~12(i)
gives
$\sqrt{n_d}T_6=o_p(1)$.

\emph{The leading residual term $T_2$.}
The sampled pairs
$(Y_{dj},\mathbf Z_{dj})_{j\in\mathcal Q_d^s}$
are i.i.d.\ $P_d$, so the central limit theorem
gives $\sqrt{n_d}(P_{n_d}-P_d)\varepsilon_d^*
\xrightarrow{d}N(0,\sigma_d^2)$. Since $a_n\to a$, $\sqrt{n_d}T_2
\xrightarrow{d}N(0,a^2\sigma_d^2)$, with the degenerate limit at zero when $a=0$.

\emph{Conclusion.}
The preceding bounds imply
$T_{3,n}=o_p(1)$,
$T_4=o_p(1)$,
$T_5=o_p(1)$, and
$T_6=o_p(1)$, while $T_2=o_p(1)$.
Hence
\eqref{eq:T_decomposition_cf_agg}
establishes part (i).

If $a_n\equiv1$, the first component of $T_4$
vanishes identically, while its fold-imbalance
component is $o_p(n_d^{-1})$.
Thus all terms except $T_2$ are
$o_p(n_d^{-1/2})$, and Slutsky's lemma proves
part (ii) without
Assumption~10 in the main document.

If $a\in(0,1]$,
Assumption~10 in the main document
implies $\sqrt{n_d}P_d\bar\Delta_{n,d}=o_p(1)$.
Together with the fold-imbalance bound, this
gives $\sqrt{n_d}T_4=o_p(1)$.
After subtracting $(a_n-1)\beta_d$,
all terms except $T_2$ are therefore negligible
at the $\sqrt{n_d}$ scale, and Slutsky's lemma
proves part (iii).

Finally, suppose that $a_n\to0$ and
$P_d(\bar g_{n,d}-g_d^*)
=O_p(n^{-\gamma})$ for some
$\gamma\in(0,1/2]$.
Then
$T_2=O_p(a_n n_d^{-1/2})$,
$T_{3,n}=a_no_p(n_d^{-1/2})$,
$T_5=O_p(N_d^{-1/2})$, and
$T_6=o_p(n_d^{-1/2})$.
The first component of $T_4$ is
$(1-a_n)O_p(n^{-\gamma})$, while its fold-imbalance
component is $a_no_p(n_d^{-1})$.
Substitution into
\eqref{eq:T_decomposition_cf_agg}
therefore proves part (iv).

The conclusion for $\mu_d$ follows by recalling that 
\begin{equation}
\label{eq:mu_agg}
\mu_d^{\mathrm{agg}}
:=
g_d^*(\bar{\mathbf z}_d)
+
E_{P_d}\{Y-g_d^*(\mathbf Z)\},
\qquad
b_d
:=
\mu_d^{\mathrm{agg}}-\mu_d
=
g_d^*(\bar{\mathbf z}_d)
-
E_{P_d^{\mathbf Z}}\{g_d^*(\mathbf Z)\}.
\end{equation}
The aggregation-bias bounds follow from
Lemma~5.1 in the main document.
\end{proof}

%%%% Proof Corollary 

\begin{proof}[Proof of Corollary~5.2 in the main document]
By \eqref{eq:mu_agg},
\begin{align*}
&\hat\mu_d^{(\cdot),\mathrm{agg}}(a_n)
-\bar Y_d-b_d-(a_n-1)\beta_d\\
&\quad=
\{\hat\mu_d^{(\cdot),\mathrm{agg}}(a_n)
-\mu_d^{\mathrm{agg}}
-(a_n-1)\beta_d\}
-(\bar Y_d-\mu_d).
\end{align*}
The first term satisfies the corresponding
aggregation proposition. As in the proof of
Corollary~4.2 in the main document,
$\bar Y_d-\mu_d=O_p(N_d^{-1/2})$,
and Assumption~6(ii) in the main document gives
$\sqrt{n_d}(\bar Y_d-\mu_d)=o_p(1)$.
Hence the convergence in probability and the
limit laws in parts (ii) and (iii) follow from
Slutsky's lemma. For part (iv), the additional
term $\bar Y_d-\mu_d$ is already of order
$O_p(N_d^{-1/2})$ and is therefore absorbed by
the stated rate bound.
\end{proof}

%%% Proof Coroallrt CI aggregation 

\subsubsection{Variance estimation and confidence intervals}
\label{sec:variance_agg}

The aggregated families use the same sampled
residuals as their unit-level counterparts.
Consequently, the variance estimator
$\hat\sigma_d^2$, defined in equation (15) in the main document, applies unchanged, using fitted residuals without sample splitting or their out-of-fold counterparts under cross-fitting. Its construction requires only the unit-level covariates available for sampled observations.

As in Sections~4.1--4.2 of the main document, an exact residual correction converts every member into its fully
corrected counterpart. For either $\hat\mu_d^{\mathrm{agg}}(a_n)$ in
Definition~5 in the main document or
$\hat\mu_d^{cf,\mathrm{agg}}(a_n)$ in
Definition~6 in the main document, let
$\bar r_{n_d}^{(\cdot)}
:=P_{n_d}(Y-\hat g_{n,d}^{(\cdot)})$,
where the fitted values are out-of-fold under
cross-fitting. Then
$\hat\mu_d^{(\cdot),\mathrm{agg}}(a_n)
+(1-a_n)\bar r_{n_d}^{(\cdot)}
=\hat\mu_d^{(\cdot),\mathrm{agg}}(1)$.

\begin{corollary}[Confidence intervals under aggregation]
\label{cor:ci_agg}
Under the baseline assumptions of
Proposition~5.1 or~5.2 in the main document,
respectively, $\hat\sigma_d^2\xrightarrow{p}\sigma_d^2$.
For any deterministic $a_n\in[0,1]$, define
$\tilde\mu_d^{(\cdot),\mathrm{agg}}(a_n)
:=\hat\mu_d^{(\cdot),\mathrm{agg}}(a_n)
+(1-a_n)\bar r_{n_d}^{(\cdot)}$.
Then, for $\alpha\in(0,1)$,
\begin{equation}
\label{eq:ci_agg}
\Pr\left\{
\left|
\tilde\mu_d^{(\cdot),\mathrm{agg}}(a_n)
-\mu_d^{\mathrm{agg}}
\right|
\leq
z_{1-\alpha/2}
\frac{\hat\sigma_d}{\sqrt{n_d}}
\right\}
\longrightarrow1-\alpha.
\end{equation}
If $b_d$ is known, the interval centred at
$\tilde\mu_d^{(\cdot),\mathrm{agg}}(a_n)-b_d$
has the same asymptotic coverage for $\mu_d$
and, under Corollary~5.2 in the main document,
for $\bar Y_d$.
\end{corollary}

For an unadjusted member with $a\in(0,1)$
satisfying part (iii)  of Propositions~5.1 and 5.2 in the main document, a consistent first-order standard
error is
$a_n\hat\sigma_d/\sqrt{n_d}$.
Inference for $\mu_d^{\mathrm{agg}}$ additionally
requires accounting for $(a_n-1)\beta_d$.
Correcting this term using the same sampled
residuals recovers the fully corrected member
exactly, with the standard error in
\eqref{eq:ci_agg}. The endpoint $a_n\to0$ is
instead governed by part (iv) of Propositions~5.1 and~5.2 in the main document.

\begin{remark}[Inference for $\mu_d$ or $\bar Y_d$ under aggregation]
\label{rem:agg_bias_inference}
When $b_d$ in \eqref{eq:mu_agg} is unknown, estimating it from the sampled covariates requires care. In particular, the natural plug-in correction collapses the fully
corrected aggregated estimator to the direct estimator: without cross-fitting,
subtracting
$\hat b_d=\hat g_{n,d}^{\mathrm{agg}}-P_{n_d}\hat g_{n,d}$
yields $P_{n_d}Y$, and the same holds under cross-fitting with
$\hat b_d=\bar g_{n,d}^{\mathrm{agg}}
-P_{n_d}\hat g_{n,d}^{\mathrm{oof}}$.
Its first-order variance is therefore generally $\operatorname{Var}_{P_d}(Y)/n_d$ rather than $\sigma_d^2/n_d$.
Other estimators of $b_d$ require joint analysis of their variability and
covariance with the aggregated estimator. Alternatively, if a known bound $B_d\geq|b_d|$ is available, enlarging the
radius in \eqref{eq:ci_agg} by $B_d$ yields asymptotic coverage at least
$1-\alpha$ for both $\mu_d$ and $\bar Y_d$.
Lemma~5.1 in the main document provides such bounds when the required
regularity constants and covariate dispersion are known or conservatively
bounded.
\end{remark}

\begin{proof}[Proof of Corollary~\ref{cor:ci_agg}]
Aggregation does not alter the sampled residuals.
The variance-consistency argument of
Corollary~4.3 in the main document therefore applies
unchanged. Under cross-fitting, it uses only
fold-specific $L^2$-consistency and independence
of training and evaluation observations,
not the stronger sampling-fraction condition.

For either family, direct substitution gives
$\tilde\mu_d^{(\cdot),\mathrm{agg}}(a_n)
=\hat\mu_d^{(\cdot),\mathrm{agg}}(1)$.
Part (ii) of the corresponding aggregation
proposition and Slutsky's lemma therefore give
\eqref{eq:ci_agg}. If $b_d$ is known, subtracting
$b_d$ from the interval centre gives the stated
coverage for $\mu_d$ by \eqref{eq:mu_agg}, and
for $\bar Y_d$ by
Corollary~5.2 in the main document.
\end{proof}

\section{Discussion}\label{sec:discussion}
In this section, we discuss several aspects of the proposed estimators that
complement the preceding theoretical results.

\subsection{On the hierarchical correction weights}
\label{sec:hierarchical_w}
Consider the bias--variance approximation from Section~4.2 of the main document, i.e., $\operatorname{AMSE}_d(a\mid\beta_d)
=(1-a)^2\beta_d^2+a^2\sigma_d^2/n_d$, minimised at the oracle weight
$a_d^\star=\beta_d^2/(\beta_d^2+\sigma_d^2/n_d)$. The oracle weight $a_d^\star$ is conditional on the realized approximation
error $\beta_d$. A two-level model, in which a distribution is placed on the area-level
approximation errors $\beta_d$, gives a different member of the family in Definition~3 in the main document. Suppose that, conditionally on $g^*$ in Assumption~7 in the main document and on the area covariate laws $P_d^{\mathbf Z}=\mathcal L(\mathbf Z\mid G=d)$ of equation (2) of the main document, the
area-level approximation errors $\beta_d$ of \eqref{eq:sigma_d} are mutually
independent across areas, with $E(\beta_d)=0$ and
$\operatorname{Var}(\beta_d)=A_\beta$, and independent of the sampling error
in $R_d:=P_{n_d}(Y-g^*_d)=\beta_d+\bar e_d$, where
$\operatorname{Var}(\bar e_d\mid\beta_d)=\sigma_d^2/n_d$. Averaging over the distribution of $\beta_d$ gives the mean squared error of the correction $aR_d$
as a predictor of $\beta_d$,
$(1-a)^2A_\beta+a^2\sigma_d^2/n_d$,
minimised at $\Gamma_d
=A_\beta/\left(A_\beta+\sigma_d^2/n_d\right)$. The resulting predictor $\Gamma_dR_d$ is also the best linear predictor (BLP) of $\beta_d$ from $R_d$
and, under Gaussianity, the posterior mean
(the Bayes predictor, BP), by the conjugate calculation
\citep[Chapter~6]{rao2015}. This weight therefore selects the member of the family in Definition~3 in the main document with $a=\Gamma_d$. Moreover, $\Gamma_d$ has exactly the form of the shrinkage factor of the models in Sections~3.2.1 and~3.2.2 in the main document.

The three weights, i.e., $a=1$, $a=a_d^\star$ and $a=\Gamma_d$, answer three different questions.
First, $a=1$ defines the fully corrected member of the family in Definition~3 in the main document and is consistent under the stated assumptions irrespective of the value of $\beta_d$. Second, $a=a_d^\star$ minimises the approximate MSE criterion conditional on the realized $\beta_d$. Third, $a=\Gamma_d$ minimizes that criterion on average under the second-level model; this optimality relies on the specified second-moment and independence assumptions.

\subsection{On the low-complexity class}
\label{sec:low_complex}

The class $\mathcal G_{c,n}$ of
Definition~2 in the main document accommodates
uniformly bounded linear, penalised (such as lasso, ridge or elastic net over a dictionary of uniformly bounded basis functions), series,
spline and additive learners, as well as
fixed-complexity trees and ensembles.
For a $K_n$-dimensional dictionary, the entropy
index is typically $v_n=O(K_n)$, so
Assumption~5 in the main document requires
$K_n\log n_d=o(\sqrt{n_d})$. Unrestricted
area indicators contribute to $K_n$ and
therefore restrict the growth of $D_n$.
Similarly, for tree ensembles with $M_n$
trees of depth at most $k_n$, classes satisfying
$v_n\lesssim M_n2^{k_n}$ meet the entropy
condition when
$M_n2^{k_n}\log n_d=o(\sqrt{n_d})$.
Any additional complexity from splitting
on the area indicator must also be accounted for. These complexity bounds must be combined
with the area-wise convergence required by
Assumption~7 in the main document. When the entropy conditions are difficult to verify,
cross-fitting removes their necessity but
still requires area-wise learner consistency and, for the unit-level population-covariate formulation, the population-averaging condition
in Lemma~4.1 in the main document.

\subsection{On sampling design and the unweighted residual correction}
\label{sec:weighting}

The unweighted residual correction in
Definition~3 in the main document relies on
Assumption~1. Noninformativeness
given the covariates alone is insufficient:
if inclusion probabilities vary with $\mathbf Z$
within an area, the sampled covariate distribution
is tilted, and the unweighted residual mean
generally does not estimate
$P_d(Y-g_d^*)$. An exception occurs when
$g_d^*=m_d$, since the residual is then
conditionally centred. Self-weighting within
areas removes this distinction and preserves
robustness to learner misspecification, while
allowing sampling fractions to differ across areas.

Under unequal-probability sampling, the natural alternative is the Horvitz–Thompson-normalised version of the weighted residual correction used in model-assisted estimation with flexible learners \citep{dagdoug2026machine,an2026agnostic}. The distinction in SAE is that surveys are generally not designed to ensure adequate sample sizes within every small area. Hence, area-specific weighted residual corrections may be unstable or unavailable, even when the pooled learner can borrow strength across areas. Extending our results to this setting would require conditions on within-area inclusion probabilities and a central limit theorem for the weighted residual mean.

%\subsection{On uniform statements}
%\label{sec:uniform}
%The results in Sections~4.1--4.2 of the main document and Sections~5.1--5.2 of the main document are area-wise: they concern a fixed area $d$ along the sequence of populations. Uniform results over $\mathcal D_n$,%such as bounds on $\max_d|\hat\mu_d-\mu_d|$, simultaneous confidence intervals, or multiple testing procedures, would require uniform versions of the learner assumptions, suitable moment and nondegeneracy conditions, and rates accounting for the growth of $D_n$. In particular, they require sufficient growth of the smallest area sample size and uniform control of learner errors and empirical-process or cross-fitted remainders. Uniform bias control is additionally needed for partially corrected estimators, while simultaneous inference may require accounting for dependence across areas. We leave these extensions for future work.

\section{On the choice of learner}
\label{sec:choice_learner}

\subsection{Introduction}\label{sec:learner-selection}
A central practical question is how to choose a learner for $g$ and the final estimator in real data analysis. Classical model-based methods guide estimation through a specified working model, but selecting among competing procedures still requires a practical assessment criterion. In applications, where the true area targets are unknown, we propose evaluating candidate
procedures by aggregating their small-area estimates to coarser subpopulations and comparing these aggregates with the more stable direct
estimates available at that level. This approach exploits the precision available at coarser scales to support selection when area-level direct
estimates are too noisy. A key benefit is that learners are assessed through the target estimates they produce, using an observable survey
benchmark. The application in Section~\ref{sec:data_analysis} illustrates
the promise of this approach. Since the comparison uses shared survey observations, the criterion is interpreted as internal validation of
agreement at the chosen aggregation level. It thus provides a practical, target-oriented basis for ranking candidate learners and estimators.

\subsection{Description of the procedure}
Suppose that the population $\mathcal Q$ is also partitioned into $C$ mutually disjoint coarse subpopulations, $C<D$, indexed by
$\mathcal C=\{1,\dots,C\}$, and let $\varrho\colon\mathcal D\to\mathcal C$ be
the map assigning each area to the coarse subpopulation containing it, so
that $\mathcal D_{(c)}:=\{d\in\mathcal D:\varrho(d)=c\}$
partitions $\mathcal D$. The coarse subpopulations, their sizes, their
sampled parts and the corresponding sample sizes are
\[
\mathcal Q_{(c)}
=
\bigcup_{d\in\mathcal D_{(c)}}\mathcal Q_d,\;
N_{(c)}
=
\sum_{d\in\mathcal D_{(c)}}N_d,\;
\mathcal Q^s_{(c)}
=
\bigcup_{d\in\mathcal D_{(c)}}\mathcal Q^s_d,\;
n_{(c)}
=
\sum_{d\in\mathcal D_{(c)}}n_d,
\]
with the nonsampled parts defined analogously.

At the coarse level, the target and its direct estimators are
\begin{equation}
\label{eq:target_parameter_region}
\bar y_{(c)}
=
\frac{1}{N_{(c)}}
\sum_{d\in\mathcal D_{(c)}}
\sum_{j\in\mathcal Q_d}y_{dj},
\;
\widehat{\bar y}^{\,ht}_{(c)}
=
\frac{1}{N_{(c)}}
\sum_{d\in\mathcal D_{(c)}}
\sum_{j\in\mathcal Q_d^s}w_{dj}y_{dj},
\;
\widehat{\bar y}^{\,h}_{(c)}
=
\frac{1}{\widehat N_{(c)}}
\sum_{d\in\mathcal D_{(c)}}
\sum_{j\in\mathcal Q_d^s}w_{dj}y_{dj},
\end{equation}
with
\[
\widehat N_{(c)}
=
\sum_{d\in\mathcal D_{(c)}}\widehat N_d,
\qquad
\widehat N_d
=
\sum_{j\in\mathcal Q_d^s}w_{dj}.
\]
The variance of the direct estimators in
\eqref{eq:target_parameter_region} is of order $n_{(c)}^{-1}$ rather than
$n_d^{-1}$, so that they are more stable than their area-level counterparts
when the coarse subpopulations are chosen large enough, and the choice of
learner may be based on the closeness of the aggregated small area estimates
to them. Let
\begin{equation*}
\hat\mu^{(\cdot)}_{+(c)}
=
\sum_{d\in\mathcal D_{(c)}}
\tilde W_{cd}\,
\hat\mu^{(\cdot)}_d,
\qquad
\tilde W_{cd}
=
\frac{N_d}{N_{(c)}}
\quad\text{or}\quad
\tilde W_{cd}
=
\frac{\widehat N_d}{\widehat N_{(c)}},
\end{equation*}
denote the aggregated estimator, the weights summing to one over
$\mathcal D_{(c)}$.

Consider a library of $M$ candidate learners
$\hat g^1_{n,d},\dots,\hat g^M_{n,d}$, providing alternative working
prediction functions used in $\mu_d
=
E_{P_d^{\mathbf Z}}\{g_d(\mathbf Z)\}
+
E_{P_d}\{Y-g_d(\mathbf Z)\}$, each obtained by a different method, such as a linear mixed model, a random forest, or gradient
boosting, and write $\hat\mu^{(\cdot)}_{+(c)}(\hat g^l_{n,d})$ for the
corresponding aggregated estimator, where $(\cdot)$ denotes any of the estimators of Sections~4 and ~5 in the main document.

The selected index is
\begin{equation}
\label{eq:learner_selection}
\hat\ell_n
=
\arg\min_{\ell\in\{1,\ldots,M\}}
\frac{1}{C}
\sum_{c=1}^C
\mathcal L
\left\{
\widehat{\bar y}^{\,(\cdot)}_{(c)},
\hat\mu^{(\cdot)}_{+(c)}(\hat g^\ell_{n,d})
\right\},
\end{equation}
where the selected
learner is
$\hat g_n^{\mathrm{opt}}:=\hat g_n^{\hat\ell_n}$,
with area restriction
$\hat g_{n,d}^{\mathrm{opt}}:=\hat g_n^{\mathrm{opt}}(\cdot,d)$.
The superscript ``opt'' denotes minimisation of
\eqref{eq:learner_selection} among the $M$ candidate learners.
Here $\mathcal L$ is a chosen loss function and
$\widehat{\bar y}^{\,(\cdot)}_{(c)}$ is either of the direct estimators
in \eqref{eq:target_parameter_region}. The criterion is evaluated
separately for each estimator and fitting scheme under comparison.
In the application, we use absolute loss and weight the coarse
subpopulations equally; weighting them by $N_{(c)}/N$ instead would
reflect the population composition.

Direct estimators thus play two distinct roles in the paper. They serve as design-based benchmarks
against which model- and learning-based estimates are compared; in addition, they serve as validation targets for selecting among learners. The second role requires care. A direct estimator used for
validation is itself subject to sampling error, so that
the criterion in \eqref{eq:learner_selection} compares an estimate with an estimate rather than with a known quantity. It is also computed from the same sampled units as the small area learners it is used to assess, so that the two are statistically dependent and the comparison is not an out-of-sample assessment. Aggregation to coarse subpopulations mitigates the first difficulty, since $n_{(c)}\gg n_d$, but not the second.

The pointwise inferential results established above continue to apply after learner selection provided that the selected procedure itself satisfies Assumption~7 in the main document, or Assumption~9 under cross-fitting. This holds, for example, if the selection probability converges to one for a candidate satisfying the relevant conditions; otherwise, the additional randomness induced by learner selection is not covered by the preceding theory.

\section{Data-generating mechanism}
\label{app:generation}

This section describes how the populations and samples of Section~7 of the main document and Section \ref{sec:simulations_comp} are generated. The mechanism is designed to reproduce the main structural features of Sections~2--3 in the main document and to govern each source of between-area heterogeneity through a single interpretable parameter.

\subsection{Monte Carlo protocol}
\label{app:gen_protocol}

The study is conducted within the superpopulation framework of
Section~2.1 in the main document. Each Monte Carlo replication draws a fresh
finite population from the area laws $P_d$ and then a fresh sample from that
population, so that the finite-population targets $\bar y_d$ vary across
replications together with the realized area components. Estimators are
scored against the realized $\bar y_d$ of the replication in which they are
computed. The assumptions used for the within-area analysis
hold conditionally on the realized area-specific components, as discussed
after Assumption~2. The Monte Carlo summaries reported in
Section~7 of the main document and Section~\ref{sec:simulations_comp}, however, average over their repeated
generation unless explicitly stated otherwise.

\subsection{Covariates}
\label{app:gen_covariates}

Every area $d\in\mathcal D$ receives a vector of $q$ contextual covariates
$\mathbf u_d$, constant across the units of the area, so that the
$\mathbf U$-marginal of $P_d^{\mathbf Z}$ is degenerate as required in
Section~2.1 in the main document. With $q=2$ and
$\varepsilon_{d1},\varepsilon_{d2}\stackrel{\text{iid}}{\sim}N(0,1)$,
\begin{equation}
\label{eq:app_ctx}
u_{d1}=\psi\,\varepsilon_{d1},
\qquad
u_{d2}=\psi\left(0.5\,u_{d1}+\varepsilon_{d2}\right),
\end{equation}
where $\psi\ge 0$ governs the strength of covariate shift.

The unit-level covariates are generated with an area-specific location and an
area-specific dependence parameter, the two channels through which
$P_d^{\mathbf Z}$ differs across areas. For $k=1,\dots,p$,
\begin{equation}
\label{eq:app_mu}
\mu_{dk}=\psi\left(\nu_{dk}+0.3\,u_{d1}+0.2\,u_{d2}\right),
\qquad
\nu_{dk}\stackrel{\text{iid}}{\sim}N(0,0.4^2),
\qquad
\gamma_d=\tau\,\mathrm{Unif}(0.8,1.2),
\end{equation}
with $\tau\ge 0$ controlling the strength of the nonlinear dependence among
covariates. Given $(\boldsymbol\mu_d,\gamma_d)$, latent vectors
$\mathbf V_{dj}\sim N(\boldsymbol\mu_d,\mathbf I_p)$ are drawn independently
across $j$ and transformed recursively,
\begin{equation}
\label{eq:app_W}
\begin{aligned}
W_{dj1}&=V_{dj1}, &
W_{dj2}&=V_{dj2}+0.25\,\gamma_d W_{dj1}^{2}, \\
W_{dj3}&=V_{dj3}+0.40\,\gamma_d \sin(W_{dj1}), &
W_{dj4}&=V_{dj4}+0.20\,\gamma_d W_{dj1}W_{dj2}, \\
W_{dj5}&=V_{dj5}+0.30\,\gamma_d \cos(W_{dj2}), &
W_{djk}&=V_{djk}+0.15\,\gamma_d \sin(W_{dj1}W_{dj,k-1}),\quad k=6,\dots,p .
\end{aligned}
\end{equation}
The recursion makes $P_d^{\mathbf Z}$ non-Gaussian and dependent, with a
dependence structure that itself varies across areas; $\tau=0$ returns
independent Gaussian covariates. Conditionally on
$(\boldsymbol\mu_d,\gamma_d,\mathbf u_d)$, the population pairs of an area are
independent and identically distributed, as required by
Assumption~2. Throughout we set $p=10$ and $q=2$; only
$W_1,\dots,W_5$ and $u_1,u_2$ enter the regression function below, so that
$W_6,\dots,W_{10}$ act as noise covariates which the learners must screen out.

\subsection{Shared regression function}
\label{app:gen_m0}

The shared component of Definition~1 in the main document is one of
three functions of increasing complexity: ``simple'' in \eqref{eq:app_m0_simple}, ``moderate'' in \eqref{eq:app_m0_moderate} and ``high'' in \eqref{eq:app_m0_high}:
\begin{align}
\tilde m_0(\mathbf z)&=0.8w_1-0.5w_2+0.3u_1,
\label{eq:app_m0_simple}\\[2pt]
\tilde m_0(\mathbf z)&=\sin(w_1)+0.5w_2^{2}-0.4w_3+0.3w_1w_2+0.3u_1-0.2u_2,
\label{eq:app_m0_moderate}\\[2pt]
\tilde m_0(\mathbf z)&=1.2\sin(w_1w_2)+0.7\cos(w_3)
+\frac{0.4\,w_4^{2}}{1+|w_1|}+0.3\tanh(w_5u_1)+0.2\,u_2w_2 .
\label{eq:app_m0_high}
\end{align}
All three are smooth; none contains jumps or axis-aligned thresholds, so no
specification is constructed to favour a particular class of learner.

\subsection{Area deviations}
\label{app:gen_deviations}

Deviations act additively, i.e., $m_d(\mathbf z)=m_0(\mathbf z)+\delta_d(\mathbf z)=m(\mathbf z,d)=m(\mathbf x)$,
$\mathbf x=(\mathbf z^\top,d)^\top$, and are expanded in a basis $\mathbf g=(g_1,\dots,g_K)^{\top}$ common to all areas,
\begin{equation}
\label{eq:app_delta}
\delta_d(\mathbf z)=\sum_{k=1}^{K}b_{dk}\,g_k(\mathbf z),
\end{equation}
which is what makes Assumption~4 in the main document exactly enforceable
rather than approximate. Two forms are used. Under \emph{location}
deviations, $K=1$ and $g_1\equiv 1$, so that $\delta_d$ is constant in
$\mathbf z$ and the specification reduces to the random intercept of the
nested-error model of Section~3.2.1 in the main document. Under \emph{functional}
deviations, $K=3$ with
\begin{equation}
\label{eq:app_basis}
\mathbf g(\mathbf z)=
\begin{cases}
\left(1,\;w_1,\;u_1\right)^{\top}, & \tilde m_0 \text{ as in }\eqref{eq:app_m0_simple},\\
\left(1,\;\sin(w_1),\;w_2u_1\right)^{\top}, & \tilde m_0 \text{ as in }\eqref{eq:app_m0_moderate},\\
\left(1,\;\sin(w_1w_2),\;\cos(w_3+u_1)\right)^{\top}, & \tilde m_0 \text{ as in }\eqref{eq:app_m0_high},
\end{cases}
\end{equation}
so that $\delta_d$ varies with $\mathbf z$ and cannot be represented by an
area intercept.

Raw coefficients are drawn independently, $\tilde b_{dk}\sim N(0,s_k^{2})$,
and then subjected to the two requirements of
Assumption~4 in the main document. The identifiability normalisation
$\sum_{d}\omega_d\delta_d=0$ holds if and only if each coefficient is centred
with the area weights, the basis functions being linearly independent, so we
set
\begin{equation}
\label{eq:app_centering}
\check b_{dk}=\tilde b_{dk}-\sum_{e\in\mathcal D}\omega_e\tilde b_{ek},
\qquad
\omega_d=N_d/N .
\end{equation}
The complexity functional is the $L_2$ norm under the area's own covariate
law, $\rho^2(\delta_d)=\int\delta_d^2\,dP_d^{\mathbf Z}$, evaluated on the
realized population of the area; for \eqref{eq:app_delta} it is the quadratic
form $\check{\mathbf b}_d^{\top}\widehat{\mathbf G}_d\check{\mathbf b}_d$
with $\widehat{\mathbf G}_d=N_d^{-1}\sum_{j\in\mathcal Q_d}\mathbf g(\mathbf z_{dj})\mathbf g(\mathbf z_{dj})^{\top}$. The constraint from Assumption~4 in the main document is then imposed exactly by the single rescaling
\begin{equation}
\label{eq:app_budget}
b_{dk}=\kappa\,\check b_{dk},\;
\kappa=
\Big\{
(C_D/\varsigma^2)
\Big/
\textstyle\sum_{e\in\mathcal D}
\check{\mathbf b}_e^\top
\widehat{\mathbf G}_e\check{\mathbf b}_e
\Big\}^{1/2},
\;\text{so that}\;
\sum_{d\in\mathcal D}
\widehat\rho_d^2(\varsigma\delta_d)=C_D.
\end{equation}
Because $\kappa$ is common to all areas, the rescaling preserves the centring
\eqref{eq:app_centering}. Only the ratios $s_1:s_2:s_3$ affect the generated
population, the overall scale being fixed by \eqref{eq:app_budget}: the
choice $s=(0,1,1)$ removes the location component of the deviation entirely,
so that $\delta_d$ varies within an area, whereas $s=(1,0,0)$ reproduces a
pure location shift.

\begin{remark}[Calibration of the budget]
\label{rem:app_budget}
Let $C_D$ denote the total squared complexity of the
deviations after application of the outcome-scale factor $\varsigma$. Since
the effective error standard deviation is $\varsigma\sigma_e$, we define $r
=(C_D/D)/(C_D/D+\varsigma^2\sigma_e^2)$. Solving for $C_D$ gives $C_D
=(D\varsigma^2\sigma_e^2r)/(1-r)$
Thus, $C_D/D$ is the average mean-square magnitude of the
deviations on the scale of the outcome. The parameter $r$ measures relative
deviation-to-noise heterogeneity and has the usual variance-ratio
form of an intraclass correlation for location deviations. Because
both the deviations and the errors are multiplied by $\varsigma$, the value
of $r$ is unaffected by the outcome-scale transformation.
\end{remark}

\subsection{Outcomes}
\label{app:gen_outcomes}

Outcomes are generated as
\begin{equation}
\label{eq:app_outcome}
Y_{dj}=m_d(\mathbf Z_{dj})+e_{dj},
\quad
m_d(\mathbf z)=\beta_0+\varsigma\left\{\tilde m_0(\mathbf z)+\delta_d(\mathbf z)\right\},
\quad
e_{dj}=\varsigma\,\tilde e_{dj},
\end{equation}
where $\beta_0$ sets the level of the outcome and $\varsigma>0$ its scale.
Since $\tilde m_0$, $\delta_d$ and $\tilde e$ are multiplied by the same
factor, every signal-to-noise ratio, including the intraclass correlation
$r
=(C_D/D)/(C_D/D+\varsigma^2\sigma_e^2)$, is left exactly unchanged by $\varsigma$; only the
units of $Y$ change. The two parameters are used to place the area means in a
range appropriate to a positive target such as log income, and are reported
with the scenarios in Table~\ref{tab:app_scenarios}.

The errors $\tilde e_{dj}$ are mean-zero and sub-Gaussian with a parameter
common to all areas, so that the error law belongs to the shared component
$\boldsymbol\eta$ of Assumption~3 in the main document rather than to
$\delta_d$. Two laws are used, both standardised to standard deviation
$\sigma_e$: a Gaussian law, and a two-component scale mixture in which the
scale is inflated by a factor $2$ with probability $0.1$.

\subsection{Sampling design}
\label{app:gen_sampling}

Within each area, $n_d$ of the $N_d$ population units are selected by simple
random sampling without replacement, independently across areas. The design
is noninformative given the covariates and self-weighting within areas, with
$p_d=n_d/N_d$ and weights $w_{dj}=N_d/n_d$ constant within an area and free
to differ across areas, as required by Assumption~1. The
direct estimator therefore reduces to the unweighted
sample mean of area $d$.

\section{Simulation scenarios and design}
\label{app:design}

\subsection{Scenarios}
\label{app:design_scenarios}

Five scenarios are considered. The first three form a ladder from a correctly
specified nested-error model to a specification in which every channel the
framework admits is stressed simultaneously; the last two repeat the shape of
the third at extreme settings of the heterogeneity and the noise. Their
parameter values, in the notation of Section~\ref{app:generation}, are
collected in Table~\ref{tab:app_scenarios}.

\begin{table}[htbp]
\centering\small
\caption{Parameter values defining the five scenarios. Covariate dimensions
are $p=10$ and $q=2$ throughout.}
\label{tab:app_scenarios}
\begin{tabular}{lccccc}
\hline
& S1 & S2 & S3 & S4 & S5 \\
\hline
shared function $\tilde m_0$ & \eqref{eq:app_m0_simple} & \eqref{eq:app_m0_moderate}
    & \eqref{eq:app_m0_high} & \eqref{eq:app_m0_high}
    & \eqref{eq:app_m0_high} \\
deviation type               & location & functional & functional & functional & functional \\
coefficient ratios $s$       & ---      & $(0.4,0.3,0.3)$ & $(0,1,1)$ & $(0,1,1)$ & $(0,1,1)$ \\
covariate shift $\psi$       & $1$      & $1$   & $1.5$ & $1.5$ & $1.5$ \\
dependence $\tau$            & $0$      & $1$   & $1.5$ & $1.5$ & $1.5$ \\
heterogeneity $r$            & $0.20$   & $0.30$& $0.50$& $0.30$& $0.70$ \\
noise $\sigma_e$             & $1$      & $1$   & $1$   & $0.25$& $1$ \\
error law                    & Gaussian & Gaussian & mixture & mixture & mixture \\
level $\beta_0$              & $10$     & $10$  & $10$  & $10$  & $10$ \\
scale $\varsigma$            & $1$      & $1$   & $0.3$ & $1$   & $1$ \\
\hline
deviation size $\sqrt{C_D/D}$ & $0.50$ & $0.65$ & $0.30$ & $0.16$ & $1.53$ \\
\hline
\end{tabular}
\end{table}

\emph{Scenario S1} reproduces the conditional mean and
random-intercept structure of the nested-error model of
Section~3.2.1 in the main document: with $\tau=0$ the covariates are independent Gaussian
vectors, $\tilde m_0$ is linear, and the deviation is a location shift, so
that
\begin{equation}
\label{eq:app_s1}
Y_{dj}=\beta_0+0.8W_{dj1}-0.5W_{dj2}+0.3u_{d1}+a_d+e_{dj},
\qquad
\mathrm{sd}(a_d)=\sqrt{C_D/D}.
\end{equation}
The scenario is a reference point at which the classical
conditional mean specification is correctly specified and any flexible learner must pay a price for its generality.

\emph{Scenario S2} moves one step along each channel: $\tilde m_0$ is smooth
but nonlinear, the covariates are nonlinearly dependent, and the deviation
retains a location component alongside a functional one, so that the
nested-error model captures part of the area heterogeneity and misses the
remainder.

\emph{Scenario S3} is adversarial for the classical specification. The
regression function is strongly nonlinear, with interactions between
unit-level and contextual covariates; the areas occupy visibly different
regions of the covariate space, so that pooling requires genuine adjustment
and any area-wise fit must extrapolate; the deviation carries no location
component, so that a random intercept cannot represent it; the deviations are
as large as the noise; and the errors are non-Gaussian. Diagnostics on the
generated populations confirm the intended structure: the deviation varies
more within an area, with standard deviation $0.91$, than across area means,
with standard deviation $0.41$, which is precisely the configuration a random
intercept is unable to capture.

\emph{Scenarios S4 and S5} retain the shape of S3 and vary only the size of
the deviation relative to the noise, which is the quantity governing whether the residual correction of estimators in Definition~3 in the main document pays. S4 sets
$r=0.30$ with $\sigma_e=0.25$, so that both the deviation and the noise are
small in absolute terms; S5 sets $r=0.70$ with $\sigma_e=1$, so that both are
large. The two therefore bracket S3 from either side while holding the
regression function, the covariate shift and the error law fixed, and they
isolate the ratio $\beta_d^2/(\sigma_d^2/n_d)$ from the other channels of the
design.

The level $\beta_0$ and scale $\varsigma$ place the area means in a range
appropriate to a positive target such as log income and, as noted in
Section~\ref{app:gen_outcomes}, leave all signal-to-noise ratios unchanged;
$\varsigma=0.3$ in S3 compensates for the larger spread induced by the
strongly nonlinear regression function. Because $\varsigma$ and the deviation
size differ across scenarios by construction, absolute errors are not
comparable between them and all comparisons are made within a scenario.

\subsection{Design cells and asymptotic regimes}
\label{app:design_cells}

Each scenario is run on a $2\times 2$ factorial in the number of areas and
the area sample size, with the population size held fixed at $N_d=1000$: $D\in\{25,50\}$, $n_d\in\{10,20\}$,
%\begin{equation}
%\label{eq:app_cells}
%D\in\{25,50\},
%\qquad
%n_d\in\{10,20\},
%\end{equation}
giving twenty setups in total, the four cells being listed in
Table~\ref{tab:app_cells}. $500$ Monte Carlo replications are generated for each setup. The two asymptotic regimes of Section~4 are read off as margins of this grid: the behaviour as the area sample size grows is obtained by comparing $n_d=10$ with $n_d=20$ at fixed $D$, and the behaviour as the number of areas grows by comparing $D=25$ with $D=50$ at fixed $n_d$.
Holding $r$, rather than the total budget $C_D$, fixed across cells, as in Remark~\ref{rem:app_budget}, ensures that a change in $D$ alters the number of areas without changing the average amount of heterogeneity
carried by each area.

\begin{table}[htbp]
\centering
\caption{The four design cells, each run under all five scenarios.
Sampling fractions $n_d/N_d$ are one or two per cent.}
\label{tab:app_cells}
\begin{tabular}{lccccc}
\hline
cell & $D$ & $N_d$ & $n_d$ & $N$ & $n$ \\
\hline
$D25\_n10$ & 25 & 1000 & 10 & 25\,000 & 250 \\
$D25\_n20$ & 25 & 1000 & 20 & 25\,000 & 500 \\
$D50\_n10$ & 50 & 1000 & 10 & 50\,000 & 500 \\
$D50\_n20$ & 50 & 1000 & 20 & 50\,000 & 1000 \\
\hline
\end{tabular}
\end{table}

\subsection{Estimators, learners and fitting schemes}
\label{app:design_estimators}

Three fitting schemes are compared: the estimators of
Definitions~3 and~5 in the manuscript, which use no sample splitting; and the cross-fitted estimators of Definitions~4 and~6 in the main document with $L=5$ folds, partitioning the pooled sampled index set $\mathcal Q^s$ either within each area, so that every training set holds $n_d(L-1)/L$ units from every area, or without regard to area. Both schemes define computable estimators, but the stated cross-fitting propositions assume stratified folds. The pooled scheme is included as an empirical comparison. Under the pooled scheme the number of area-$d$ units in a
training set is random and the intersection
$\mathcal Q_d^s\cap\mathcal I_\ell$ may be empty for some
pairs $(d,\ell)$; this does not prevent computation of the estimator, the folds still partitioning $\mathcal Q_d^s$, so that the double sum in Definition~4 in the main document always contains exactly $n_d$ terms.

The nuisance function is learnt on the full covariate vector
$\mathbf x=(\mathbf z^{\top},d)^{\top}$, so that the area indicator is among
the learner's features and the learner may estimate $\delta_d$ itself. The
library consists of a linear model, a regression spline with one interior knot
per continuous covariate, boosted trees of depth two, a random forest of depth
three, BART with a regularising prior, a linear mixed model with a random area
intercept, M-quantile regression and a mixed-effects random forest. Under
cross-fitting it is extended by unrestricted gradient boosting, an
unrestricted random forest and unrestricted BART. Tuning parameters are fixed
in advance and are not selected on the data, so that the low-complexity
members of the library remain within a class of the kind of
Definition~2 in the main document.

\subsection{Evaluation}
\label{app:design_evaluation}

For each setup, $R$ replications are generated. Writing $\bar y_d^{(r)}$ for the realized finite-population mean of area $d$ in replication $r$ and $\widehat{\bar y}_d^{(r)}$ for an estimate of it, all summaries are computed area by area and then averaged over areas, so that between-area variation in $\bar y_d$ does not enter the bias and variance measures. With
$\mathrm{B}_d=R^{-1}\sum_r(\widehat{\bar y}_d^{(r)}-\bar y_d^{(r)})$ we report
\begin{equation}
\label{eq:app_metrics}
\begin{aligned}
\text{absolute bias} &= \frac{1}{D}\sum_{d}\left|\mathrm{B}_d\right|, &
\text{relative bias} &= \frac{1}{D}\sum_{d}
\left|\frac{\mathrm{B}_d}{R^{-1}\sum_r\bar y_d^{(r)}}\right|, \\[4pt]
\text{MSE} &= \frac{1}{RD}\sum_{r}\sum_{d}
\left(\widehat{\bar y}_d^{(r)}-\bar y_d^{(r)}\right)^{2}, &
\text{MC variance} &= \frac{1}{D}\sum_{d}
\widehat{\operatorname{Var}}_r\big(\widehat{\bar y}_d^{(r)}-\bar y_d^{(r)}\big),
\end{aligned}
\end{equation}
together with the empirical coverage of the intervals of
Corollary~4.3 in the main document, that is the proportion of pairs $(r,d)$ for
which $\bar y_d^{(r)}$ falls inside the interval. The Monte Carlo variance is
centred on the error rather than on the estimate, since each replication
regenerates the population and $\bar y_d$ therefore varies across
replications.

Intervals are computed both with the normal quantile  and with
the quantile of the $t_{n_d-1}$ distribution. The latter has the same
asymptotic justification, since $t_{n_d-1}\to N(0,1)$, and is reported as a
finite-sample refinement appropriate to $n_d\in\{10,20\}$; it is not exact,
because the residuals entering $\hat\sigma_d^2$ are computed from an estimated
$g^{*}_d$.

The consistency of the variance estimator asserted in Corollary~4.3 in the main document is assessed in two ways. The first compares
$\hat\sigma_d^2$ with $n_d$ times the realized sampling variance of the
estimator, the ratio tending to one. The second compares
it with the variance of the residual $Y-\hat g_{n,d}(\mathbf Z)$ over all
$N_d$ population units of the area, computed with the same fitted learner;
since $N_d\gg n_d$ this is essentially free of sampling error and isolates the
variability of $\hat\sigma_d^2$ from the question of what $g^*_d$ is.

\subsection{Complementary results} \label{sec:simulations_comp}

Tables~\ref{tab:top3_S1}, \ref{tab:top3_S3},
\ref{tab:top3_S4}, and~\ref{tab:top3_S5} report selected results for Scenarios~S1, S3, S4, and~S5. For each design cell, they show the three configurations with the smallest RMSE, the best fully corrected estimator when not already included, and the direct estimator. Scenario~S2 is reported in Section~7 of the main manuscript.

\begin{table}[htbp]
\centering\small
\caption{Scenario S1: the 3 combinations of learner, fitting scheme and estimator with the smallest RMSE for $\bar y_d$ in each design cell, followed by the best fully corrected member where none ranks among them, and by the direct estimator. Absolute bias and RMSE are multiplied by 100; the gap is the percentage excess in RMSE over the best row of the cell. Coverage is of nominal $95\%$ intervals with the $t_{n_d-1}$ quantile; for \texttt{sae\_naive} it refers to the interval centred at the uncorrected predictor $\hat\mu_d(f_d)$ with half-width $t_{n_d-1,\,0.975}f_d\hat\sigma_d n_d^{-1/2}$, which is not covered by Corollary~4.3 in the main document.} 
\label{tab:top3_S1}
\begin{tabular}{llllrrrr}
\toprule
cell & scheme & learner & estimator & abs.\ bias & RMSE & gap (\%) & coverage \\
\midrule
S1\_D25\_n10 & sal & lmm & sae\_naive & 1.05 & 28.33 & 0.0 & 0.020 \\
S1\_D25\_n10 & cf\_ & lmm & sae\_naive & 1.08 & 28.55 & 0.8 & 0.021 \\
S1\_D25\_n10 & cf\_p & lmm & sae\_naive & 1.08 & 28.60 & 0.9 & 0.022 \\ \cmidrule(lr){2-8}
S1\_D25\_n10 & sal & boosted\_trees & sal & 1.12 & 31.07 & 9.7 & 0.943 \\
S1\_D25\_n10 & --- & --- & direct & 1.48 & 42.88 & 51.4 & 0.952 \\ \midrule
S1\_D25\_n20 & sal & lmm & sae\_naive & 0.77 & 21.01 & 0.0 & 0.034 \\
S1\_D25\_n20 & cf\_ & lmm & sae\_naive & 0.79 & 21.17 & 0.8 & 0.036 \\
S1\_D25\_n20 & cf\_p & lmm & sae\_naive & 0.78 & 21.18 & 0.8 & 0.036 \\ \cmidrule(lr){2-8}
S1\_D25\_n20 & sal & mquantile & sal & 0.73 & 22.10 & 5.2 & 0.951 \\
S1\_D25\_n20 & --- & --- & direct & 0.68 & 30.24 & 44.0 & 0.951 \\ \midrule
S1\_D50\_n10 & sal & lmm & sae\_naive & 0.91 & 27.50 & 0.0 & 0.022 \\
S1\_D50\_n10 & cf\_ & lmm & sae\_naive & 0.91 & 27.64 & 0.5 & 0.022 \\
S1\_D50\_n10 & cf\_p & lmm & sae\_naive & 0.91 & 27.67 & 0.6 & 0.022 \\ \cmidrule(lr){2-8}
S1\_D50\_n10 & sal & mquantile & sal & 1.14 & 31.52 & 14.6 & 0.949 \\
S1\_D50\_n10 & --- & --- & direct & 1.73 & 43.09 & 56.7 & 0.952 \\ \midrule
S1\_D50\_n20 & sal & lmm & sae\_naive & 0.68 & 20.62 & 0.0 & 0.035 \\
S1\_D50\_n20 & cf\_ & lmm & sae\_naive & 0.66 & 20.73 & 0.5 & 0.037 \\
S1\_D50\_n20 & cf\_p & lmm & sae\_naive & 0.67 & 20.74 & 0.6 & 0.037 \\ \cmidrule(lr){2-8}
S1\_D50\_n20 & sal & mquantile & sal & 0.79 & 22.12 & 7.3 & 0.952 \\
S1\_D50\_n20 & --- & --- & direct & 1.02 & 30.57 & 48.3 & 0.952 \\
\bottomrule
\end{tabular}
\end{table}

\begin{table}[htbp]
\centering\small
\caption{Scenario S3: the 3 combinations of learner, fitting scheme and estimator with the smallest RMSE for $\bar y_d$ in each design cell, followed by the best fully corrected member where none ranks among them, and by the direct estimator. Absolute bias and RMSE are multiplied by 100; the gap is the percentage excess in RMSE over the best row of the cell. Coverage is of nominal $95\%$ intervals with the $t_{n_d-1}$ quantile; for \texttt{sae\_naive} it refers to the interval centred at the uncorrected predictor $\hat\mu_d(f_d)$ with half-width $t_{n_d-1,\,0.975}f_d\hat\sigma_d n_d^{-1/2}$, which is not covered by Corollary~4.3 in the main document.}
\label{tab:top3_S3}
\begin{tabular}{llllrrrr}
\toprule
cell & scheme & learner & estimator & abs.\ bias & RMSE & gap (\%) & coverage \\
\midrule
S3\_D25\_n10 & cf\_ & bart & sae\_naive & 0.99 & 16.30 & 0.0 & 0.023 \\
S3\_D25\_n10 & cf\_p & bart & sae\_naive & 1.00 & 16.34 & 0.3 & 0.023 \\
S3\_D25\_n10 & cf\_p & bart\_flexible & sae\_naive & 0.83 & 16.40 & 0.6 & 0.023 \\ \cmidrule(lr){2-8}
S3\_D25\_n10 & sal & spline & sal & 0.80 & 17.21 & 5.6 & 0.917 \\
S3\_D25\_n10 & --- & --- & direct & 0.90 & 26.59 & 63.1 & 0.948 \\ \midrule
S3\_D25\_n20 & cf\_ & bart\_flexible & sae\_naive & 0.70 & 11.25 & 0.0 & 0.040 \\
S3\_D25\_n20 & cf\_p & bart\_flexible & sae\_naive & 0.71 & 11.29 & 0.4 & 0.044 \\
S3\_D25\_n20 & sal & spline & sae\_naive & 0.42 & 11.59 & 3.0 & 0.031 \\ \cmidrule(lr){2-8}
S3\_D25\_n20 & sal & spline & sal & 0.42 & 11.59 & 3.0 & 0.934 \\
S3\_D25\_n20 & --- & --- & direct & 0.79 & 19.26 & 71.2 & 0.948 \\ \midrule
S3\_D50\_n10 & cf\_s& bart\_flexible & sae\_naive & 0.83 & 15.44 & 0.0 & 0.022 \\
S3\_D50\_n10 & cf\_p & bart\_flexible & sae\_naive & 0.82 & 15.57 & 0.8 & 0.023 \\
S3\_D50\_n10 & cf\_p & bart & sae\_naive & 0.88 & 16.01 & 3.7 & 0.024 \\ \cmidrule(lr){2-8}
S3\_D50\_n10 & sal & spline & sal & 0.60 & 16.60 & 7.5 & 0.935 \\
S3\_D50\_n10 & --- & --- & direct & 0.88 & 26.94 & 74.5 & 0.946 \\ \midrule
S3\_D50\_n20 & cf\_p & bart\_flexible & sae\_naive & 0.59 & 10.75 & 0.0 & 0.037 \\
S3\_D50\_n20 & cf\_s& bart\_flexible & sae\_naive & 0.61 & 10.76 & 0.1 & 0.037 \\
S3\_D50\_n20 & cf\_p & gradient\_boosting & sae\_naive & 0.64 & 11.40 & 6.1 & 0.043 \\ \cmidrule(lr){2-8}
S3\_D50\_n20 & sal & spline & sal & 0.39 & 11.43 & 6.3 & 0.943 \\
S3\_D50\_n20 & --- & --- & direct & 0.66 & 18.35 & 70.7 & 0.949 \\
\bottomrule
\end{tabular}
\end{table}

\begin{table}[htbp]
\centering\small
\caption{Scenario S4: the 3 combinations of learner, fitting scheme and estimator with the smallest RMSE for $\bar y_d$ in each design cell, followed by the best fully corrected member where none ranks among them, and by the direct estimator. Absolute bias and RMSE are multiplied by 100; the gap is the percentage excess in RMSE over the best row of the cell. Coverage is of nominal $95\%$ intervals with the $t_{n_d-1}$ quantile; for \texttt{sae\_naive} it refers to the interval centred at the uncorrected predictor $\hat\mu_d(f_d)$ with half-width $t_{n_d-1,\,0.975}f_d\hat\sigma_d n_d^{-1/2}$, which is not covered by Corollary~4.3 in the main document.}
\label{tab:top3_S4}
\begin{tabular}{llllrrrr}
\toprule
cell & scheme & learner & estimator & abs.\ bias & RMSE & gap (\%) & coverage \\
\midrule
S4\_D25\_n10 & cf\_s& spline & sae\_naive & 2.98 & 36.43 & 0.0 & 0.021 \\
S4\_D25\_n10 & sal & spline & sal & 2.68 & 36.45 & 0.1 & 0.912 \\
S4\_D25\_n10 & sal & spline & sal\_a\_bc & 2.68 & 36.45 & 0.1 & 0.912 \\ \cmidrule(lr){2-8}
S4\_D25\_n10 & --- & --- & direct & 2.67 & 77.62 & 113.1 & 0.943 \\ \midrule
S4\_D25\_n20 & cf\_s& bart\_flexible & sae\_naive & 1.96 & 22.14 & 0.0 & 0.042 \\
S4\_D25\_n20 & cf\_p & bart\_flexible & sae\_naive & 2.00 & 22.24 & 0.4 & 0.043 \\
S4\_D25\_n20 & sal & spline & sae\_naive & 1.11 & 24.66 & 11.4 & 0.032 \\ \cmidrule(lr){2-8}
S4\_D25\_n20 & sal & spline & sal & 1.11 & 24.66 & 11.4 & 0.930 \\
S4\_D25\_n20 & --- & --- & direct & 2.35 & 56.60 & 155.6 & 0.942 \\ \midrule
S4\_D50\_n10 & cf\_s& bart\_flexible & sae\_naive & 2.33 & 32.54 & 0.0 & 0.023 \\
S4\_D50\_n10 & cf\_p & bart\_flexible & sae\_naive & 2.36 & 32.59 & 0.1 & 0.022 \\
S4\_D50\_n10 & sal & spline & sae\_naive & 1.62 & 35.40 & 8.8 & 0.016 \\ \cmidrule(lr){2-8}
S4\_D50\_n10 & sal & spline & sal & 1.62 & 35.40 & 8.8 & 0.933 \\
S4\_D50\_n10 & --- & --- & direct & 2.55 & 77.14 & 137.0 & 0.940 \\ \midrule
S4\_D50\_n20 & cf\_s& bart\_flexible & sae\_naive & 1.51 & 20.22 & 0.0 & 0.043 \\
S4\_D50\_n20 & cf\_p & bart\_flexible & sae\_naive & 1.50 & 20.32 & 0.5 & 0.042 \\
S4\_D50\_n20 & cf\_p & spline & sae\_naive & 0.94 & 24.19 & 19.6 & 0.038 \\ \cmidrule(lr){2-8}
S4\_D50\_n20 & sal & spline & sal & 0.93 & 24.24 & 19.9 & 0.943 \\
S4\_D50\_n20 & --- & --- & direct & 1.82 & 54.26 & 168.3 & 0.943 \\
\bottomrule
\end{tabular}
\end{table}

\section{Data analysis -- SAL in action}
\label{sec:data_analysis}

\subsection{Description of the data}
The case study combines administrative data from the 2011 Italian Census with
survey data from the 2012 EU-SILC, taking 2011 as the reference year. The same
data were analysed in a causal small area setting by
\cite{reluga2022causalSAE}, where both sources are described in detail; here we
use them to illustrate SAL in practice and give only a summary, referring the
reader to that paper for the full description and definitions. The areas are the $D=41$ Italian provinces belonging to six regions: Lombardy
(11 provinces), Tuscany (10), Sicily (9), Campania (5), Marche (4) and
Umbria (2). These span northern, central, southern and insular Italy, and so
cover the country's pronounced north--south divide. The regions serve as the
coarse subpopulations used in the selection of learners as explained in Section~\ref{sec:choice_learner}. The
population comprises $N=10 323 770$ households with a head of working age
(15--64) resident in these provinces, of which $n=6528$ are sampled, a
sampling fraction of $f\approx0.0006$. Provincial sample sizes $n_d$ range from
$12$ to $676$, with mean $159$ and median $96$; the sampling fraction is
negligible in every province, which is the regime small area estimation
addresses. The survey weights are calibrated to the census population counts
of the regions.

The outcome $Y_{dj}$ is the equivalised household disposable income, on its
original scale. \cite{reluga2022causalSAE} model its logarithm, as is common
for income data under the nested-error model. Here the parameter of interest
is the provincial mean income $\bar y_d$ itself, and SAL estimates it directly:
the residual correction requires no distributional model, so the skewness of
income need not be removed by a transformation, and no back-transformation,
with its bias correction, is needed. The covariates are the sex, nationality,
age, marital status and education level of the household head, the household
size, a regional indicator and the average fiscal income of the province, the
last two being contextual covariates in the sense of
Section~2.1 in the main document.

\subsection{Estimation and choice of learner}
\label{sec:app_estimation}

As in Section~7 of the main document and Section~\ref{sec:simulations_comp}, we estimate the provincial means
$\bar y_d$ using the estimators of Definitions~2
and~5 in the main document, without sample splitting, and the cross-fitted estimators of Definitions~4 and~6 in the main document, with $L=5$ folds under both stratified and pooled fold construction. The direct Horvitz–Thompson estimator is included for comparison.
The nuisance function $g$, defined through
$\mu_d=E_{P_d^{\mathbf Z}}\{g_d(\mathbf Z)\}
+E_{P_d}\{Y-g_d(\mathbf Z)\}$, is learned
using the full covariate vector $\mathbf x=(\mathbf z^\top,d)^\top$,
pooling observations across provinces. For cross-fitting, each learner
is trained on the corresponding training folds.
We use the same learners as in Section~7 in the main document and Section~\ref{sec:simulations_comp}, excluding
fully flexible BART because of its prohibitive computational cost.
Learner selection follows the aggregation approach described in
Section~\ref{sec:learner-selection} and detailed in
Section~\ref{sec:choice_learner}. The six regions serve as the coarser
subpopulations, with regional Horvitz–Thompson estimates as validation benchmarks. Selection is based on
absolute deviations between these benchmarks and the aggregated
provincial estimates.

\subsection{Results}
\label{sec:app_results}
Table~\ref{tab:app_best_learner} reports, for each combination of estimator
and fitting scheme, the best learner with its mean and maximum distance to the regional target; Table~\ref{tab:app_best_learner_short} retains the three best combinations in terms of MSE. 

\emph{Choice of learner}.  Every procedure in Table~\ref{tab:app_best_learner}-~\ref{tab:app_best_learner_short} improves on the direct estimator, reducing the mean distance by $36\%$--$53\%$.
The closest procedure is the uncorrected predictor
$\hat\mu_d^{cf}(f_d)$ in Example~2 in the main document, using a cross-fitted
random forest, with a mean distance of $364$, approximately $2\%$
of the mean regional income. The best fully corrected estimator,
$\hat\mu_d^{cf}(1)$ in Example~1 in the main document, uses cross-fitted gradient
boosting with pooled folds and follows at $378$, only $4\%$ higher.
This small gap between the best procedures at the two endpoints of
the family is consistent with the simulation results in
Section~7 in the main document and Section~\ref{sec:simulations_comp}. The leading distances are comparable
to the reported sampling errors of the regional benchmarks, ranging
from approximately $361$ in Tuscany to $792$ in Campania. Together
with the small gaps between the leading procedures, this supports
interpreting the selection as identifying a group of competitive
candidates.

The comparison between $\hat\mu_d^{cf}(f_d)$ in Example~2 in the main document and $\hat\mu_d^{cf}(1)$ in Example~1 in the main document is consistent with the
shrinkage interpretation of $\hat\mu_d(a)$ discussed below
Proposition~4.2 in the main document. For the estimators using cross-fitted random forest (rows 1-2 in Table~\ref{tab:app_best_learner}), using $\hat\mu_d^{cf}(f_d)$ in Example~2 in the main document without the residual correction (line 1, ``sae'', ``cf'' in Table~\ref{tab:app_best_learner}) brings the estimates closer to the target ($364$ against $422$ for $\hat\mu_d^{cf}(1)$ in Example~1 in the main document, line 2 ``sae\_cor'', ``cf'' in Table~\ref{tab:app_best_learner}), indicating that the learner captures most of the provincial departures, so that $\beta_d^2$ measuring the area-level approximation bias is small relative to variance $\sigma_d^2/n_d$. For mixed-effects random forests fitted without sample splitting
in rows 11--12 of Table~\ref{tab:app_best_learner}, the ordering
reverses: residual correction reduces the mean distance from
approximately $432$ for $\hat\mu_d(f_d)$ defined 
in Example~2 in the main document (row 11, ``sae'', ``sal'') to $413$ for
$\hat\mu_d(1)$ in Example~1 in the main document
(row 12, ``sae\_cor'', ``sal''), suggesting that the residual correction removes a bias that outweighs its
variance.

These results support selecting the learner and estimator jointly and
data-adaptively: no single combination performs best across the settings
considered. The first row of Table~\ref{tab:app_best_learner} identifies
the uncorrected estimator using a cross-fitted random forest as the
procedure with the smallest mean distance in this application.
The best procedure without sample splitting, which uses a corrected
mixed-effects random forest, has a mean distance of approximately $413$,
about $13\%$ higher than the best cross-fitted procedure ($364$) and
$9\%$ higher than the best fully corrected cross-fitted estimator ($378$).
This ranking differs from the results in Section~7 of the main document and Section~\ref{sec:simulations_comp},
where estimators using a low-complexity spline learner without sample splitting led the rankings
(Tables~\ref{tab:top3_S1}--\ref{tab:top3_S5}). Finally, the EBLUP under the nested-error model
(the penultimate row of Table~\ref{tab:app_best_learner},
``lmm'', ``sae'', ``sal'') reduces the direct estimator's mean absolute
distance by $27\%$ (last column). Its mean absolute distance nevertheless
remains approximately $49\%$ higher than that of the best fully corrected
estimator ($562$ versus $378$), nowhere near the best performing methods.

\emph{Provincial estimates and intervals}.
Figure~\ref{fig:best_ht} shows the best performing procedures in terms of the closest mean distance to the regional-level Horvitz-Thompson estimator, together with the direct
estimates. The learner-based estimates depart most from the direct ones in
the provinces with the smallest samples, on the left of the figure, where the
direct estimates are least stable. 

Table~\ref{tab:app_ranking} reports the mean length of the $95\%$ intervals over the 41 provinces. The intervals of Corollary~4.3 in the main document are only $9$--$14\%$ shorter on average than the design-based intervals of the direct estimator ($4895$ and $5158$ against $5695$), and in
Figure~\ref{fig:best_ht} the two are hard to tell apart. This is what the
theory predicts. The length of the SAL interval is governed by the residual
variance $\sigma_d^2$ of Section~4.2 of the main document, and is shorter than the direct interval by a factor of about
$\{1-R^2_{d}\}^{1/2}$, where $R^2_{d}$ is the share of the within-province
variance of income explained by the covariates. The observed ratios
correspond to $R^2_{d}$ of about $0.2$--$0.25$, which is typical of household
income; in the simulation study the covariates explained considerably more,
and the gains were correspondingly larger. SAL borrows strength across
provinces for the level of $\bar y_d$, through the fitted $\hat g$, and this
is visible in the gains in Table~\ref{tab:app_best_learner}; the width of its
interval, by contrast, reflects only how much the covariates explain within
each province.

The EBLUP intervals, based on the parametric bootstrap MSE, are much
shorter: $3001$ on average, $47\%$ shorter than the direct intervals and
$38\%$ shorter than the SAL intervals obtained with the same linear mixed
model ($4824$). This is the shrinkage effect. Following the discussion in Section~3.2.1 in the main document and using the same notation therein, the EBLUP shrinks the provincial residual
mean by the factor
$\Gamma_d^{ner}
=
\sigma_\gamma^2/
(\sigma_\gamma^2+\sigma_\varepsilon^2/n_d)$, the model-based counterpart of the weight $a_d^\ast$ of
Section~\ref{sec:hierarchical_w}, and its MSE treats the remaining bias
$(1-\Gamma_d^{ner})\gamma_d$ as averaging to zero over the random effects $\gamma_d$. Its length
relative to the SAL interval is therefore about $\sqrt{\Gamma_d^{ner}}$; the observed
ratio of $0.62$ corresponds to $\Gamma_d^{ner}$ of about $0.4$. Nevertheless, lower MSE and shorter confidence intervals come at the price. In particular, the shorter EBLUP interval has \emph{nominal coverage only on average over provinces}, and \emph{only if the normal nested-error model holds}. It undercovers provinces with large random effects, and the model is doubtful for income on its original scale, which is skewed and heteroscedastic. The SAL interval is valid province by province under the conditions of Corollary~4.3 in the main document, without a model for $Y$ given $\mathbf z$. The difference in length is the price of that guarantee.

\begin{figure}[ht]
  \centering
  \includegraphics[width=0.8\textwidth]{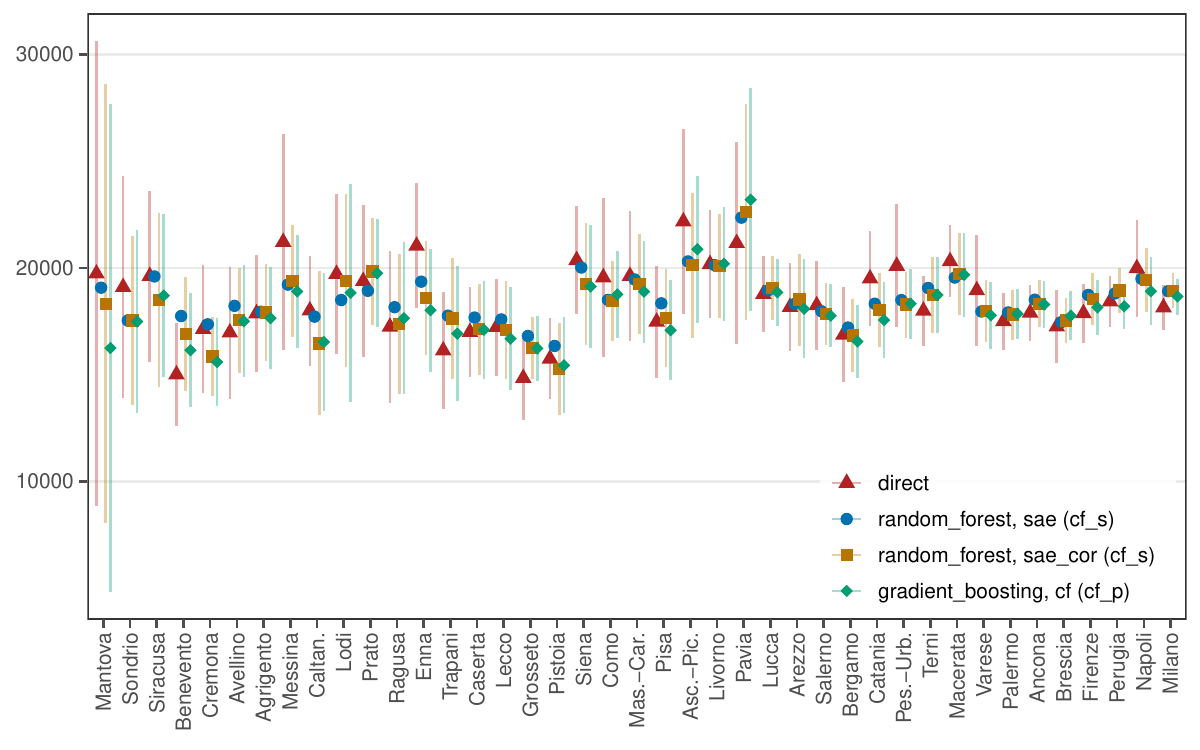}
  \caption{Estimated mean equivalised household disposable income by
  province, ordered by sample size. Triangles: direct Horvitz–Thompson estimate with design-based $95\%$ intervals. Circles: the procedure closest to the regional target, the uncorrected predictor $\hat\mu_d(f_d)$ with a cross-fitted random forest (no interval). Squares: its bias-corrected version, the fully corrected cross-fitted random forest, with the intervals of Corollary~4.3 in the main document. Diamonds: the fully corrected cross-fitted gradient boosting learner with pooled folds, with  the same intervals. Crosses: the EBLUP under the nested-error model, with intervals from its parametric bootstrap MSE.   \label{fig:best_ht}}
  % CHECK: series and symbols against plot_set(sel) and the legend
\end{figure}

\begin{table}[htbp]
\centering\small
\caption{Best learner within each combination of estimator and fitting scheme, by the mean absolute distance \eqref{eq:learner_selection} to the regional Horvitz--Thompson estimator over the 6 regions, with the maximum distance. Interval: whether the procedure supplies the interval of Corollary~4.3 in the main document. Rel.: mean distance as a percentage of the mean regional estimate. vs direct: percentage reduction in mean distance relative to the direct estimator. \texttt{sae\_cor} is the bias-corrected version of the model-based predictor \texttt{sae}, which coincides with the fully corrected member of the same learner and fitting scheme. Below the rule: \texttt{lmm}, \texttt{sae} is the EBLUP under the nested-error model; $^{*}$ its interval uses the parametric bootstrap MSE under that model.}
\label{tab:app_best_learner}
\begin{tabular}{lllrrcrr}
\toprule
learner & estimator & scheme & mean & max & interval & rel.\ (\%) & vs d. (\%) \\
\midrule
random\_forest & sae & cf\_s & 363.805 & 578.517 & no & 1.95 & 52.9 \\
random\_forest & sae\_cor & cf\_s & 422.121 & 694.576 & yes & 2.26 & 45.3 \\
random\_forest & sae & cf\_p & 375.677 & 599.741 & no & 2.01 & 51.3 \\
random\_forest & sae\_cor & cf\_p & 400.450 & 694.210 & yes & 2.15 & 48.1 \\
gradient\_boosting & cf\_a\_bc & cf\_p & 377.511 & 921.081 & yes & 2.02 & 51.1 \\
gradient\_boosting & cf & cf\_p & 377.511 & 921.081 & yes & 2.02 & 51.1 \\
merf & sal\_a\_bc & sal & 412.718 & 892.850 & yes & 2.21 & 46.5 \\
merf & sal & sal & 412.718 & 892.850 & yes & 2.21 & 46.5 \\
random\_forest & cf\_a\_bc & cf\_s & 422.121 & 694.577 & yes & 2.26 & 45.3 \\
random\_forest & cf & cf\_s & 422.121 & 694.577 & yes & 2.26 & 45.3 \\
merf & sae & sal & 431.790 & 750.658 & no & 2.32 & 44.1 \\
merf & sae\_cor & sal & 412.718 & 892.850 & yes & 2.21 & 46.5 \\
bart & cf\_a & cf\_s & 434.600 & 1273.690 & yes & 2.33 & 43.7 \\
mquantile & sal\_a & sal & 481.992 & 802.532 & yes & 2.58 & 37.5 \\
mquantile & cf\_a & cf\_p & 493.363 & 820.341 & yes & 2.65 & 36.1 \\
\midrule
lmm & sae & sal & 562.288 & 1031.664 & yes$^{*}$ & 3.01 & 27.1 \\
--- & direct & --- & 771.780 & 1267.491 & yes & 4.14 & 0.0 \\
\bottomrule
\end{tabular}
\end{table}

\begin{table}[htbp]
\centering\small
\caption{The 3 best combinations of estimator and fitting scheme in terms of MSE, each with its best learner, by the mean absolute distance \eqref{eq:learner_selection} to the regional Horvitz--Thompson estimator over the 6 regions, with the maximum distance. Interval: whether the procedure supplies the interval of Corollary~4.3 in the main document. Rel.: mean distance as a percentage of the mean regional estimate. vs direct: percentage reduction in mean distance relative to the direct estimator. \texttt{sae\_cor} is the bias-corrected version of the model-based predictor \texttt{sae}, which coincides with the fully corrected member of the same learner and fitting scheme. Below the rule: \texttt{lmm}, \texttt{sae} is the EBLUP under the nested-error model; $^{*}$its interval uses the parametric bootstrap MSE under that model.}
\label{tab:app_best_learner_short}
\begin{tabular}{lllrrcrr}
\toprule
learner & estimator & scheme & mean & max & interval & rel.\ (\%) & vs d. (\%) \\
\midrule
random\_forest & sae & cf\_s & 363.805 & 578.516 & no & 1.95 & 52.9 \\
random\_forest & sae\_cor & cf\_s & 422.121 & 694.576 & yes & 2.26 & 45.3 \\
random\_forest & sae & cf\_p & 375.677 & 599.741 & no & 2.01 & 51.3 \\
gradient\_boosting & cf\_a\_bc & cf\_p & 377.511 & 921.081 & yes & 2.02 & 51.1 \\
\midrule
lmm & sae & sal & 562.288 & 1031.664 & yes$^{*}$ & 3.01 & 27.1 \\
--- & direct & --- & 771.780 & 1267.491 & yes & 4.14 & 0.0 \\
\bottomrule
\end{tabular}
\end{table}

\begin{table}[htbp]
\centering\small
\caption{The 5 combinations of learner, estimator and fitting scheme closest to the regional Horvitz--Thompson estimator, ranked by the mean absolute distance \eqref{eq:learner_selection} over the 6 regions, with the mean of the 41 provincial estimates and the mean length of the $95\%$ intervals. The uncorrected model-based predictor \texttt{sae} has no interval of Corollary~4.3 in the main document; \texttt{sae\_cor} is its bias-corrected version, which coincides with the fully corrected member of the same learner and fitting scheme. Below the rule, for reference: \texttt{lmm}, \texttt{sae}, the EBLUP under the nested-error model ($^{*}$interval from the parametric bootstrap MSE under that model), and the direct estimator.}
\label{tab:app_ranking}
\begin{tabular}{rlllrr}
\toprule
rank & learner & estimator & scheme & mean estimate & mean CI length \\
\midrule
1 & random\_forest & sae & cf\_s & 18540.459 & --- \\
--- & random\_forest & sae\_cor & cf\_s & 18265.395 & 4895.149 \\
2 & gradient\_boosting & sae & cf\_s & 18247.116 & --- \\
3 & random\_forest & sae & cf\_p & 18538.196 & --- \\
4 & gradient\_boosting & cf & cf\_p & 18060.852 & 5158.087 \\
5 & gradient\_boosting & cf\_a\_bc & cf\_p & 18060.852 & 5158.087 \\
\midrule
--- & lmm & sae & sal & 18054.788 & 3000.884$^{*}$ \\
--- & --- & direct & --- & 18518.758 & 5694.748 \\
\bottomrule
\end{tabular}
\end{table}

\begin{table}[htbp]
\centering\small
\caption{Scenario S5: the 3 combinations of learner, fitting scheme and estimator with the smallest RMSE for $\bar y_d$ in each design cell, followed by the best fully corrected member where none ranks among them, and by the direct estimator. Absolute bias and RMSE are multiplied by 100; the gap is the percentage excess in RMSE over the best row of the cell. Coverage is of nominal $95\%$ intervals with the $t_{n_d-1}$ quantile; for \texttt{sae\_naive} it refers to the interval centred at the uncorrected predictor $\hat\mu_d(f_d)$ with half-width $t_{n_d-1,\,0.975}f_d\hat\sigma_d n_d^{-1/2}$, which is not covered by Corollary~4.3 in the main document.}
\label{tab:top3_S5}
\begin{tabular}{llllrrrr}
\toprule
cell & scheme & learner & estimator & abs.\ bias & RMSE & gap (\%) & coverage \\
\midrule
S5\_D25\_n10 & cf\_p & bart\_flexible & sae\_naive & 2.70 & 62.16 & 0.0 & 0.020 \\
S5\_D25\_n10 & cf\_s & bart\_flexible & sae\_naive & 2.69 & 62.23 & 0.1 & 0.023 \\
S5\_D25\_n10 & cf\_s & bart & sae\_naive & 3.37 & 63.86 & 2.7 & 0.022 \\ \cmidrule(lr){2-8}
S5\_D25\_n10 & sal & spline & sal & 2.75 & 66.75 & 7.4 & 0.915 \\
S5\_D25\_n10 & --- & --- & direct & 2.97 & 93.71 & 50.8 & 0.949 \\ \midrule
S5\_D25\_n20 & cf\_p & bart\_flexible & sae\_naive & 2.56 & 43.08 & 0.0 & 0.039 \\
S5\_D25\_n20 & cf\_s & bart\_flexible & sae\_naive & 2.61 & 43.09 & 0.0 & 0.042 \\
S5\_D25\_n20 & sal & spline & sae\_naive & 1.78 & 45.00 & 4.4 & 0.034 \\ \cmidrule(lr){2-8}
S5\_D25\_n20 & sal & spline & sal & 1.78 & 45.00 & 4.4 & 0.932 \\
S5\_D25\_n20 & --- & --- & direct & 2.75 & 67.79 & 57.4 & 0.947 \\ \midrule
S5\_D50\_n10 & cf\_s & bart\_flexible & sae\_naive & 2.98 & 59.11 & 0.0 & 0.023 \\
S5\_D50\_n10 & cf\_p & bart\_flexible & sae\_naive & 2.97 & 59.11 & 0.0 & 0.023 \\
S5\_D50\_n10 & cf\_s & gradient\_boosting & sae\_naive & 3.44 & 61.75 & 4.5 & 0.024 \\ \cmidrule(lr){2-8}
S5\_D50\_n10 & sal & spline & sal & 2.54 & 64.65 & 9.4 & 0.935 \\
S5\_D50\_n10 & --- & --- & direct & 3.07 & 93.42 & 58.1 & 0.947 \\ \midrule
S5\_D50\_n20 & cf\_s & bart\_flexible & sae\_naive & 2.10 & 41.74 & 0.0 & 0.037 \\
S5\_D50\_n20 & cf\_p & bart\_flexible & sae\_naive & 2.06 & 41.76 & 0.0 & 0.036 \\
S5\_D50\_n20 & cf\_p & gradient\_boosting & sae\_naive & 2.35 & 43.46 & 4.1 & 0.040 \\ \cmidrule(lr){2-8}
S5\_D50\_n20 & sal & spline & sal & 1.41 & 44.42 & 6.4 & 0.943 \\
S5\_D50\_n20 & --- & --- & direct & 2.24 & 65.56 & 57.1 & 0.950 \\
\bottomrule
\end{tabular}
\end{table}

%%%%%%%%%%%%%%%%%%%%%%%%%%%%%%%%%%%%%%%%%%%%%%
%% Funding information, if any,             %%
%% should be provided in the                %%
%% funding section.                         %%
%%%%%%%%%%%%%%%%%%%%%%%%%%%%%%%%%%%%%%%%%%%%%%
\begin{funding}
The first author was supported by the Swiss National Science Foundation through the project P2GEP2-195898 and by resources of the Center for Targeted Machine Learning, University of California at Berkeley and the HPC at the Humboldt University of Berlin. The second author was supported by the `MAPPE’ project, Programma `PE GRINS -- Growing Resilinet, Inclusive and Sustainable’ (cod. PE0000018 CUP: J33C22002910001)
\end{funding}

%% if your bibliography is in bibtex format, uncomment commands:
\bibliographystyle{imsart-number} % Style BST file (imsart-number.bst or imsart-nameyear.bst)
\bibliography{biblio}       % Bibliography file (usually '*.bib')

\end{document}